\documentclass[article,english,11pt]{article}

\usepackage{amscd,amsfonts,amsmath,amssymb,amsthm,bbm,bm,latexsym,mathrsfs}
\usepackage{epsfig,graphics,graphicx,natbib,subfigure}
\usepackage[english]{babel}
\usepackage[usenames]{color}
\makeatother

\newtheorem{theorem}{Theorem}[section]
\newtheorem{lemma}{Lemma}[section]

\newtheorem{example}{Example}[section]

\theoremstyle{remark}

\newcommand{\real}{\ensuremath{\mathbb{R}}}

\newcommand{\bmu}{\boldsymbol{\mu}}
\newcommand{\bnu}{\boldsymbol{\nu}}
\newcommand{\bomega}{\boldsymbol{\omega}}
\newcommand{\btheta}{\boldsymbol{\theta}}
\newcommand{\balpha}{\boldsymbol{\alpha}}
\newcommand{\bbeta}{\boldsymbol{\beta}}

\newcommand{\bp}{\boldsymbol{p}}
\newcommand{\bw}{\boldsymbol{w}}

\newcommand{\cD}{\mathcal{D}}
\newcommand{\cN}{\mathcal{N}}
\newcommand{\cT}{\mathcal{T}}
\newcommand{\cK}{\mathcal{S}\mathcal{N}}

\newcommand{\DK}{\Delta_K}
\newcommand{\DM}{\Delta_M}

\newcommand{\Beta}{\mathcal{B}e}
\newcommand{\Di}{\mathcal{D}ir}
\newcommand{\Ga}{\mathcal{G}a}
\newcommand{\CY}{\mathcal{Y}}
\newcommand{\one}{\mathbbm{1}}

\newcommand{\CG}{\mathfrak{F}}

\newcommand{\hsp}{\hspace{0.2mm}}

\begin{document}

\title{\vspace{-50pt}Bayesian Nonparametric Calibration and \\
       Combination of Predictive 
       Distributions  
%%%%BLIND
\thanks{This Working Paper should not be reported as representing the views of
Norges Bank. The views expressed are those of the authors and do not
necessarily reflect those of Norges Bank. We are much indebted to
Tilmann Gneiting for helpful discussions and for his contribution
to an earlier version of this work. We thank Concepcion Ausin, Luc Bawuens, Guido Consonni, Sylvia Fr\"uwirth-Schnatter, James Mitchell, Jacek Osiewalski, Dimitris Korobilis, Gary Koop, Enrique ter Horst, Shaun Vahey, Herman van Dijk, Ken Wallis, Mike West and Michael Wipper for their constructive comments, and the conference and seminar participants at: the CORE Louvain University, Cracow Polish Science Academy, Universidad Carlos III de Madrid, University Ca' Foscari of Venice, Scuola Normale Superiore of Pisa, Wien University seminar series, the European Centre for Living Technology of Venice, the 3rd Meeing on Statistic at  Athens University, the CAMP workshop on ``Uncertainty in Economics", at BI Norwegian Business School, the ``11th World Congress of the Econometric Society", Montreal, the 2014 ``Italian Statistical Society" meeting at University of Cagliari, the 2014 ``Econometric Modelling and Forecasting in Central Banks" workshop at University of Glasgow, the 2013 ``Economic Modelling and Forecasting Group" workshop at University of Warwick, the ``8th International Conference on Computational and Financial Econometrics", Pisa. Casarin's research is supported by funding from the European Union, Seventh Framework Programme FP7/2007-2013 under grant agreement
SYRTO-SSH-2012-320270, by the Institut Europlace of Finance, ``Systemic
Risk grant", the Global Risk Institute in Financial Services, the
Louis Bachelier Institute, ``Systemic Risk Research Initiative", and by
the Italian Ministry of Education, University and Research (MIUR) PRIN
2010-11 grant MISURA. Bassetti's research is supported by the INdAm-GNAMPA Project 2014. This research used the SCSCF multiprocessor cluster system at University Ca' Foscari of Venice.}
}
%%%%BLIND
\author{Federico Bassetti\setcounter{footnote}{3}\footnotemark{} \hspace{18pt}
Roberto Casarin\setcounter{footnote}{1}\footnotemark{} \hspace{18pt}
        Francesco Ravazzolo\setcounter{footnote}{2}\footnotemark{}
        \\
        \vspace{5pt}
        \\
        {\centering {\small \setcounter{footnote}{3}\footnotemark{}
             University of Pavia}} \\
        {\centering {\small \setcounter{footnote}{1}\footnotemark{}
             University of Venice}} \\
        {\centering {\small \setcounter{footnote}{2}\footnotemark{}
             Norges Bank and BI Norwegian Business School}}
       }

%\date{\today}

\maketitle

\begin{abstract}
We introduce a Bayesian approach to predictive density calibration and combination that accounts for parameter uncertainty and model set incompleteness through the use of random calibration functionals and random combination weights.  Building on the work of \cite{RanGne10} and \cite{GneRan13}, we use infinite beta mixtures for the calibration.  The proposed Bayesian nonparametric approach takes advantage of
the flexibility of Dirichlet process mixtures to achieve any continuous
deformation of linearly combined predictive distributions.  The inference procedure is based on Gibbs sampling and allows accounting for uncertainty in the number of mixture components, mixture weights, and calibration parameters.  The weak posterior consistency of the Bayesian nonparametric calibration is provided under suitable conditions for unknown true density.  We study the methodology in simulation examples with fat tails and multimodal densities and apply it to density forecasts of daily S\&P returns and daily maximum wind speed at the Frankfurt airport.
\end{abstract}




%\doublespacing

\section{Introduction}

Combining forecasts from different statistical models or other sources
of information is a crucial problem in many important applications.  A
wealth of papers have addressed this issue with \cite{BatGra69} being
one of the first attempts in this field.  The initial focus of the
literature was on defining and estimating combination weights for
point forecasts.  For instance, \cite{GraRam84} propose to combine
point forecasts with unrestricted least squares regression coefficients as
weights.  The ubiquitous Bayesian model averaging technique relies on
weighted averages of posterior distributions from different
models and implies linearly combined posterior means
\citep{HoeMadRafVol99}.  Recently, probabilistic forecasts in the form
of predictive probability distributions have become  prevalent in various fields, including macro economics with routine publications of fancharts from central banks, finance with asset allocation strategies based on higher-order moments, and meteorology with operational ensemble forecasts of future weather \citep{TayWallis2000, GneKa14}.

Therefore, research interest has shifted to the construction of
combinations of predictive distributions, which poses new challenges.  
A prominent, critically important issue is that predictive distributions ought to 
satisfy some properties. In this paper, we shall focus on the calibration, or reliability, of the predictive distributions \citep{Daw82, Dawid84, KlingBessler1989, DieGunTay98, GneBalRaf07, MitchellWallis2010}, and build on the widely used beta-calibration approach introduced by \cite{RanGne10} and \cite{GneRan13}. Our approach relies on the application of a distortion function to a given distribution and develops Bayesian estimation of the distortion function. See \cite{ArtznerDelbaenEberHeath99}, \cite{WangYoung98} and \cite{GzylMayoral08} for an application of probability distorsion to financial risk measurement and \cite{Cas15} for an application to joint calibration of predictive densities. It should be understood that the scope of our discussion is much wider and that our Bayesian inference approach can be extended to other calibration schemes, such as the conditional calibration developed in \cite{Fre96}, \cite{Wes92a}, \cite{Wes92b}, which is based on a Bayesian conditioning argument. Moreover, we will focus on the calibration of linear pool of predictive densities. Note that the traditional linear pool \citep{Stone61, HallMitchell2007} has been generalized to nonlinear aggregation schemes \citep{MitKap13, GneRan13}, and to time-varying approaches which account for time instabilities and estimation uncertainty in the combination weights \citep{BilCasRavVan13}. Our contribution can be extended to these models.

In this paper, we propose a flexible Bayesian nonparametric approach to calibration and combination that relies on beta mixtures, and nests the beta transformed linear pool introduced by \cite{RanGne10} and \cite{GneRan13}. We develop tools for Bayesian inference for both cases of known and unkown number of mixture components. In the case the number of components is not known we assume an infinite mixture representation and a Dirichlet process prior \citep{Ferguson73, Lo1984, Set94}.  This type of prior, and its  multivariate extensions (e.g., see \cite{Muller2004}, \cite{GriffinSteel2006}, \cite{Hatjispyros2011}), is now widely used due to the availability of efficient algorithms for posterior computations \citep{EscobarWest1995, MacMul98, PapRob08, Taddy2010}, including but not limited to applications in time series analysis \citep{Hir02, ChibHam02, RodriguezTerHorst2008, TaddyKottas2009, JensenMaheu2010,
Griffin2010, GriffinSteel2011, Bur14, BasCasLei14, Wie14, Joc15}. A recent account of Bayesian non-parametric inference can be found in \cite{BNP2010}.
In this paper we develop a slice sampling approach that builds on the work of \cite{walker2007} and \cite{walker2011}.

We provide some conditions under which the proposed probabilistic calibration converges in terms of weak posterior consistency to the true underlying density as the number of observations goes to infinity.  This calibration property is a very powerful result, which substantially improves upon the earlier approach of \cite{GneRan13}, which was shown to be flexibly dispersive only in the sense of second moment of the probabilistic forecast. We build on \cite{WuGhosal2008, WuGhosal2008correction} for the case of i.i.d. observations and on \cite{TangGhosal07} for the Markovian case. In this sense, we also contribute to the recent literature on posterior consistency of Bayesian nonparametric inference (e.g., see \cite{Pel14} and \cite{NorPel15}.

The remainder of the paper is organized as follows.  Section
\ref{sec:model} introduces our beta mixture calibration and
combination model and places it in the context of the general density
combination approach introduced by \cite{MitKap13}.  This is followed
by Section \ref{sec:inference}, where we propose Bayesian inference
based on slice and Gibbs sampling methods. Section
\ref{sec:consistency} provides posterior consistency of the
Bayesian nonparametric calibration and combination in the weak sense
under suitable conditions for unknown true density and under the assumption
of incomplete model set. In Section \ref{sec:simulation} we illustrate the effectiveness of our approach on simulation examples.  Section \ref{sec:empirical} provides case studies including some well-studied datasets in weather forecast and finance and see major improvements in the predictive performance for daily stock returns and
daily maximum wind speed.  The paper closes with a discussion in
Section \ref{sec:discussion}.

\section{Beta mixture calibration and combination}
\label{sec:model}

Let $F_1, \ldots, F_M$ be a set of predictive cumulative distribution
functions (cdfs) for a real-valued variable of interest, $y$, which
might be based on distinct statistical models or experts. Following the forecast combination and calibration literature, we assume the cdfs are externally provided. Also, following \cite{RanGne10} and \cite{GneRan13}, we consider combination formulas $\{ \CG_{\btheta} : \btheta \in \Theta \}$  that map the $M$-tuple $(F_1, \ldots, F_M)$ into a single, aggregated
predictive cdf, $F(\cdot | \btheta)=\CG_\theta(\cdot | F_1,\dots,F_M)$.  Given a 
sequence of observations, $y_{1},\ldots,y_{T}$, the cdf evaluated on one observation, 
e.g. $F(y_{t}|\btheta)$, is referred as probability integral transform (PIT).  Following \cite{RanGne10} and \cite{GneRan13}, we say that the PITs, $F(y_{1}|\btheta),\ldots,F(y_{T}|\btheta)$, are well calibrated  (or probabilistically calibrated) if their distribution is uniform. As noted in \cite{GneRan13}, well calibration is a critical requirement for probabilistic forecasts and checks for the uniformity of the probability integral transform have formed a cornerstone of density forecast evaluation.

The aggregation method introduced by \cite{RanGne10} and \cite{GneRan13}
considers the beta transformed linear pool
\begin{equation}  \label{eq:Beta}
F(y|\btheta) = B_{\alpha,\beta} \!
       \left( \, \sum_{m=1}^M \omega_m \hsp F_m(y) \right)
\end{equation}
for $y \in \real$, where $\btheta=(\alpha,\beta,\bomega)$, $B_{\alpha,\beta}$ denotes the cdf of the standard beta distribution with parameters $\alpha > 0$ and $\beta > 0$ and density function $b_{\alpha,\beta}(x)$, proportional to $x^{\alpha - 1} (1-x)^{\beta - 1}$ on the unit
interval and  $\bomega$
belong to the the unit simplex in $\real^M$
\begin{equation*}
\DM = \left\{ \bomega = (\omega_1, \ldots, \omega_M) \in [0,1]^M :
              \sum_{m=1}^M \omega_m = 1 \right\}.
\end{equation*}
%If $F_1, \ldots, F_M$ admit Lebesgue densities $f_1, \ldots, f_M$, respectively, the combination formula (\ref{eq:Beta}) can be written equivalently in terms of the aggregated
%probability density function (pdf), namely
%\begin{equation} \label{eq:beta}
%f(y|\btheta) = \left( \, \sum_{m=1}^M \omega_m \hsp f_m(y) \right) \,
%b_{\alpha,\beta} \!  \left( \, \sum_{m=1}^M \omega_m \hsp F_m(y) \right)
%\end{equation}
%for $y \in \real$, where $b_{\alpha,\beta}$ is the pdf of the beta
%distribution.  

We interpret $B_{\alpha,\beta}$ as a parametric calibration
function, which acts on a linear combination of $F_1, \ldots, F_M$
with mixture weights $\bomega \in \DM$.  When $\alpha = 1$ and $\beta = 1$, the calibration function is the identity function, and the beta transformed linear pool reduces to the traditional linear pool. In the case $M=1$, the beta transformation serves to achieve calibration; when $M>1$, we seek to combine and calibrate simultaneously. The linear combination weights assign relative importance to the individual predictive distributions, and the beta
transformed linear pool admits exchangeable flexible dispersivity in the sense of the second moment of the PITs histogram \citep{GneRan13}. However, in order to achieve the stronger result of probabilistic calibrated PIT the parametric class of calibration functions needs to be extended.

%Given a sequence of observations, $y_{1},\ldots,y_{T}$,
%the cdf evaluated on the observation, $F(y_{t}|\btheta)$, is referred as probability integral transform (PIT).  Following \cite{RanGne10} and \cite{GneRan13}, we say that the PITs, $F(y_{1}|%\btheta),\ldots,F(y_{T}|\btheta)$, are well calibrated  (or probabilistically calibrated) if their distribution is uniform.
%In practice, the weights $(\omega_1,\dots,\omega_M)$ and transformation parameters $\alpha$ nd $\beta$ need to be estimated on a training data set \citep{GneRan13}.
%In point of fact, one can consider combination formulas
%of the form $\{ \CG_{\btheta}(y|F_1,\dots,F_M)=U_{\btheta_c} ( \sum_{m=1}^M \omega_m \hsp F_m(y) );  \btheta=({\btheta}_c, \bomega) \in \Theta=\Theta_c \times \Delta_M \}$ where %$U_{\btheta_c}$ is cdf in the class $\{ U_{\btheta_c}: \btheta_c \in \Theta_c\}$ of distribution functions on $[0,1]$.

In point of fact,  combination formulas \eqref{eq:Beta}  can be generalised to 
\[
\Big \{ \CG_{ C,\bomega}(y|F_1,\dots,F_M)={C} \Big( \sum_{m=1}^M \omega_m \hsp F_m(y) \Big) ; \,\,  \bomega \in \Delta_M, C \in \mathcal{C}  \Big \}
\]
 where $\mathcal{C}$ is
a class of distribution functions on $[0,1]$.  Given a single observation $y$ with continuous and strictly increasing cdf $F_0$, then one can write 
$F_0(y)=C_0(\sum_{m=1}^M \omega_m \hsp F_m(y))$, where 
$C_0(z)=F_0(F^{-1}(z))$ and $F^{-1}$ is the inverse function of $F=\sum_{m=1}^M \omega_m \hsp F_m$. This shows that if for any $F_0$  the corresponding $C_0$ belong to $\mathcal{C}$, then the PITs of the model are well calibrated. In practice the distribution $F_0$ is not known and one is left with the issues of choosing the class $\mathcal{C}$ and estimating the calibration formula. The choice of $\mathcal{C}$ should achieve a compromise between parsimony and flexibility. It is well known that any continuous function $g$ on the closed unit interval $[0,1]$ can be well approximated by a beta mixture in the sup norm. See, e.g., Theorem 1 in  \cite{DiaYlv85} or \cite{AltCam94}, p. 333. Moreover any continuous density on $(0,1)$ with finite entropy can be well approximated with respect to the  Kullback-Leibler divergence by a finite mixture of beta densities, see Theorem 1 in  \cite{RobRou02}. For this reason, we extend (\ref{eq:Beta}) of \cite{GneRan13} by introducing an aggregation method based on beta mixture combination formulas
\begin{equation}   \label{eq:Betamixture}
F(y|\btheta) = \sum_{k=1}^K \, w_k \,
       B_{\alpha_k, \, \beta_k} \!
       \left( \, \sum_{m=1}^M \omega_{km} \hsp F_m(y) \right)
\end{equation}
for $y \in \real$, where $\btheta=(\bw,\balpha,\bbeta,\bomega)$, the vector $\bw = (w_1, \ldots, w_K) \in \DK$
comprises the beta mixture weights, $\balpha=(\alpha_1, \ldots, \alpha_K)$ and
$\bbeta=(\beta_1, \ldots, \beta_K)$ are beta calibration parameters, and $\bomega=(\bomega_1,\ldots,\bomega_K)$, with
$\bomega_k = (\omega_{k1}, \ldots, \omega_{kM})\in \DM$, $k=1\ldots,
K$ the component-specific sets of linear combination weights.

When $K<\infty$ we refer to the general beta mixture model in (\ref{eq:Betamixture}) and (\ref{eq:betamixture}) as the BM$_K$ model, which is much more flexible, and has the beta transformed linear pool proposed by \cite{RanGne10} and
\cite{GneRan13} as a special case for $K = 1$. When $K$ is unknown, Bayesian inference can provide guidance in choosing appropriate compromises between parsimony and flexibility.  In particular, our Bayesian approach allows us to treat the parameter $K$ as unbounded and random.  We refer to this latter setting as the infinite beta mixture or BM$_\infty$ calibration, for which we give details in the following section.

As a final remark, we note that the beta mixture calibration and combination model can also be interpreted in terms of generalized linear pool, introduced by \cite{MitKap13}. Assume $F_1, \ldots, F_M$ admit Lebesgue densities $f_1, \ldots, f_M$, respectively, then the combination formula (\ref{eq:Betamixture}) can be written equivalently in terms of the aggregated
probability density function (pdf), namely
\begin{equation}   \label{eq:betamixture}
f(y|\btheta) = \sum_{k=1}^K \, w_k
       \left( \, \sum_{m=1}^M \omega_{km} \hsp f_m(y) \right) \,
       b_{\alpha_k, \, \beta_k} \!
       \left( \, \sum_{m=1}^M \omega_{km} \hsp F_m(y) \right)
\end{equation}
which can be written as
\begin{equation*}
f(y|\btheta) = \sum_{m=1}^M \tilde{\omega}_m(y) \, f_m(y)
\end{equation*}
for $y \in \real$, where the generalized weight functions are given by
\begin{equation*}
\tilde{\omega}_m(y) = \sum_{k=1}^K \omega_{km}w_k \,
       b_{\alpha_k, \, \beta_k} \!
       \left( \, \sum_{m=1}^M \omega_{km} \hsp F_m(y) \right)
\end{equation*}
for $m = 1, \ldots, M$.  We should notice that this simple result provides an alternative interpretation of our model as a generalized combination model, similarly to the scheme of \cite{MitKap13}. However, there are two major differences with respect to \cite{MitKap13}. First, our weights are built to achieve uniformity of the PITs, thus our model should be used each time calibration, and not only combination, is needed. Secondly, they use weight functions that are piecewise constant, whereas the weight functions implied by the beta mixture model are continuous, thus allowing for a smooth combination density.

\section{Bayesian inference}  \label{sec:inference}

In Bayesian settings, it is convenient to express the standard beta
distribution with parameters $\alpha > 0$ and $\beta > 0$ in terms of its
mean $\mu = \alpha / (\alpha + \beta)$ and the parameter $\nu = \alpha
+ \beta > 0$ \citep{Epstein66,RobRou02, BilCas11, CasDalLei12}.  We refer to
the reparameterized pdf as
\begin{equation*}
b^*_{\mu,\nu}(x) =
  \frac{\Gamma(\nu)}{\Gamma(\mu\nu) \hsp \Gamma((1-\mu)\nu)} \,
  x^{\hsp \mu\nu - 1} (1-x)^{(1 - \mu) \nu - 1} \, \one_{[0,1]}(x),
\end{equation*}
where $\Gamma$ denotes the gamma function, and we use the symbol $B^*_{\mu,\nu}$
to denote the corresponding cdf.

We discuss inference in a time series setting at the unit prediction
horizon, where the training data comprise the predictive cdfs $F_{1t},
\ldots, F_{Mt}$, which are conditional on information available at
time $t - 1$, along with the respective realization, $y_t$, at time $t
= 1, \ldots, T$, respectively.  We then wish to estimate a
calibration and combination formula of the form (\ref{eq:Betamixture})
that maps the tuple $F_{1t}, \ldots, F_{Mt}$ into an aggregated cdf,
$F_t$.  In practice, we use the estimated calibration and combination
formula to aggregate the predictive cdfs $F_{1,T+1}, \ldots, F_{M,T+1}$,
which are based on information available at time $T$, into a
single predictive cdf, $F_{T+1}$, for the subsequent value, $y_{T+1}$, of the
variable of interest. Extensions to multi-step ahead
forecasts is possible, and we leave this for further research.

To ease the notational burden in the time series setting, let $\bomega_k = (\omega_{k1},
\ldots, \omega_{kM}) \in \DM$, and write
\begin{equation}  \label{eq:H}
H_t(y_t|\bomega_k) = \sum_{m=1}^M \omega_{km} \, F_{mt}(y_t)
\end{equation}
and
\begin{equation}  \label{eq:h}
h_t(y_t|\bomega_k) = \sum_{m=1}^M \omega_{km} \, f_{mt}(y_t)
\end{equation}
for $t = 1, \ldots, T$ and $k = 1, 2, \ldots, K$, respectively.

\subsection{Bayesian finite beta mixture model}
We work with a reparameterized version of the finite beta mixture
calibration and combination model (i.e., $K<\infty$), in which the
aggregated cdf and pdf can be represented as
\begin{equation}  \label{eq:BM.K}
F_t(y_t|\btheta) = \sum_{k=1}^K \, w_k \, B^*_{\mu_k, \, \nu_k}(H_t(y_t|\bomega_k))
\end{equation}
and
\begin{equation}  \label{eq:BM.K}
f_t(y_t|\btheta) = \sum_{k=1}^K \, w_k \, h(y_t|\bomega_k)b^*_{\mu_k, \, \nu_k}(H_t(y_t|\bomega_k))
\end{equation}
for $t = 1, \ldots, T$. The parameter vector for the BM$_K$ model can
then be written as $\btheta = (\bw, \bmu, \bnu, \bomega)$, where $\bw
= (w_1, \ldots, w_K) \in \DK$, $\bmu = (\mu_1, \ldots, \mu_K) \in
(0,1)^K$, $\bnu = (\nu_1, \ldots, \nu_K) \in (0,\infty)^K$ and
$\bomega = (\bomega_1, \ldots, \bomega_K) \in \Delta_M^K$, with $K$
being a fixed positive integer. The parameter space is defined as
$\Theta=\DK\times(0,1)^{K}\times(0,\infty)^{K}\times\Delta_M^{K}$.

Our Bayesian approach assumes that
\begin{eqnarray}
\bw   & \sim & \Di(\xi_{w1},\ldots,\xi_{wM}) \label{eq:w}\\
\mu_k & \sim & \Beta(\xi_{\mu 1},\xi_{\mu 2}),
\label{eq:mu} \\
\nu_k & \sim & \Ga(\xi_{\nu 1},\xi_{\nu 2}), \label{eq:nu} \\
\bomega_k & \sim &\Di(\xi_{\, \omega 1},\ldots,\xi_{\, \omega M}) \label{eq:omega}
\end{eqnarray}
for $k = 1, \ldots, K$, where $\Beta(\alpha,\beta)$ is a Beta
distribution with density proportional to $x^{\alpha - 1} (1-x)^{\beta
- 1}$ for $x \in \Delta_1$, $\Ga(\gamma,\delta)$ is a Gamma
distribution with density proportional to $x^{\gamma} \exp\{-\delta
x\}$ for $x > 0$, and $\Di(\xi_1, \ldots, \xi_M)$ is a Dirichlet
distribution with density proportional to $\prod_{m=1}^M
w_m^{\xi_m-1}$ for $(w_1, \ldots, w_M) \in \Delta_M$, with all these
distributions being independent.  Guided by symmetry arguments in the
Beta and Dirichlet case, and using a standard, uninformative prior in
the Gamma case \citep{Spi04}, we parameterize parsimoniously and set
$\xi_{w1} = \cdots = \xi_{wM}$, $\xi_\mu = \xi_{\mu 1} = \xi_{\mu 2}$,
$\xi_{\nu 1} = \xi_{\nu 2}$, and $\xi_{\, \omega 1} = \cdots = \xi_{\,
\omega M}$.  In what follows, we refer to the common hyper\-parameter
values as $\xi_w$, $\xi_\mu$, $\xi_\nu$, and $\xi_\omega$, respectively

Adopting a data augmentation framework \citep{Fru06}, we introduce the
allocation variables $d_{kt} \in \{0,1\}$, where $k = 1,\ldots,K$ and
$t = 1,\ldots,T$.  The likelihood of the BM$_K$ calibration model is the marginal of the complete data
likelihood
\begin{equation*}
L(Y,D \, | \, \btheta) = \prod_{t=1}^T \prod_{k=1}^K
\left( w_k \, h_t(y_t|\bomega_k) \, b^*_{\mu_k,\nu_k} \left(H_t(y_t|\bomega_k) \right)
\right)^{d_{kt}},
\end{equation*}
where we let $Y = (y_1,\ldots,y_T)$ and $D =
(d_{11},\ldots,d_{K1},\ldots,d_{1T},\ldots,d_{KT})$.  The implied
joint posterior of $D$ and $\btheta$ given the observations $Y$ satisfies
\begin{equation*}
\pi(D,\btheta \, | \, Y) \propto g(\bmu,\bnu,\bomega) \,
\prod_{k=1}^K w_k^{\xi_w + T_k - 1}
\prod_{t \in \cD_k} h_t(y_t|\bomega_k) \, b^*_{\mu_k,\nu_k} (H_t(y_t|\bomega_k)),
\end{equation*}
where $g(\bmu,\bnu,\bomega)$ is the prior density, $\cD_k = \{
t=1,\ldots,T \, | \, d_{kt} = 1 \}$, and $T_k$ is the number of
elements in $\cD_k$.  To sample from the joint posterior, we use a
Gibbs sampler that draws iteratively from $\pi(D \, | \, \btheta,Y)$,
$\pi(\bmu,\bnu \, | \, \bw,\bomega,D,Y)$, $\pi(\bomega \, | \,
\bw,\bmu,\bnu,D,Y)$, and $\pi(\bw \, | \, \bomega,\bmu,\bnu,D,Y)$,
respectively, for which we give details in Appendix
\ref{sec:details.finite}.

The output of the algorithm is a sample $\btheta^{(i)} = (\bw^{(i)}, \bmu^{(i)}, \bnu^{(i)},
\bomega^{(i)})$ for $i = 1, \ldots, I$, where $I$ is the number of
iterations in the Gibbs sampler.  The sample is used to approximate
with $\hat{F}_{T+1}(y_{T+1})$ the desired one-step-ahead cumulative posterior predictive distribution $F_{T+1}(y_{T+1})$
$=\int_{\Theta}F_{T+1}(y_{T+1}|\btheta)\pi(\btheta|Y)d\btheta$, where $\pi(\btheta|Y)$ is the marginal distribution of $\pi(D,\btheta|Y)$. In the special case when $K = 1$ we get
\begin{equation}  \label{eq:BayesBM1}
\hat{F}_{T+1}(y_{T+1}) =
\frac{1}{I} \sum_{i=1}^I B^*_{\mu^{(i)},\nu^{(i)}} \!
\left( \sum_{m=1}^M \omega_{m}^{(i)} \, F_{m,T+1}(y_{T+1}) \right) \! ,
\end{equation}
which can be thought of as a Bayesian implementation of the beta
transformed linear pool (\ref{eq:Beta}) of \cite{RanGne10} and
\cite{GneRan13}. An advantage of the proposed approach based on Gibbs sampling approximation is that parameter uncertainty can be taken into consideration in the prediction. A plug-in approximation of the predictive, which does not account for the parameter uncertainty, can be used, namely $\hat{F}_{T+1}(y_{T+1})=F_{T+1}(y_{T+1}|\hat{\btheta})$ where $\hat{\btheta}$ is the parameter posterior mean which can be approximated by the empirical average of $\btheta^{(i)}$ $i=1,\ldots,I$.  Another advantage of our approach is that credible intervals for the calibrated predictive cdf can be easily approximated by using the output of the Gibbs sampler.

\subsection{Bayesian infinite beta mixture model}
In the finite-mixture beta calibration and combination model the
number of beta densities is given, and model selection procedures
can be used to choose the number of mixture components.  As evidenced
in previous studies (see, e.g., \cite{BilCasRavVan13} and \cite{MitKap13}),
in a time series context the model pooling scheme can be subject to
time instability, thus as a new group of observations arrives the
pooling scheme can change dramatically. Moreover, \cite{Geweke2010} 
discusses how Bayesian model averaging probabilities and linear opinion pooling 
weights converge in the limit to select one model (or a subset of
models) and therefore they might not properly cope with such instability. 
The finite-mixture beta calibration and combination model might be subject to 
similar issues and to solve it we propose to work with an infinite prior number
of calibration functions and local pooling schemes of which only a finite
number are selected on a given finite sample.  The
consequence is that the number $K$ of beta mixture components can vary
and increase with the sample size, accounting for time instability in the pooling. 
Therefore, the model with infinite calibration components provides
an answer to the problem of selecting the number of components in the
finite mixture approach.

We propose here a Bayesian non-parametric models which
allows for estimating the number of components and also for including
the model uncertainty in the posterior predictive.  We refer to this
model as the infinite-mixture calibration model $BM_{\infty}$.  Let us assume
\begin{equation*}
f_t(y_t|\btheta)=b^{*}_{\mu,\nu}\left(H_t(y_t|\bomega)\right)h_t(y_t|\bomega),
\end{equation*}
where $\btheta = (\mu,\nu,\bomega)$, with $\bomega =
(\omega_1,\ldots,\omega_M)$.  Our prior for the $BM_{\infty}$ parameters $\btheta$ is
nonparametric, i.e. $\btheta\sim G(\btheta)$ where $G$ is a random
probability measure
\begin{equation*}
G \sim DP(\psi, G_0)
\end{equation*}
and $DP(\psi, G_0)$ denotes a Dirichlet process (DP)
(\cite{Ferguson73}) with concentration parameter $\psi$ and base
measure $G_0$.  Following the standard result of \cite{Set94}, the
Dirichlet process prior can be represented as
\begin{equation*}
G(d\btheta)=\sum_{k=1}^{\infty}w_k\delta_{\btheta_k}(d\btheta)
\end{equation*}
with random weights $w_k$ generated by the stick-breaking construction
\begin{equation*}
w_k=v_k\prod_{l=1}^{k-1}(1-v_{l})
\end{equation*}
where the stick-breaking components $v_{l}$ are i.i.d. random
variables from $\Beta(1,\varphi)$.  The atoms $\theta_k$ are
i.i.d. random variables from the base measure $G_0$.  In our model the
base measure is given by the product of the following distribution
\begin{equation*}
\Beta(\xi_{\mu},\xi_{\mu})
\Ga(\xi_{\nu}/2,\xi_{\nu}/2)
\Di(\xi_{\omega},\ldots,\xi_{\omega}).
\end{equation*}
The Dirichlet process prior assumption and the stick-breaking
representation of the DP allow us to write the combination and
calibration model in terms of infinite mixtures of random beta
distributions with the following random pdf
\begin{eqnarray*}
f_t(y_t|G)
& = & \int f_t(y_t|\btheta)G(d\btheta)\\
& = & \sum_{k=1}^{\infty}w_k b^{*}_{\mu_k,\nu_k}\left(H_t(y_t|\bomega_k)\right)
      h_t(y_t|\bomega_k).
\end{eqnarray*}

The number of components sampled in the first $T$ observations is random and its prior distribution is
(\cite{Ant74})
\begin{equation*}
P(K=k|\psi,T)=\frac{T! \Gamma(\psi)}{\Gamma(\psi+T)}|s_{Tk}|\psi^{k}
\end{equation*}
for $k=1,2,\ldots$, where $s_{Tk}$ is the
signed Stirling number \citep[p.~824]{AbrSte72}.  The dispersion
hyper-parameter $\psi>0$ is driving the prior expected number of
components.  Large values of $\psi$ increase the probability of introducing new components in the mixture.  As the prior dispersion depends crucially on this parameter, the results of the posterior
inference on the infinite mixture model are usually presented for
different values of $\psi$.  It also possible to extend the
nonparametric models by assuming a further stage of the prior
hierarchical structure and assuming a prior for $\psi$.  A common
choice for the prior is a gamma distribution, $\Ga(c,d)$ (see
\cite{EscobarWest1995}).  The second important feature is that our
inference approach provides, as a natural product, the posterior
distribution of the number of components given a sample of data and
allows for the inclusion of the number of components uncertainty in
the predictive density.

Inference on infinite mixture models resulting from a Dirichlet prior
assumption requires the use of simulation methods.  Gibbs samplers
have been proposed in \cite{Escobar94} and \cite{IsJam2001}, which
make use of the Polya-urn representation of the Dirichlet
process. \cite{IshZar00a} proposed a sampler based on a truncation of
the infinite mixture representation.  \cite{PapRob08} proposed an
exact simulation algorithm based on retrospective sampling.  In this
paper we apply the slice sampling algorithm proposed in
\cite{walker2007} and \cite{walker2011}.  The algorithm uses a set of
auxiliary variables to deal with the infiniteness problem of the
mixture model.  More specifically, let us introduce a sequence of
slice sampling variables $u_t$, $t=1,2,\ldots,T$, then $f_t(y_t|G)$ is the marginal of
\begin{equation*}
f_t(y_t,u_t|G)=\sum_{k=1}^{\infty}\one_{\{u_t<w_k\}}
b^{*}_{\mu_k,\nu_k}\left(H_t(y_t|\bomega_k)\right)h_t(y_t|\bomega_k)
\end{equation*}

Note that given a set of observations, $y_t$ and slice variables, $u_t$, $t=1,\ldots,T$,
the complete data likelihood can be written as
\begin{equation*}
L(Y,U|G) =
\prod_{t=1}^{T}\sum_{k\in A_t}b^{*}_{\mu_k,\nu_k}\left(H_t(y_t|\bomega_k)\right)
h_t(y_t|\bomega_k),
\end{equation*}
where $Y=(y_{1},\ldots,y_{T})$, $U=(u_{1},\ldots,u_T)$,
$A_t=\{k|u_t < w_k\}$.  Note that $N_t=Card(A_t)$, that is the
number of components of the infinite sum, is finite when conditioning
on the slice variables.  Thus, the introduction of the auxiliary
variables allows us to have a finite mixture representation of the
infinite mixture model.  Following a standard approach to inference
for mixture models (e.g., see \cite{Fru06}) we now introduce a
sequence of allocation variables, $d_t$, $t=1,\ldots,T$, with
$d_t\in A_t$.  Each of these variables indicates which component of
the finite mixture provides the observation $y_t$.  The complete
data likelihood is
\begin{equation*}
L(Y,U,D|G) = \prod_{t=1}^{T}
\one_{\{u_t<w_{d_t}\}}\!b^{*}_{\mu_{d_t},\nu_{d_t}}\!
\left(H_t(y_t|\bomega_{d_t})\right) h_t(y_t|\bomega_{d_t})\!\!\!\!\!\!\!\!\!
\end{equation*}
where $D=(d_{1},\ldots,d_T)$.

Let us denote by $V=(v_1,v_2,\ldots)$ and
$\Theta=(\btheta_{1},\btheta_{2},\ldots)$, with
$\btheta_k=(\mu_k,\nu_k,\bomega_k)$,
$\bomega_k=(\omega_{1k},\ldots,\omega_{Mk})$, the infinite dimensional
vectors of the stick-breaking components and atoms respectively.  In
what follows we assume the dispersion parameter $\psi$ is unknown with prior distribution $\pi(\psi)$.

From the completed likelihood function and our assumptions on the
prior distributions, the joint posterior distribution of $U$, $D$,
$V$, $\Theta$ and $\psi$ given $Y$ is
\begin{eqnarray*}
&& \pi(U,D,V,\Theta,\psi |Y) \propto
\prod_{t=1}^{T}\one_{\{u_t<w_{d_t}\}}b^{*}_{\mu_{d_t},\nu_{d_t}}
\left(H_t(y_t|\bomega_{d_t})\right) h_t(y_t|\bomega_{d_t})\\
&& \cdot \; \prod_{k\geq1}(1-v_k)^{\psi-1}\mu_k^{\xi_{\mu}-1}(1-\mu_k)^{\xi_{\mu}-1}
\nu_k^{\xi_{\nu}/2}\exp\{-\xi_{\nu}\nu_k/2\}\prod_{i=1}^{M}\omega_{ik}^{\nu/2-1}\pi(\psi).
\end{eqnarray*}
Joint sampling from the posterior is not possible and this calls for the application of a Gibbs sampling procedure.  Adapting the sampler described in \cite{walker2007} and \cite{walker2011} to our setting, we develop an efficient collapsed Gibbs sampling procedure which generates
sequentially the parameters and the latent variables from the full
conditional distributions $\pi(\Theta|U,D,V,Y,\psi)$,
$\pi(V,U|\Theta,D,Y,\psi)$, $\pi(D|\Theta,V,U,Y,\psi)$ and
$\pi(\psi|Y)$.  The details of the steps of the Gibbs sampler
are given in Appendix \ref{sec:details.infinite}.

The output of the algorithm are samples $\bw^{(i)}$ and $\btheta^{(i)}=(\bmu^{(i)},\bnu^{(i)},\bomega^{(i)})$ for $i=1,\ldots,I$ where $I$ is the number of MCMC iterations, and can be used to sample from the one-step-ahead cumulative predictive distribution. For further details see Appendix \ref{sec:details.infinite}.

%namely
%\begin{equation}  \label{eq:BayesBMK}
%\hat{F}_{T+1}(y_{T+1}) =
%\frac{1}{I} \sum_{i=1}^I \sum_{k=1}^{\infty} w_{k}^{(i)} \,
%B^*_{\mu_{k}^{(i)},\nu_{k}^{(i)}} \! \left( \sum_{m=1}^M \omega_{km}^{(i)} \, F_{m,T+1}(y_{T+1}) \right)
%\end{equation}

%\newcommand{\th}{\theta}
\newcommand{\RE}{\mathbb{R}}
\newcommand{\eps}{\epsilon}

\section{Posterior consistency}\label{sec:consistency}
In this section we provide some conditions under which the proposed probabilistic calibration formula, $BM_\infty$, converges to the true underlying density, implying uniformity of the PITs in the limit. The convergence is studied in terms of weak posterior consistency when the number of observations in the training set goes to infinity. 
We consider the case the set of combined models  is externally provided 
 and  both combination weights and calibration parameters are estimated. 
We prove   consistency of our $BM_\infty$ calibration  for i.i.d. observations and show that posterior consistency is still valid under the assumption of Markovian kernels.  As commonly found in the literature, for non i.i.d. observations, the posterior consistency is case-specific depending heavily on the model used. See, for instance \cite{TangGhosal07} and \cite{ChGhRo04}.
 
\subsection{Consistency results for i.i.d. observations}\label{sec:consistency_iid} 
%For the sake of simplicity we limit the study to the i.i.d. observations case.
Let  $\mathcal{F}$ be the set of all possible densities (with respect to Lebesgue measure) on the sample space $\mathcal{Y}\subset \RE$ and $\Pi^{*}$ be a prior on $\mathcal{F}$.
%Note that every $f$ in $\mathcal{F}$ corresponds to the probability measure  $P_f(dy)=f(y)dy$ and hence, with a slight
%abuse of language one can denote by  $\Pi^{*}$ also  the the prior induced on the set of probability measures by the map
%$f \mapsto P_f$.
In order to deal with posterior consistency one needs to specify which kind of topology
is assumed on $\mathcal{F}$. Since we are interested in weak consistency,
we  consider the topology of weak convergence, i.e. the topology  induced by the weak
convergence of probability measures. Given a prior $\Pi^{*}$ on $\mathcal{F}$, the posterior is said to be {\it weakly consistent} at $f_0$ if $\Pi^*(U|y_1,\dots,y_n)$ converges a.s. to 1 for every neighbourhood $U$ of $f_0$ in the topology of weak convergence, for every i.i.d. sequence $y_1,y_2,\ldots$ of random variables with common density $f_0$.
 %Note that every $f$ in $\mathcal{F}$ corresponds to the probability measure  $P_f(dy)=f(y)dy$
%on  $\mathcal{Y}\subset \RE$ and recall that for any probability measure $P_{f_0}$,
% a neighbourhood base (of the weak topology) consists of sets of the form $\{ P: |\int \phi_i dP_{f_0} -\int \phi_i dP|<\eps : i=1,\dots,k\}$ where $\phi_i:i=1,\dots,k$ are bounded %continuous function on $\mathcal{Y}$.
For more details see, e.g., Chapter 4 in \cite{GhoshRamamoorthiNP}.
 
 In the i.i.d. case, Schwartz theorem states that weak consistency at a ``true density'' $f_0$ holds if the
prior assigns positive probabilities to Kullback-Leibler  neighborhoods of $f_0$.
%and the size of the model is restricted in some appropriate sense.
%
%For the weak topology the model size condition for the weak consistency holds automatically.
Hence,  in this setting, in order to prove week consistency, one only needs to check if the Kullback-Leibler property is satisfied by the prior setting and the true density $f_{0}$, see Theorem 4.4.2 in \cite{GhoshRamamoorthiNP}.

It is worth recalling that  a Kullback-Leibler neighbourhood of a density $f\in\mathcal{F}$ of size $\varepsilon$ is defined as
$$
\mathcal{K}_{\varepsilon}(f_{0})=\left\{g\in\mathcal{F} | \int f\log\left(\frac{f}{g}\right)\leq\varepsilon\right\},
$$
and the Kullback-Leibler property holds at $f_{0}\in\mathcal{F}$, for short $f_{0}\in KL(\Pi^{*})$, if $\Pi^{*}(\mathcal{K}_{\varepsilon}(f_{0}))>0$ for all $\varepsilon>0$.
We will denote with $supp(\mu)$  the weak support of a probability measure $\mu$
and with $KL(f,g)$ the Kullback-Leibler divergence
 between the two densities $f$ and $g$, i.e. $KL(f,g):=\int f \log\big(\frac{f}{g}\big)$.

In this section we will exploit the type I mixture prior representation of  $\Pi^*$. Let us recall that
a prior on $\mathcal{F}$ is said to be a type I mixture prior if it is induced via the map
\begin{equation}
G\mapsto f_{G}(y)=\int_\Theta K(y;\boldsymbol{\theta})G(d\boldsymbol{\theta}),
\end{equation}
where
$\Theta$ is  the mixing parameter space, $K(y;\boldsymbol{\theta})$ a density kernel  on $\mathcal{Y}\times\Theta$ and
$G$ has distribution $\Pi$ on the  space  $\mathcal{M}(\Theta)$ of probability measures on $\Theta$ (see \cite{WuGhosal2008}).
In our joint calibration and combination model, the kernel is
\begin{equation}
K(y;\boldsymbol{\theta})=b^*_{\mu,\nu}(H(y|\boldsymbol{\omega}))h(y|\boldsymbol{\omega})
\end{equation}
with $\boldsymbol{\theta}=(\boldsymbol{\theta}_{p},\boldsymbol{\theta}_{c})$, where $\boldsymbol{\theta}_{p}=\boldsymbol{\omega}$ indicates the pooling parameters, and $\boldsymbol{\theta}_{c}=(\mu,\nu)$ the calibration parameters.
%For the sake of simplicity the proof of weak consistency is given for i.i.d. observations case, thus
Since in this subsection  we deal with  of i.i.d. observations,
we drop from the kernel $K$ the observation index $t$.
The random mixing distribution $\Pi$ is given by a Dirichlet process prior, so that
%\begin{equation}
$\boldsymbol{\theta}|G \sim G$
%\end{equation}
where $G\sim DP(\psi,G_{0})$.  For the sake of simplicity we assume that the concentration parameter $\psi$ is given.
%
%\subsubsection{Joint calibration and combination consistency}
%Let us first consider the case in which both the pooling parameters and the calibration parameters are unknown. 
%In this case $\Theta={\Delta_M} \times [0,1]\times\RE^+$ and $G$ is a DP process on $\mathcal{M}({\Delta_M} \times [0,1]\times\RE^+)$ with base measure $G_0$ on $\Delta_M \times [0,1]\times\RE^+$
%and concentration parameter  $\psi>0$.
%
%Here 
 In this way, $\Pi^{*}$ turns out to be
the prior on $\mathcal{F}$ induced by
\[
G \mapsto  \int b^*_{\mu,\nu}(H(y|\boldsymbol{\omega}))  h(y|\boldsymbol{\omega}) G(d \boldsymbol{\omega} d\mu d\nu)
\]
when $G \sim DP(\psi,G_0)$
%
%Before stating the first result, let us recall that
and $h(y|\boldsymbol{\omega})=\sum_{m=1}^M \boldsymbol{\omega}_m f_m(y)$.

\begin{theorem}\label{corollario}
Assume that there is a point $\boldsymbol{\omega^*}$ in the interior of $\Delta_M$ such that
$h(\cdot|\boldsymbol{\omega^*})$ is continuous and  that,
for every compact set $C \subset \mathcal{Y}$,
%\begin{equation}\label{lowerboundh0}
 $\inf_{y \in C} h(y|\boldsymbol{\omega^*})>0$.
%\end{equation}
Assume also that the true density $f_0$ is continuous on $\mathcal{Y}$
and that
\begin{equation}\label{entropy2}
\begin{split}
& \int [|\log(H(y|\boldsymbol{\omega^*}))|+|\log(1-H(y|\boldsymbol{\omega^*}))|]f_0(y)dy<+\infty \\
& \text{and} \quad KL(f_0,h(\cdot|\boldsymbol{\omega^*})) <+\infty. \\
\end{split}
\end{equation}
If $G_0$ has full support, then
$f_{0}\in KL(\Pi^{*})$.
\end{theorem}

The proof of the previous theorem is given in the Supplementary Materials \ref{appendix_consistencyl}.

The result in Theorem  \ref{corollario} is of general validity and the theorem assumptions can be easily checked for many applied contexts. 
In  Supplementary Materials \ref{additionalconsistency} we provide two additional results { for the case not only the combined models are given, but also  the combinations weights  are fixed,  see Theorems \ref{teocons2}-\ref{thm_cons1}. }
In Theorem \ref{corollario} we assumed the combination weights  $\boldsymbol{\omega^*}$ belong to the interior of $\Delta_M$. This assumption is not restrictive and does not necessarily  implies neither that the true density $f_0$ is a mixture  of $f_m$ nor that $f_0$ is one of the models in the combination. In the following examples, we provide a clearer explanation of this aspect and show how Theorem \ref{corollario} assumptions are satisfied for the Gaussian mixture and Student-t mixture examples considered later on in this paper for the simulation study. 
{ We also considered the two cases the set of combined models includes the true density (complete model set) and does not include it (incomplete model set). 
%Also, we show that it is not difficult to check if the theorem assumptions are satisfied when true density is heavy-tailed or has a discrete mixture representation.}

\begin{example}[Complete and incomplete model set]\label{ex:completeincomplete}
Consider the case in which
\[
%\begin{split}
 h(y|\boldsymbol{\omega})=\sum_{m=1}^M\omega_{m} \varphi(y|{\mu_m,\sigma_m^2}) \quad \text{and} \quad
 f_0(y)=  \varphi(y|{\mu_0,\sigma_0^{2}}) \\
%\end{split}
\]
where  $\varphi(\cdot|{\mu,\sigma^2})$ is the pdf of a normal distribution of mean $\mu$ and variance $\sigma^2$.
When $(\mu_0,\sigma^2_0)=(\mu_m,\sigma^2_m)$ for some $m=1,\dots,M$ the model set is complete, otherwise it is incomplete. In both cases it turns out that $f_{0}\in KL(\Pi^{*})$. 

In order to apply Theorem  \ref{corollario} one needs to check  that \eqref{entropy2} is satisfied
for some $\boldsymbol{\omega^*}$ in the interior of $\Delta_M$. We shall show that this is true for  the equal weights linear pooling, $\boldsymbol{\omega^*}=(1/M,\dots,1/M)$.
To this end, denoting  by $\Phi(\cdot|\mu,\sigma^2)$  the cumulative distribution function of $\varphi(\cdot|{\mu,\sigma^2})$,
observe that:
\begin{itemize}
\item[(i)] given a mixture of $M$ normal distributions with means and variances $(\mu_m,\sigma_m^2)$, $m=1,\dots,M$,
%Given $0<\sigma_{-} < \sigma_i < \sigma_{+}$ ($i=1,\dots,M$) and $\mu_1,\dots,\mu_M$ real numbers,
if $0<\sigma_- < \min_m \sigma_m \leq \max_m \sigma_m< \sigma_{+}$, then
there are two constants $C^-$ and $C^+$
such that, for every $y$,
\[
C^{-} \varphi(y|{0,\sigma_{-}^2}) \leq \sum_{m=1}^M \omega_m \varphi(y|{\mu_m,\sigma_{m}^2}) \leq C^{+} \varphi(y|{0,\sigma_{+}^2});
\]
See Lemma \ref{lemma:gaussiantailmixture} in Supplementary Materials. 
\item[(ii)]  as $y \to +\infty$, one has  $(1-\Phi(y|{0,1}))/\varphi(y|{0,1}) \sim 1/y$ and hence
$|\log(1-\Phi(y|{0,\sigma^2}))| \sim y^2/\sigma^2$.
\end{itemize}
Using (i) and (ii)  one can check that
\[
|\log(1-H(y|\boldsymbol{\omega^*}))| \leq C \max\{|\log(1-\Phi(y|{0,\sigma{-}^2})|,|\log(1-\Phi(y|{0,\sigma{+}^2})|\}
\leq C' y^2
\]
for suitable constants $C,C'$.
Hence
\[
\int |\log(1-H(y|\boldsymbol{\omega^*}))|  f_0(y)dy \leq C' \int   y^2 \varphi(y|\mu_0,\sigma_0^2)dy<+\infty.
\]
Analogous considerations hold for $|\log(H(y|\boldsymbol{\omega^*}))|$.
Hence  the first condition in  \eqref{entropy2} is satisfied.
Using (i) and  the fact that $KL( \varphi(\cdot|{\mu_1,\sigma_{1}^2}), \varphi(\cdot|{\mu_2,\sigma_{2}^2}))<+\infty$,
 it is easy to obtain also that $KL(h(\cdot|\boldsymbol{\omega^*}),f_0)<+\infty$.
\end{example}

\begin{example}[Mixture of normals]
Following the same line of the previous example, one can treat the case in which
\[
%\begin{split}
 h(y|\boldsymbol{\omega})=\sum_{m=1}^M\omega_{m} \varphi(y|{\mu_m,\sigma_m^2})  \quad \text{and} \quad
 f_0(y)=\sum_{i=1}^K   \omega_{0,i}   \varphi(y | {\mu_{0,i},\sigma_{0,i}^{2}}), \\
%\end{split}
\]
with possibly $K>M$. Here some of the true parameters $(\mu_{0,i},\sigma_{0,i}^{2})$, $i=1,\dots K$,
can be equal to some of the parameters of the models $(\mu_m,\sigma_m^2)$, $m=1,\dots M$.
Also in this case it is easy to see that the
assumptions of Theorem  \ref{corollario} are satisfied for any $\omega^*$ in the interior of $\Delta_M$.
\end{example}

The following example shows that the assumptions can be checked also in the cases the set of models is incomplete and the true distribution has heavy tails.

\begin{example}[Heavy tails]
Consider the case in which
\[
h(y|\boldsymbol{\omega})=\sum_{m=1}^M\omega_{m} \varphi(y|{\mu_m,\sigma_m^2})
  \quad \text{and} \quad
f_0(y)=\sum_{i=1}^K  \omega_{0,i}   \mathcal{T}_{\mu_{0,i},\sigma_{0,i},\nu}(y),
\]
where $\mathcal{T}_{\mu,\sigma,\nu}$ is a $t$-distribution with location, scale and degrees of
freedom paramters $\mu,\sigma$  and $\nu$ respectively.
Since  $f_0(y)\sim Cy^{-\nu-1}$ as $|y| \to +\infty$, arguing as in the previous example, it is easy to
see that  \eqref{entropy2} is satisfied whenever $\nu>2$. In this case $f_{0}\in KL(\Pi^{*})$.
\end{example}

\subsection{Consistency results for Markovian observations}\label{ConMarkov}
 
Let  $\mathcal{F}$ be the set 
 of all possible transition densities (with respect to Lebesgue measure) on the sample space $\mathcal{Y}\subset \RE$.
%In this section we shall consider calibration consistency for Markov processes.
As in the i.i.d. case, we assume that set of combined models is given, 
but the distribution of the current observation $y_t$ given all the past observations, 
satisfies the Markovian property: $F_{mt}(y_t)=F_{m}(y_t|y_{t-1})$, where
\[
F_{m}(y_t|y_{t-1})=\int_{(-\infty,y_t]} f_m(y|y_{t-1})dy, \quad m=1,\dots,M,\, t=1,\dots,T
\]
for a  given set of  transition densities $\{ f_m:\CY \times \CY \to \RE ;m=1,\dots,M\}$, subset of  $\mathcal{F}$. Hence
$ H_t(y_t|\boldsymbol{\omega})=H(y_t|y_{t-1},\boldsymbol{\omega})$
where $H(y_t|y_{t-1},\boldsymbol{\omega})=\sum_{m=1}^M \omega_m F_{m}(y_t|y_{t-1})$ and $\boldsymbol{\omega}=(\omega_1,\dots,\omega_M)$.
%Note that, as in the i.i.d. case, we assume that set of combined models is given.

%Following  \cite{TangGhosal07}, 
In order to introduce the definition of weak consistency for Markovian observations, 
let us
consider  a real valued ergodic Markov process $(y_0,y_1,\dots)$ with true transition density $f_0(y|x)$  belonging to  $\mathcal{F}$ and stationary distribution $\pi$
and write $P^\infty_{f_0}$ for the distribution of the infinite
sequence $(y_0, y_1, \dots)$. %For simplicity we assume that $y_0$ is fixed. 
If $\Pi^*$ is a prior over  $\mathcal{F}$, 
we shall abbreviate the posterior $\Pi^*(\cdot|y_0,\cdots,y_n)$ by $\Pi^*_n(\cdot)$ and by a.s., 
we shall mean a.s. with respect to $P^\infty_{f_0}$. 
A sequence of posterior distributions 
$\Pi^*_n(\cdot)$ $(n \geq 1)$ is said to be weakly consistent at $f_0$ if for every weak neighbourhood $B$ of $f_0$, it follows that  $\Pi_n(B^c) \to 0$ a.s. as $n \to +\infty$. 
%Clearly, as in the i.i.d. case,  the concept depends on the definition of such neighbourhoods, so different types of consistency can be considered.
Following  \cite{TangGhosal07},  the weak topology is induced by 
the sub-base of neighbourhoods (of a point $f_0$) 
$
 A_{g,\eps}(f_0):=\Big \{ f: \int \Big | \int g(y) f(y|x) dy -  \int g(y) f_0(y|x) dy \Big|  \lambda(dx) \leq \eps \Big \}
$
where $\lambda$ is a fixed probability distribution and $g$ varies among bounded continuous function on $\CY$. A standard choice for $\lambda$ is $\pi$. 
%For more details see \cite{TangGhosal07}. 

We need also to introduce an adequate generalization of the Kullback-Leibler property: 
a Kullback-Leibler neighbourhood  of size $\varepsilon$ of a transition density $f_0 \in\mathcal{F}$ with stationary distribution $\pi$ is defined as
$$
\mathcal{K}_{\varepsilon}(f_{0}):=\left\{f\in\mathcal{F} : \int \int f_0(y|x)\log\left(\frac{f_0(y|x)}{f(y|x)}\right)dy \pi(dx)  \leq\varepsilon\right\},
$$
and the Kullback-Leibler property holds at $f_{0}\in\mathcal{F}$, for short $f_{0}\in KL(\Pi^{*})$, if $\Pi^{*}(\mathcal{K}_{\varepsilon}(f_{0}))>0$ for all $\varepsilon>0$.

Differently from  the i.i.d. case, the Kullback-Leibler property alone is not enough to ensure weak consistency
%One needs additional  assumptions. 
and needs to be complemented with  Corollary 4.1 in \cite{TangGhosal07} (see proof of Theorem \ref{teoconsMarkov3}) . 

Recall that we are  dealing with a prior  $\Pi^*$ induced via the map
$G\mapsto f_{G}(y|x) :=\int_\Theta K(y| x,\boldsymbol{\theta})G(d\boldsymbol{\theta})$,
where
%$\Theta$ is  the mixing parameter space, $K(y| x,\boldsymbol{\theta})$ is a transition kernel and
$G$ has a Dirichlet process prior with base measure $G_0$ and
$
K(y|x, \boldsymbol{\theta})=b^*_{\mu,\nu}(H(y|x,\boldsymbol{\omega}))h(y|x,\boldsymbol{\omega})
$
with $\boldsymbol{\theta}=(\boldsymbol{\theta}_{p},\boldsymbol{\theta}_{c})$. 
%, 

\begin{theorem}\label{teoconsMarkov3}
Assume that the functions $f_m(\cdot|\cdot)$ are continuous on $\mathcal{Y} \times \mathcal{Y}$
and that, for every compact set $C \subset \mathcal{Y} \times \mathcal{Y}$,
%\begin{equation}\label{lowerboundh0tris}
$ \min_{m=1,\dots,M} \inf_{(y,x) \in C} f_m (y|x)>0$.
%\end{equation}
Assume also  that for every compact set $C' \subset \mathcal{Y}$ and every  $\eta \in (0,1)$ there is $L>0$ such that 
\begin{equation}\label{upperboundF}
 \max_{m=1,\dots,M} \sup_{x \in C'} \int_{[-L,L]^c}f_m (y|x)dy \leq \eta.
\end{equation}
Let $f_0(y|x)=\int b_{\mu,\nu}^*(H(y|x,\boldsymbol{\omega})) h(y|x,\boldsymbol{\omega}) G^*(d\mu d\nu d \boldsymbol{\omega})$,
$G^*$ being a probability measure on $(0,1) \times \RE^+ \times \Delta_M$, such that $supp(G^*) \subset supp(G_0)$
and the marginal distribution $G_{0}(d\mu d\nu \times \Delta_M)=\int_{\Delta_M} G_{0}(d\mu d\nu d \boldsymbol{\omega})$ has compact support on  $(0,1) \times \RE^+$.
If  $ R_-(y|x) \leq H(y|x,\boldsymbol{\omega})(1-H(y|x,\boldsymbol{\omega})) \leq R_+(y|x)$ 
 and $r_-(y|x) \leq h(y|x, \boldsymbol{\omega})$
 for every $y,x$ and every 
 $\boldsymbol{\omega}$ with
 \begin{equation}\label{eq:condmarkov1}
 \int\int \Big [|\log(R_-(y|x))\Big|+\Big|\log \Big (\frac{R_+(y|x)}{r_-(y|x)} \Big )\Big| + \log(f_0(y|x))  \Big]f_0(y|x) dy \pi(dx)<+\infty,
 \end{equation}
then
$f_{0}\in KL(\Pi^{*})$ and
the posterior is weakly consistent at $f_0$.
%, i.e. 
%for any   $A_{g,\eps}(f_0):=\{f: \int|\int g(y)[f(y|x)-f_0(y|x)]dy| \lambda(dx)\leq \eps \}$ where $g$ is a bounded continuous functions, 
%$\Pi^*_n(A_{g,\eps}(f_0)^c) \to 0$ a.s. for any neighbourhood  $A_{g,\eps}(f_0)$.
\end{theorem}

The proof of the previous theorem is given in the Supplementary Materials \ref{appendix_consistencyl}.

\begin{example}[Mixture of autoregressive processes]\label{exMix}
Consider the case in which the models are normal autoregressive processes of the first
order, i.e. 
\[
 h(y_t|y_{t-1},\boldsymbol{\omega})=\sum_{m=1}^M\omega_{m} \varphi(y_t |{\mu_m +\phi_m y_{t-1}  ,\sigma_m^2}) 
\]
where  $\varphi(\cdot|{\mu,\sigma^2})$ is the pdf of a normal distribution of mean $\mu$ and variance $\sigma^2$ and $M \geq 2$.
Assume that the data are generated from the following
dynamic mixture model (mAR) 
\[
y_t \sim p \mathcal{SN} (\mu_1 + \phi_1 y_{t-1}, \sigma_1,\rho_1) + (1-p)\mathcal{SN} (\mu_2 + \phi_2 y_{t-1}, \sigma_2,\rho_2) 
\]
 with $\rho_i \in \{-1,1\}$, 
where $\mathcal{SN} (\mu, \sigma,\rho)$ denotes a skew-normal  (see \cite{RSSB:RSSB391}) with $\mu$, $\sigma$ and  $\rho$ the 
location, scale and asymmetry parameters.
Recalling that the density of a skew normal distribution is 
$2\Phi(y|{\mu,\sigma^2}) \varphi(y|{\mu,\sigma^2})$ for  $\rho=1$  and 
 $2(1-\Phi(y|{\mu,\sigma^2})) \varphi(y|{\mu,\sigma^2})$ for $\rho=-1$, 
one can write 
$f_0(y|x)=\int b_{\mu,\nu}^*(H(y|x,\boldsymbol{\omega})) h(y|x,\boldsymbol{\omega}) G^*(d\mu d\nu d \boldsymbol{\omega})$
for a suitable $G^*$. For example if $M=3$, $\rho_1=1$ and $\rho_2=-1$, then 
$G^*(d\mu d\nu d \boldsymbol{\omega})= p \delta_{(1,0,0)}(d \boldsymbol{\omega} )\delta_{2/3,3}(d\mu d\nu) +   
(1-p) \delta_{(0,1,0)}(d \boldsymbol{\omega}) \delta_{1/3,3}(d\mu d\nu)$.
With a little  bit of effort, it is possible to show that  all the assumptions of Theorem \ref{teoconsMarkov3} are satisfied. Details are given in Supplementary Materials \ref{appendix_consistencyl}.
\end{example}

\section{Simulation examples}  \label{sec:simulation}
The consistency results given in the previous section imply uniformity of the PITs only in the limit, when the number of observations goes to infinity. On a finite sample, different sets of externally provided models can have different forecasting performances and one is left with the issue of studying the finite sample properties of the probabilistic calibration method. This motivates the following simulation studies.

\subsection{Multimodality and heavy tails}  \label{sec:multimodality}
We assume that a combined predictive distribution can be obtained from
the two normal predictive distributions with different location and
equal scale parameters, $\cN(-1,1)$ and $\cN(2,1)$, where
$\cN(\mu,\sigma^{2})$ denotes the normal distribution with location
$\mu$ and scale $\sigma$. Let us denote with $\varphi(x|\mu,\sigma^2)$ and
$\Phi(x|\mu,\sigma^2)$ the pdf and cdf respectively of a
$\cN(\mu,\sigma^2)$. We compare the noncalibrated linear pool (NC)
$f(y|\btheta)=\omega \varphi(y|-1,1)+(1-\omega) \varphi(y|2,1)$,
where $\btheta=\omega$ and the infinite beta mixture model (BMC). 
The model weights in the linear pooling are estimated by using the 
recursive log score, see e.g. \cite{JoreMitchellVahey2010}.

In the first set of experiments, we assume that the data are generated by the
following mixture of the three normal distributions
\begin{equation*}
y_t \overset{i.i.d.}{\sim}
p_1 \cN(-2,0.25) + p_2 \cN(0,0.25) + p_3 \cN(2,0.25),
\quad t = 1, \ldots, 1000,
\end{equation*}
where $p = (p_1, p_2, p_3) \in \Delta_3$. The inability of the beta calibration to calibrate this liner pooling is documented in the results reported 
in the Supplementary Materials \ref{appendix_simulation}. In the same Supplementary Materials, a finite beta mixture model is showed to outperform the beta calibration model. We apply the BMC model and obtain the calibrated PITs given in Fig. \ref{fig:inf.mix.sim1}. To investigate the sensitivity of the posterior quantities to the choice of the hyper\-parameters, we combine and calibrate the cdfs, on the same dataset, using two different values of the dispersion
parameter, $\psi=1$ and $\psi=5$.
\begin{figure}[t]
\begin{center}
\begin{tabular}{ccc}
\multicolumn{3}{c}{\scriptsize{$\psi=1$}}\vspace{-2pt}\\
\includegraphics[width=3.8cm]{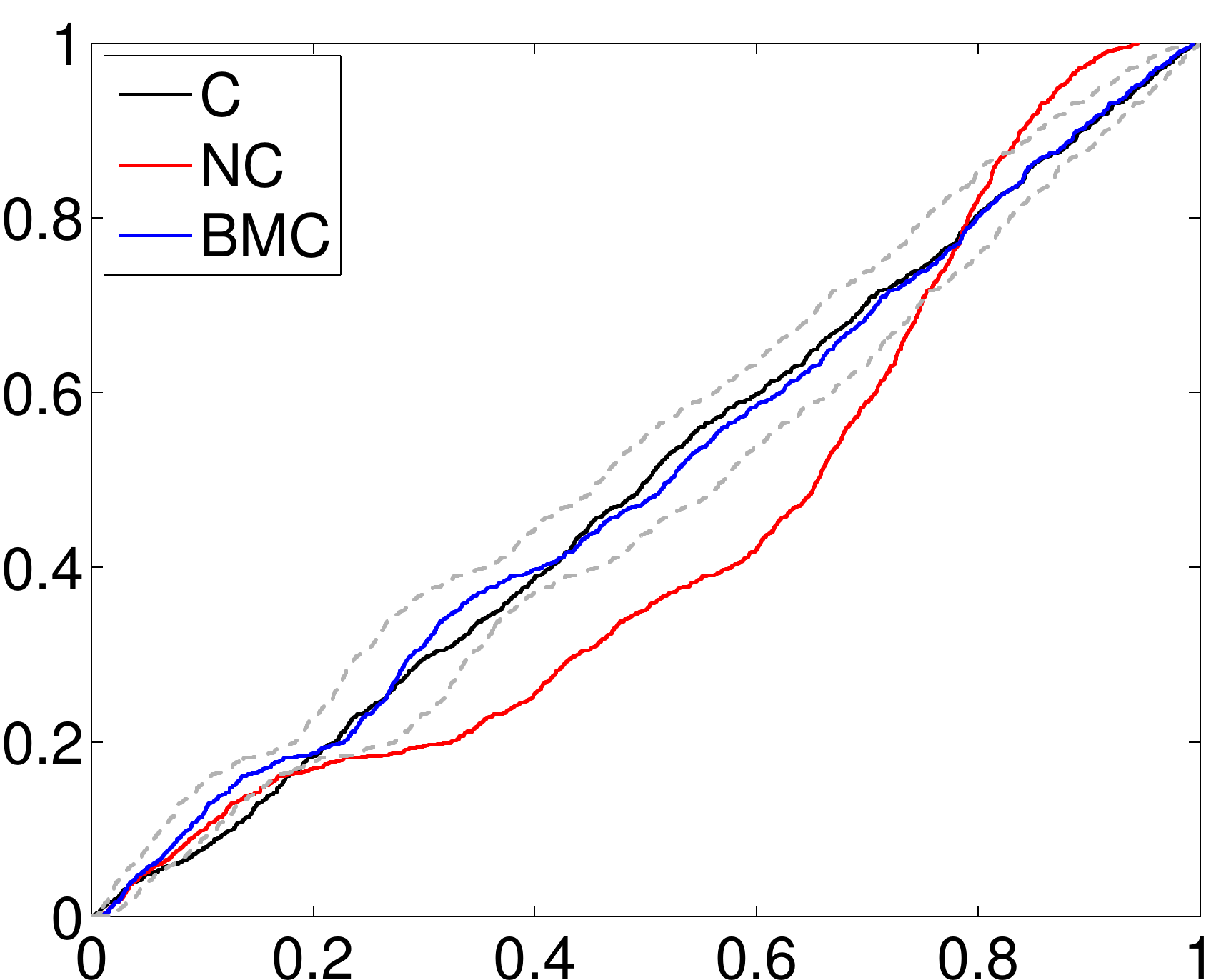}&\includegraphics[width=3.8cm]{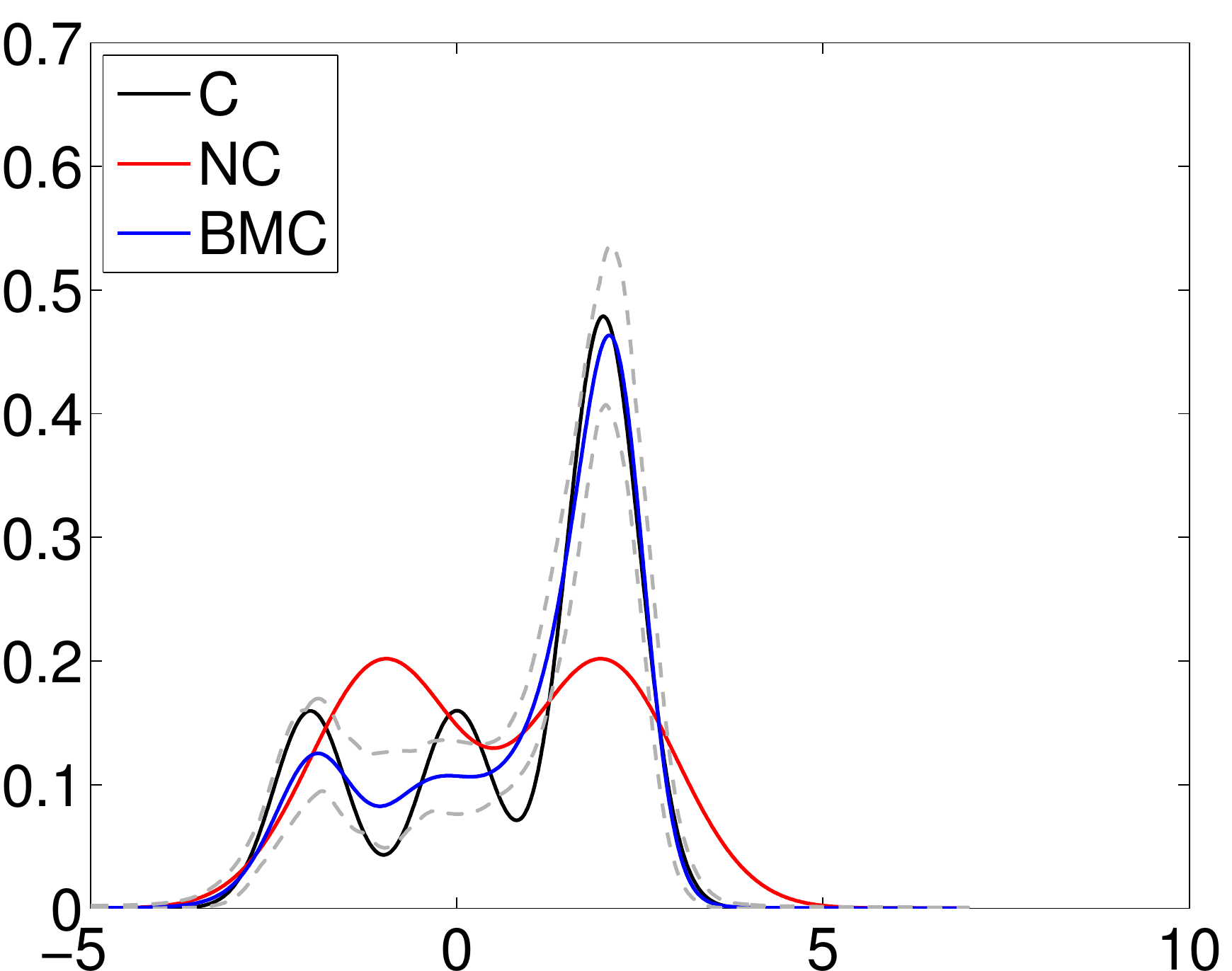}&\includegraphics[width=3.8cm]{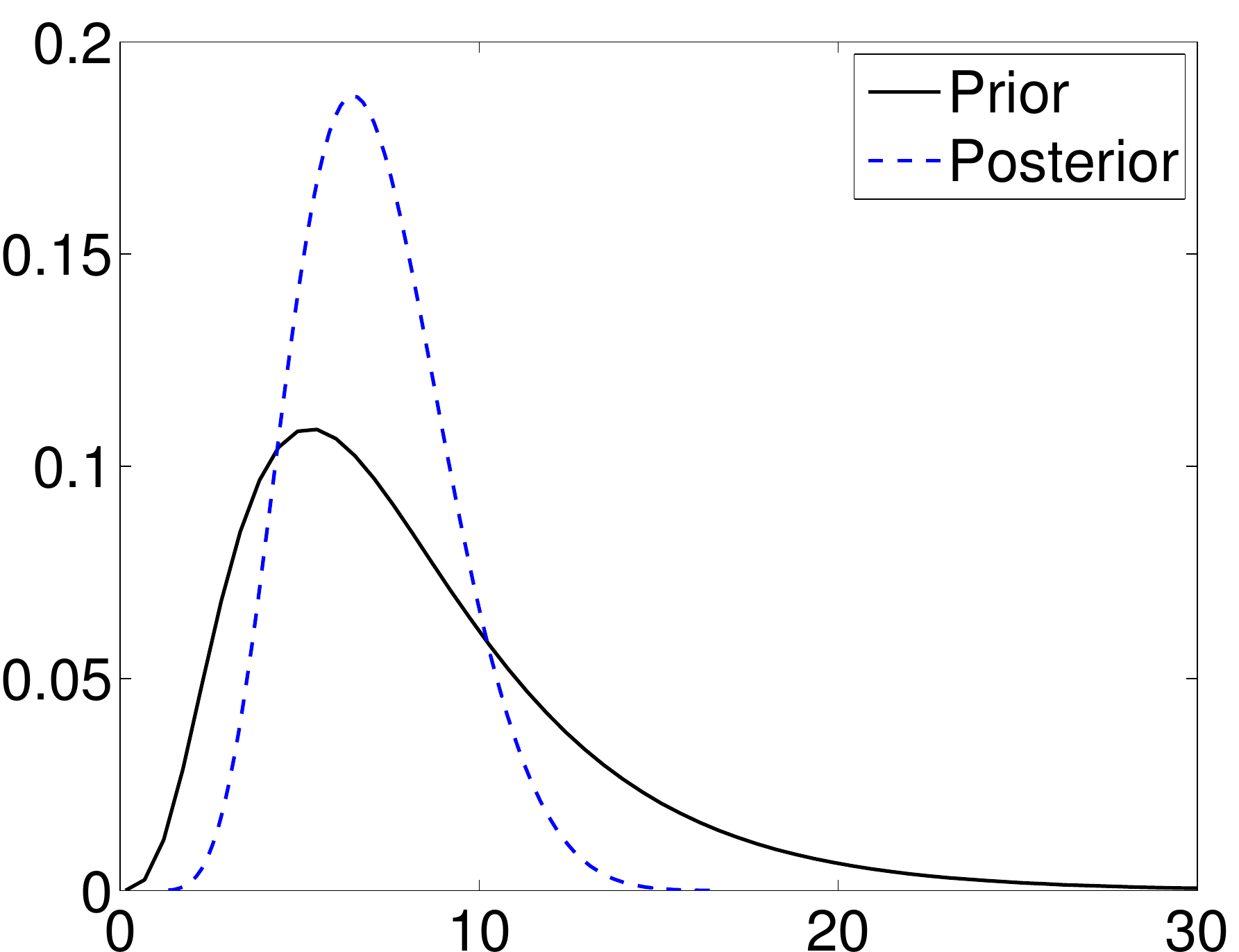}\\
\multicolumn{3}{c}{\scriptsize{$\psi=5$}}\vspace{-2pt}\\
\includegraphics[width=3.8cm]{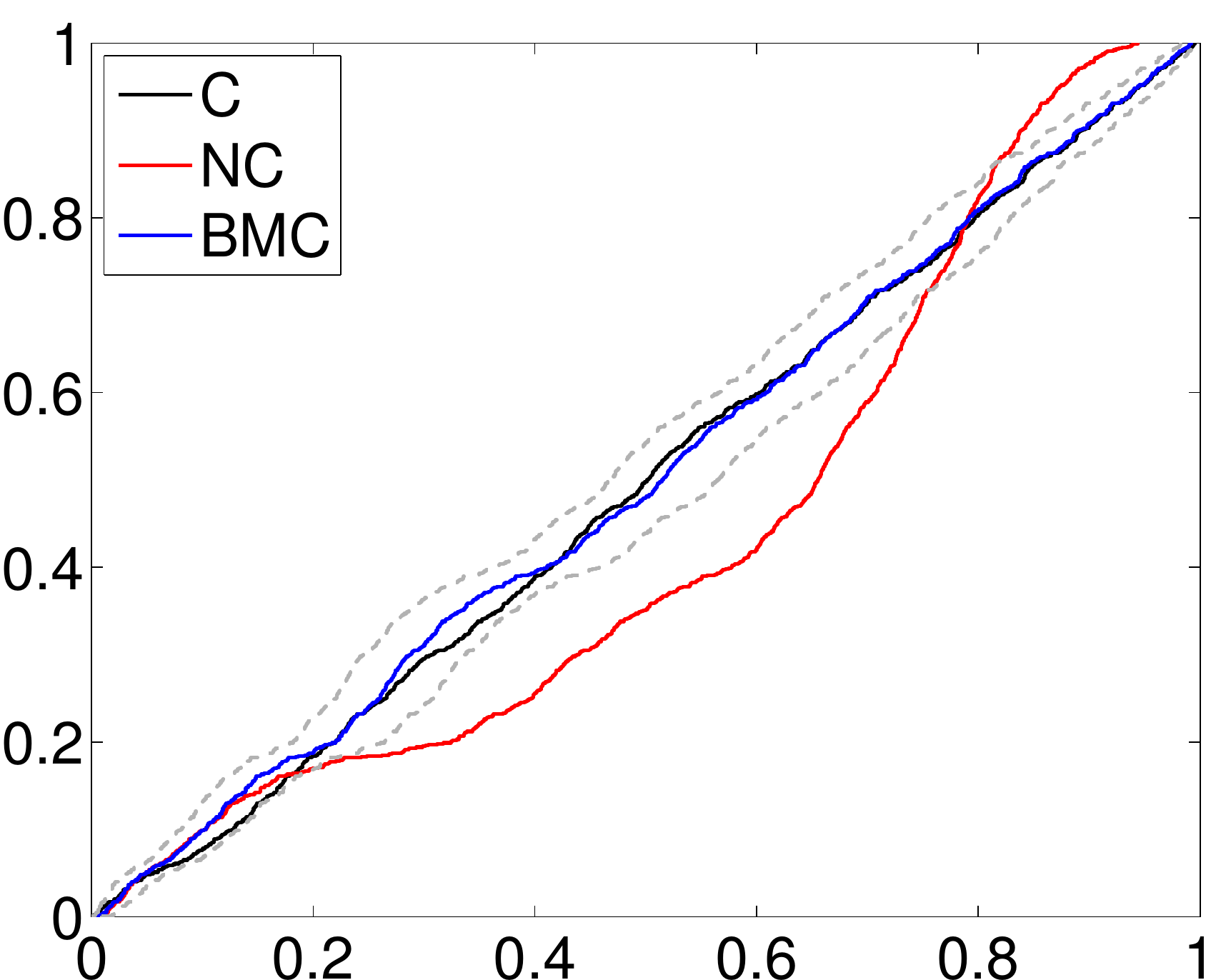}&\includegraphics[width=3.8cm]{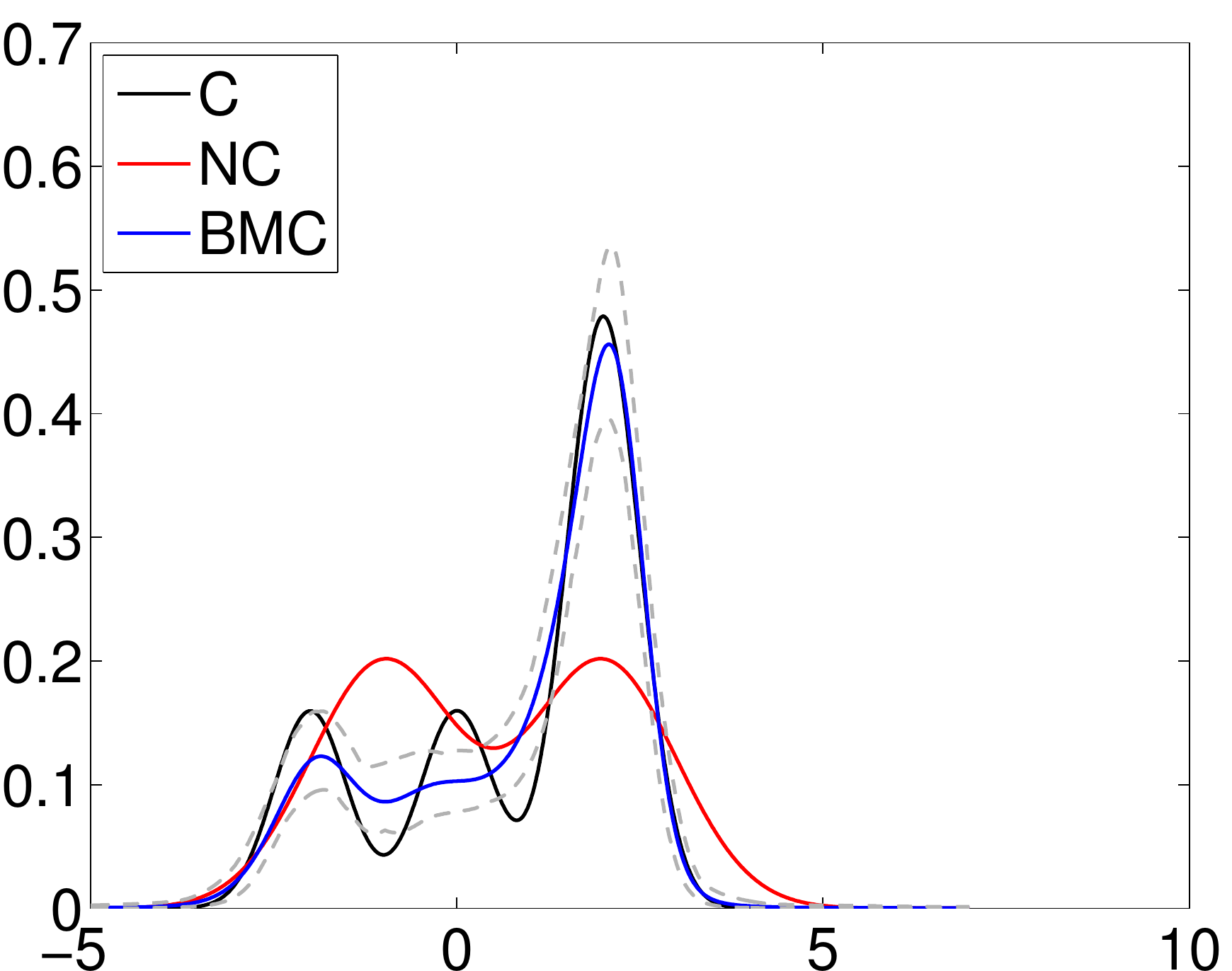}&\includegraphics[width=3.8cm]{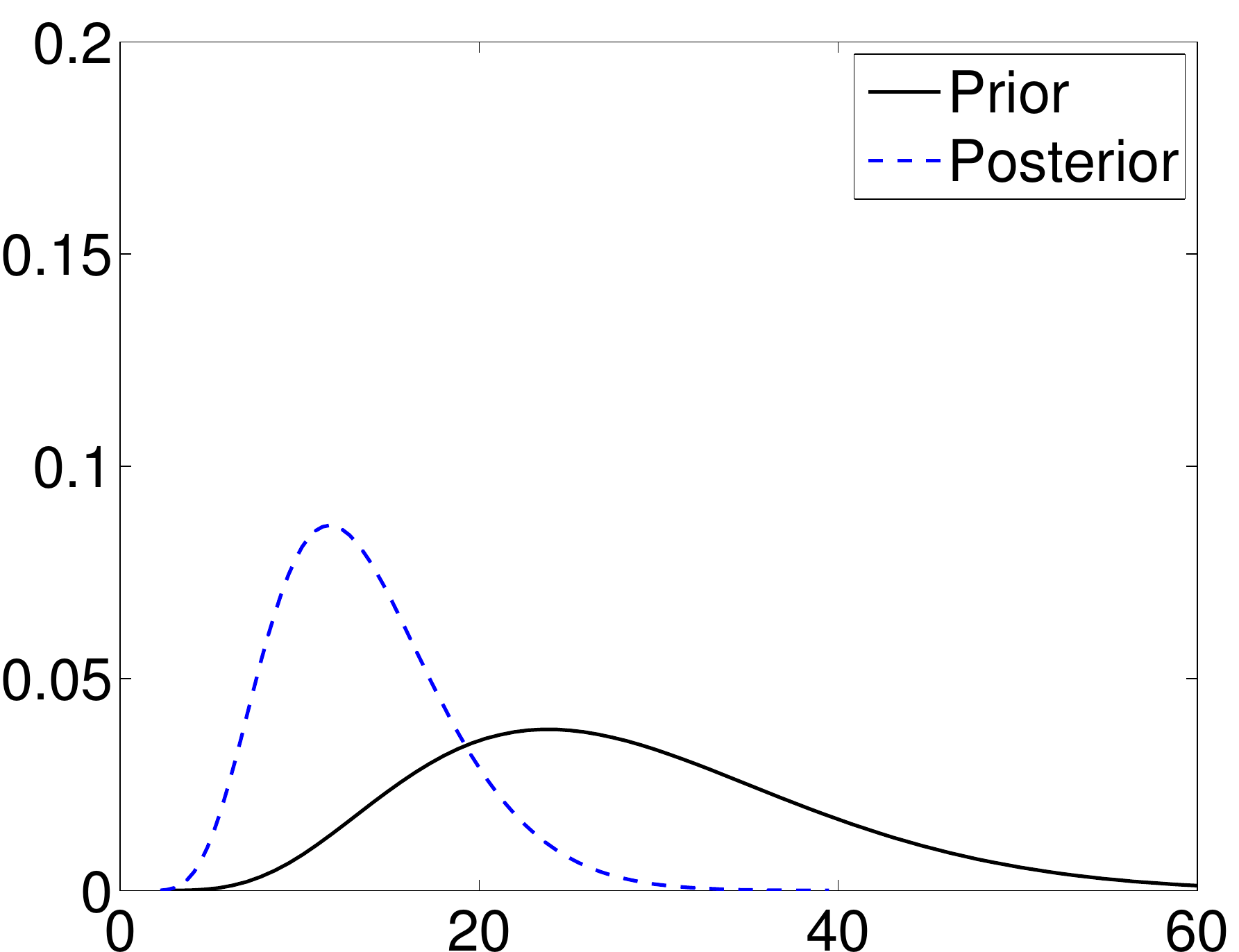}
\end{tabular}
\caption{Infinite beta mixture calibrated (BMC), calibrated (C) and
  non calibrated (NC) combinations for a dataset of 1,000 samples from
  $p_1 \cN(-2,0.25)+p_2 \cN(0,0.25)+p_3 \cN(2,0.25)$ with
  $\bp=(1/5,1/5,3/5)$. PITs cdf (left graph) and calibrated pdf (middle
  graph) of the combination models C (black), NC (red) and BMC (blue)
  and BMC 99\% HPD (gray). Prior (black) and posterior (blue) number
  of components of the random BMC model (right
  graph).}  \label{fig:inf.mix.sim1}
\end{center}
\end{figure}

The left charts of Figure \ref{fig:inf.mix.sim1} report the PITs of the
average infinite beta mixture calibration (BMC) model and their
99\% credibility intervals obtained from 1,000 MCMC samples after
convergence. The PITs of the calibrated model (black lines) belong to the 
credibility interval of the BMC, thus the resulting predictive cdf is well 
calibrated.  We should notice that the credibility intervals are usually larger 
than the one obtained using a beta mixture with a fixed number of
components.  In fact the calibrated density accounts for both
calibration parameter uncertainty and also for the uncertainty about
the number of mixture components.  A comparison between the top- and
bottom-right chart also shows that an increase in the value of the
dispersion parameter usually increases the uncertainty.

The credibility intervals (gray lines) obtained with the infinite beta
mixture calibration model, see Figure \ref{fig:inf.mix.sim1}, always
contain the PITs (first column) and the predictive density function
(second column) of the correct model.  The infinite BMC seems
particularly accurate in the tails (last column).  We also note that
the uncertainty of the number components in the infinite beta mixture
implies a wider high probability density region (HPD), see gray lines
in \ref{fig:inf.mix.sim1}, than that given by the finite beta mixture
calibration, see third panel in \ref{fig:mot1.heavy}.  The
prior and posterior distributions of the number of mixture components
in BMC are given in the right graph in Figure \ref{fig:inf.mix.sim1}.
The posterior density is more concentrated than the prior, suggesting
that data are informative on the number of calibration components.

In the second set of experiments, we assume that the data are generated by the
following mixture of $t$-distributions, i.e.
\begin{equation*}
y_t \overset{i.i.d.}{\sim}
\frac{1}{2} \cT(-1,1,6)+\frac{1}{2} \cT(2,1,6),\quad t=1,\ldots,2000,
\end{equation*}
where $\cT(\mu,\sigma,\nu)$ denotes a $t$-distribution with
location, scale and degrees of freedom parameters $\mu$, $\sigma$ and
$\nu$ respectively. The results in Supplementary Materials \ref{appendix_simulation} show 
the difficulties of the beta model in achieving well calibrated PITs. 
In this set of experiments we assume $\psi$ is unknown. The results of the infinite mixture calibration are given in Figure \ref{fig:inf.mix.sim2}.
\begin{figure}[t]
\begin{center}
\begin{tabular}{ccc}
\includegraphics[width=3.8cm]{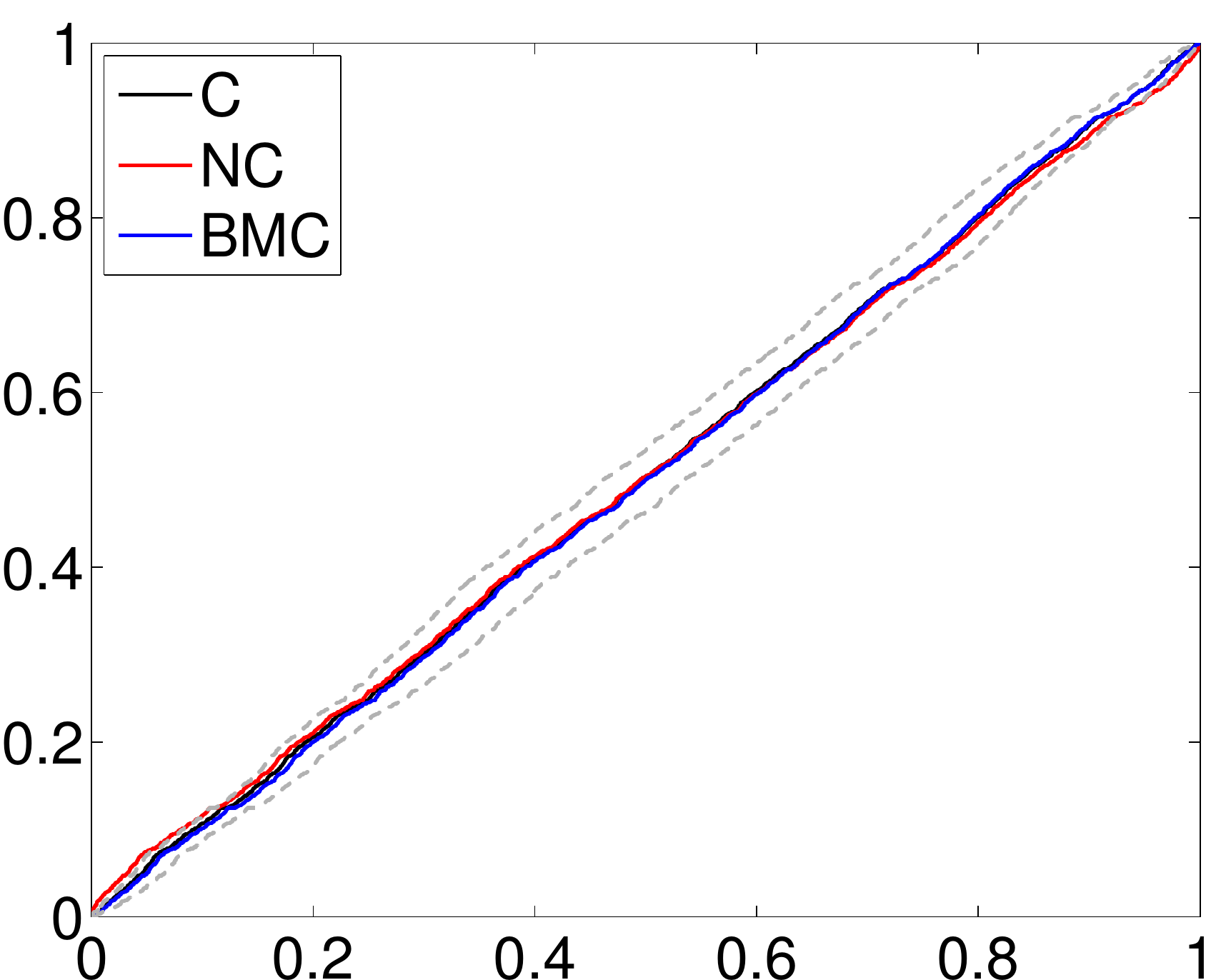} &
\includegraphics[width=3.8cm]{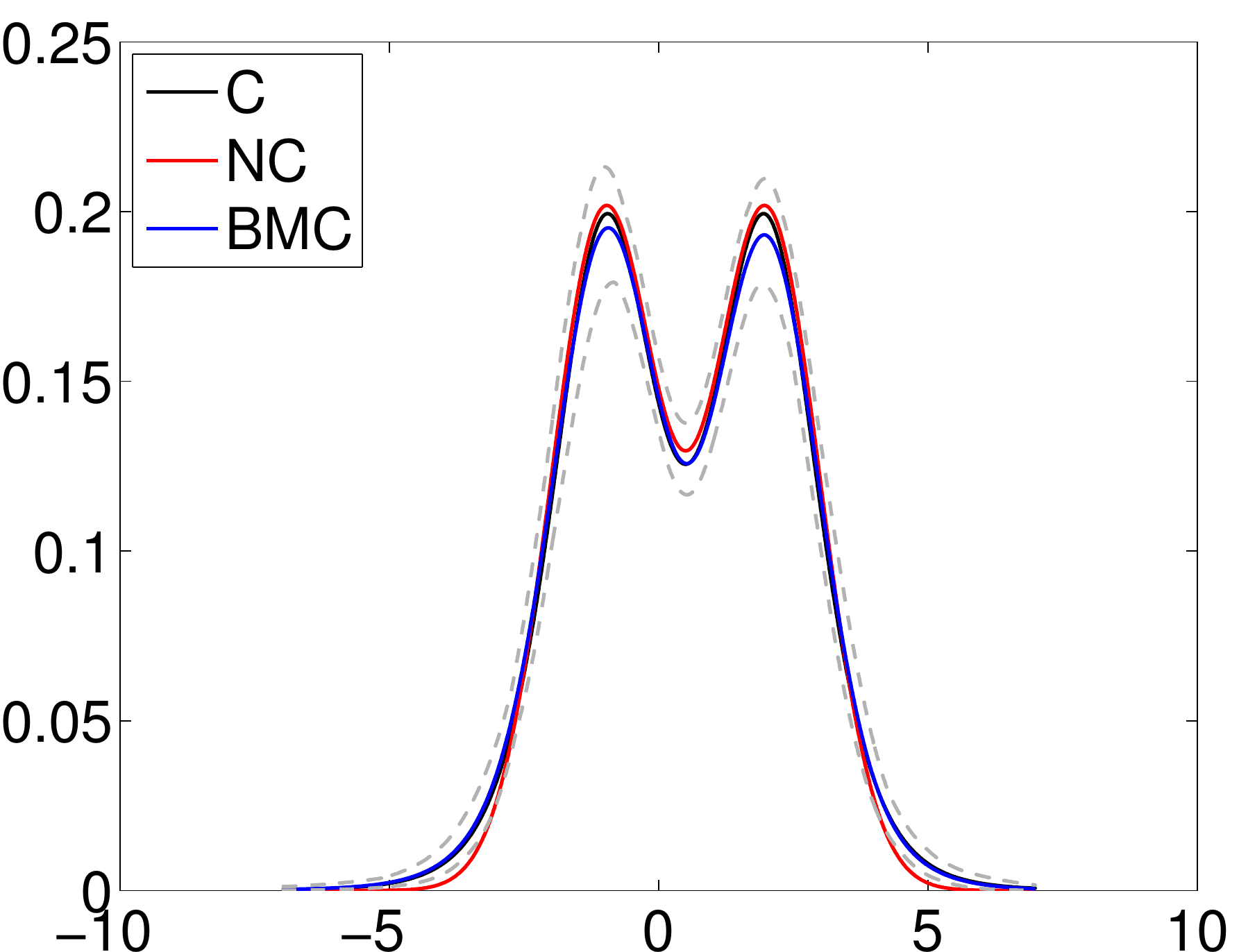} &
\includegraphics[width=3.8cm]{Nclust1-eps-converted-to.pdf}\\
\includegraphics[width=3.8cm]{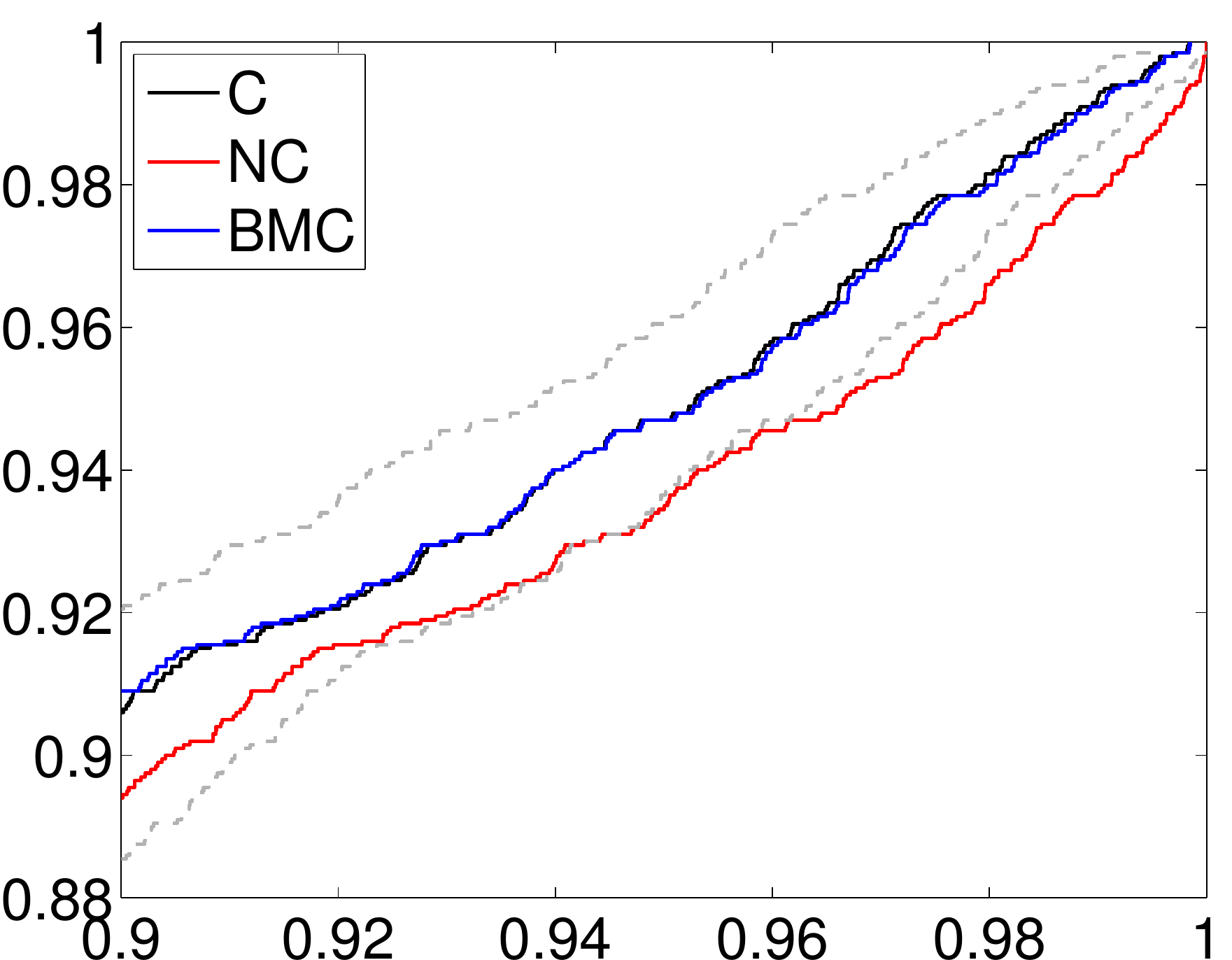}&
\includegraphics[width=3.8cm]{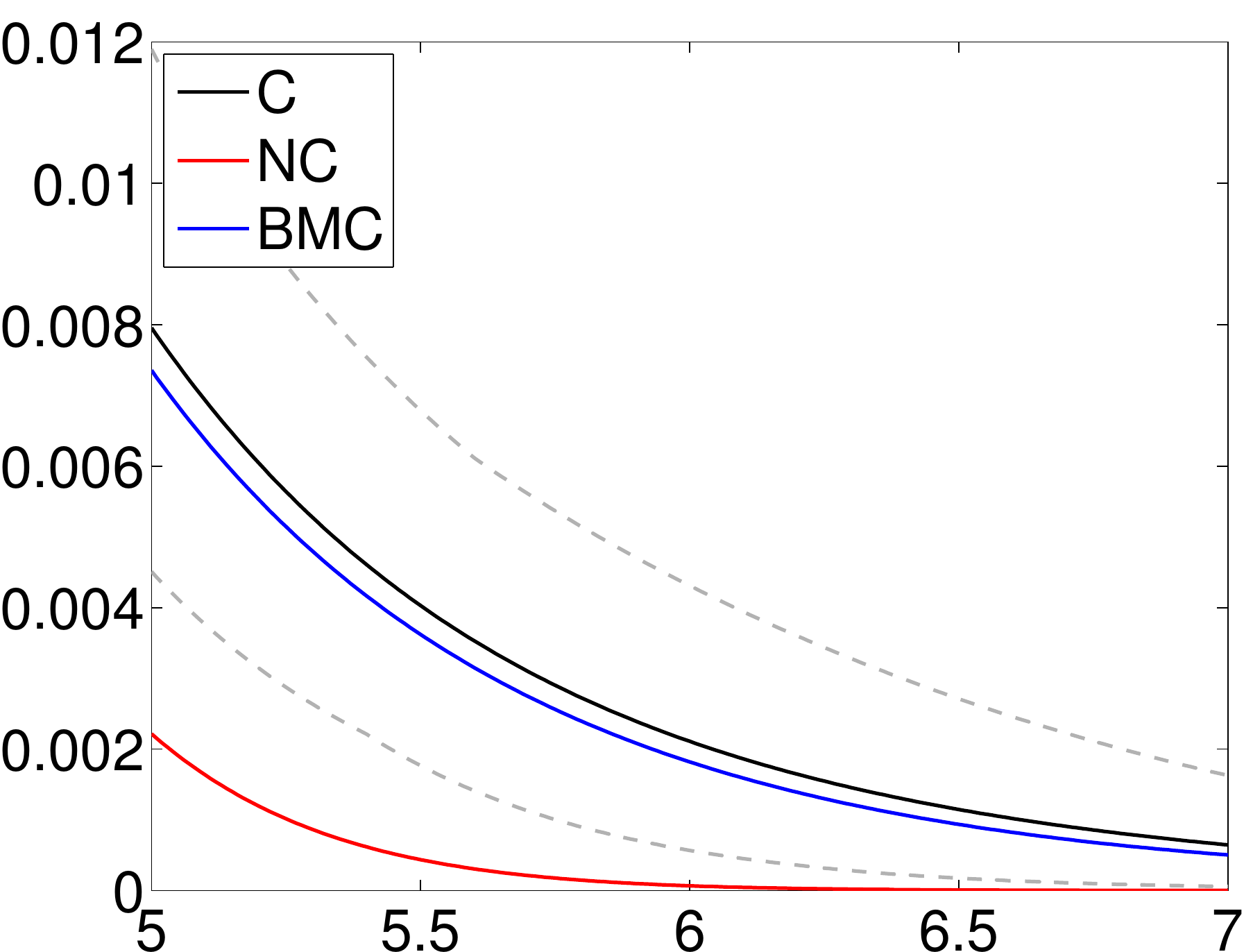}& 
\\		
\end{tabular}
\caption{Infinite beta mixture calibrated (BMC), calibrated (C) and
  non calibrated (NC) combinations for a dataset of 2,000 samples from
  $1/2 \cT(-1,1,6)+1/2 \cT(2,1,6)$. PITs cdf (left graphs) and
  calibrated pdf (middle graphs) of the combination models C (black),
  NC (red) and BMC (blue) and BMC 99\% HPD (gray). Prior (black) and
  posterior (blue) number of components of the random BMC model
  (right graph).}  \label{fig:inf.mix.sim2}
\end{center}
\end{figure}

Both cdf (first column) and pdf (second column) indicate that the Bayesian
BC has problems producing well-calibrated predictions.  The Bayesian
nonparametric calibration BMC, on the contrary, produces well-calibrated
densities, in particular on the tails; see also the 99\%
credibility intervals.  We note that the posterior distribution of the
number of clusters is more concentrated than the prior, thus there is
learning from the data on the number of mixture components.  Finally,
our experiments changing the dispersion parameter indicate no
substantial changes in the posterior density over different
hyper\-parameter values.

\subsection{Dependent observations}  \label{sec:dep}
We assume that a predictive density is obtained from the combination of two independent normal autoregressive processes of the first order, $y_{t}=\mu_1 + \phi y_{t-1} +\varepsilon_{1t}$ and $y_{t}=\mu_2 + \phi y_{t-1} +\varepsilon_{2t}$ with $\varepsilon_{it}\sim \mathcal{N}(0,\sigma_i^2)$ i.i.d., where $\mu_1=-1$, $\mu_2=2$ and $\sigma_1=\sigma_2=0.5$. Following the notation used in Example \ref{exMix} we assume the data generating process is a mixture of skew-normal autoregressive processes (mAR)
\begin{equation*}
y_t \sim \frac{1}{3}\cK(-2+\phi y_{t-1},0.5,\varrho) + \frac{2}{3} \cK(2+\phi y_{t-1},0.5,\varrho),
\end{equation*}
$t = 1, \ldots, 1000$. We compare the NC and BMC combination schemes defined in the previous section. We set the BMC dispersion parameter $\psi=0.1$ and consider two settings for the autoregressive coefficients: low persistence ($\phi=0.5$, panel (a) in Figure \ref{fig:AR}) and high persistence ($\phi=0.99$, panel (b)). For the skewness parameter, we consider one of the cases covered in Example \ref{exMix}, i.e. $\varrho=-1$ (first column in panels (a) and (b)). Moreover we show, through simulated examples, that the well-calibration property is satisfied also for other two cases: $\varrho=-3$ and $\varrho=-5$ (columns two and three in panels (a) and (b)).
\begin{figure}[h!]
\begin{center}
\begin{tabular}{ccc}
\multicolumn{3}{c}{\scriptsize{(a) Low persistence ($\phi=0.5$)}}\vspace{0pt}\\
\scriptsize{$\varrho=-1$}&\scriptsize{$\varrho=-3$}&\scriptsize{$\varrho=-5$}\vspace{0pt}\\
\includegraphics[width=3.5cm]{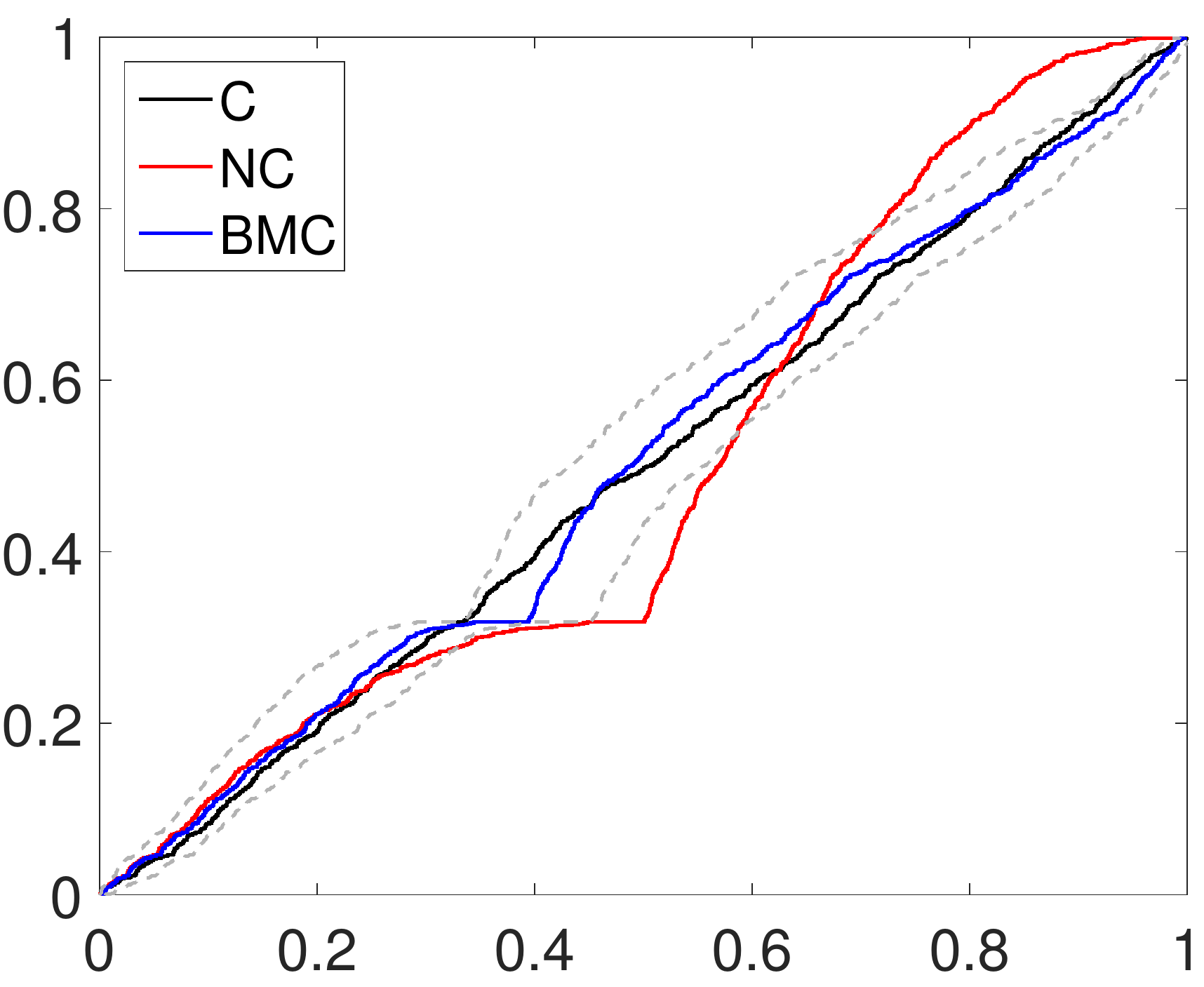} &
\includegraphics[width=3.5cm]{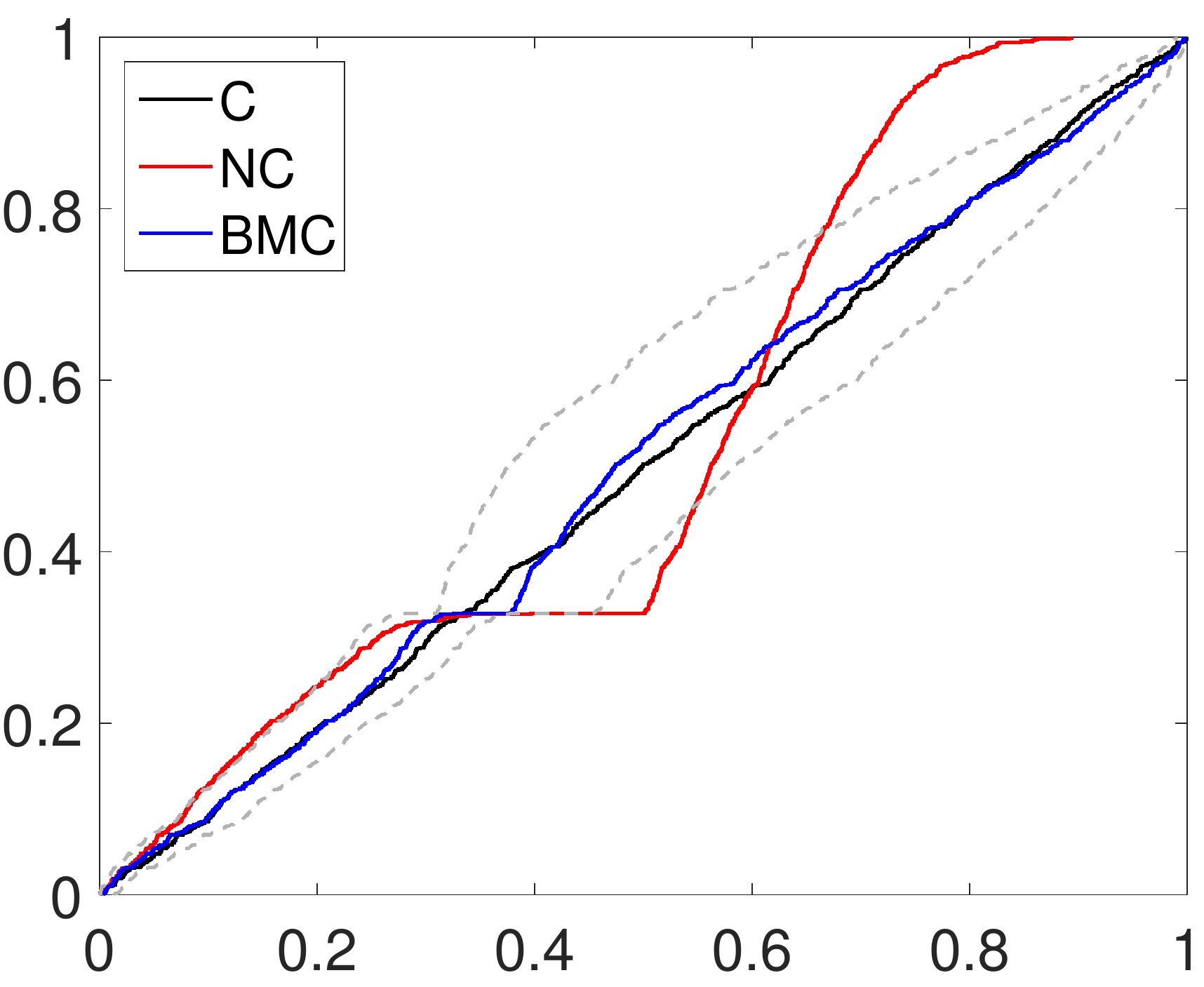} &
\includegraphics[width=3.5cm]{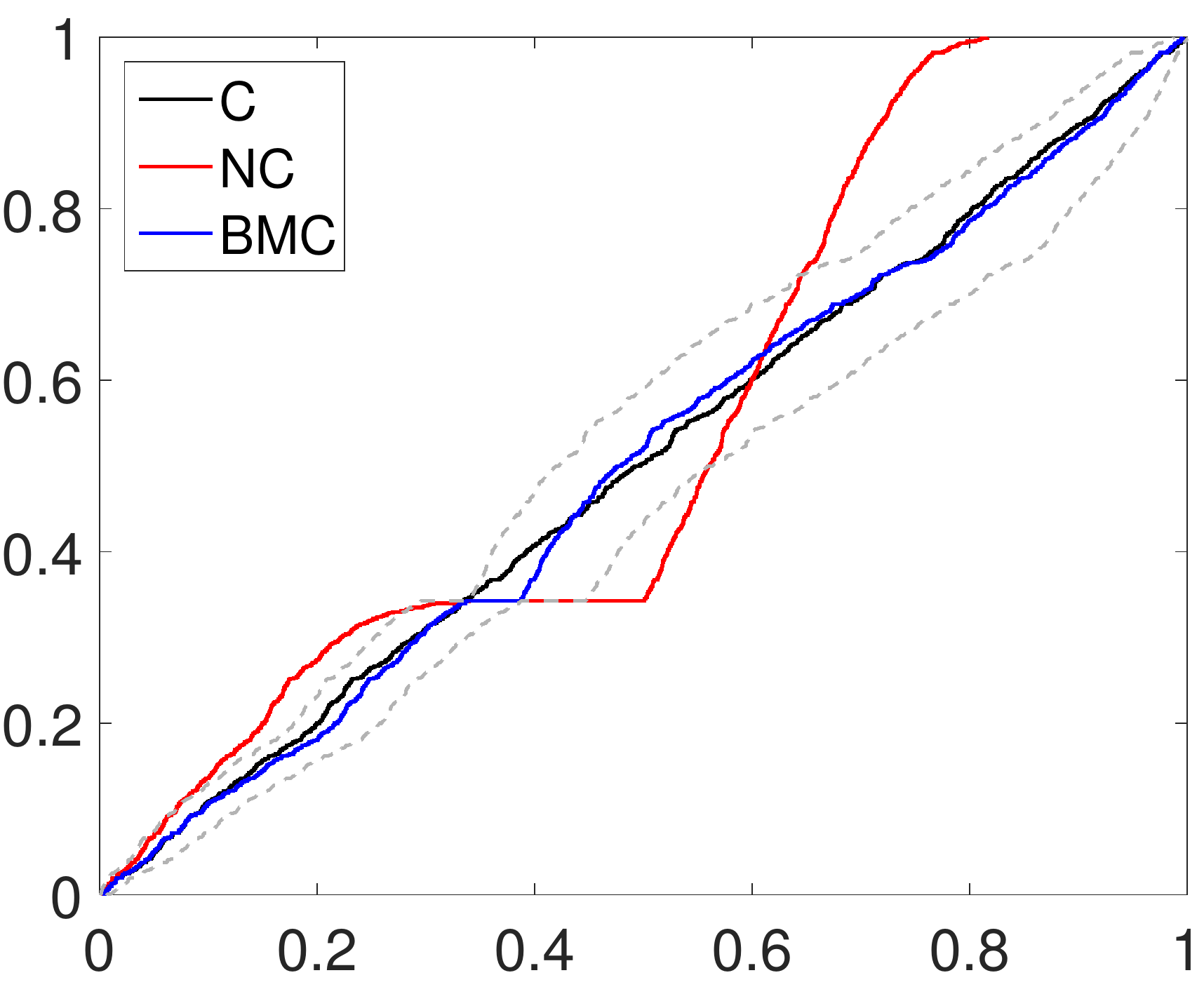}\\
\includegraphics[width=3.5cm]{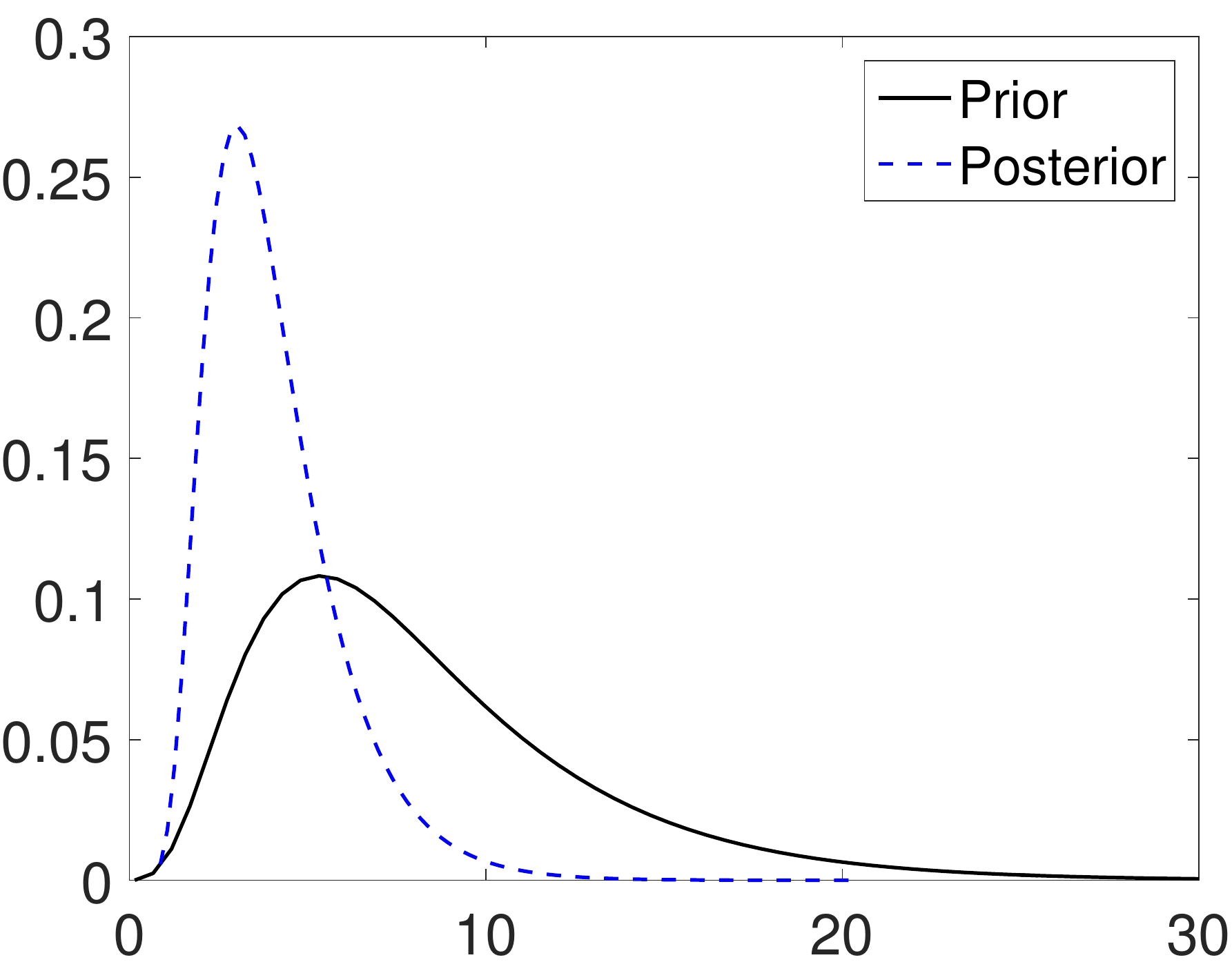} &
\includegraphics[width=3.5cm]{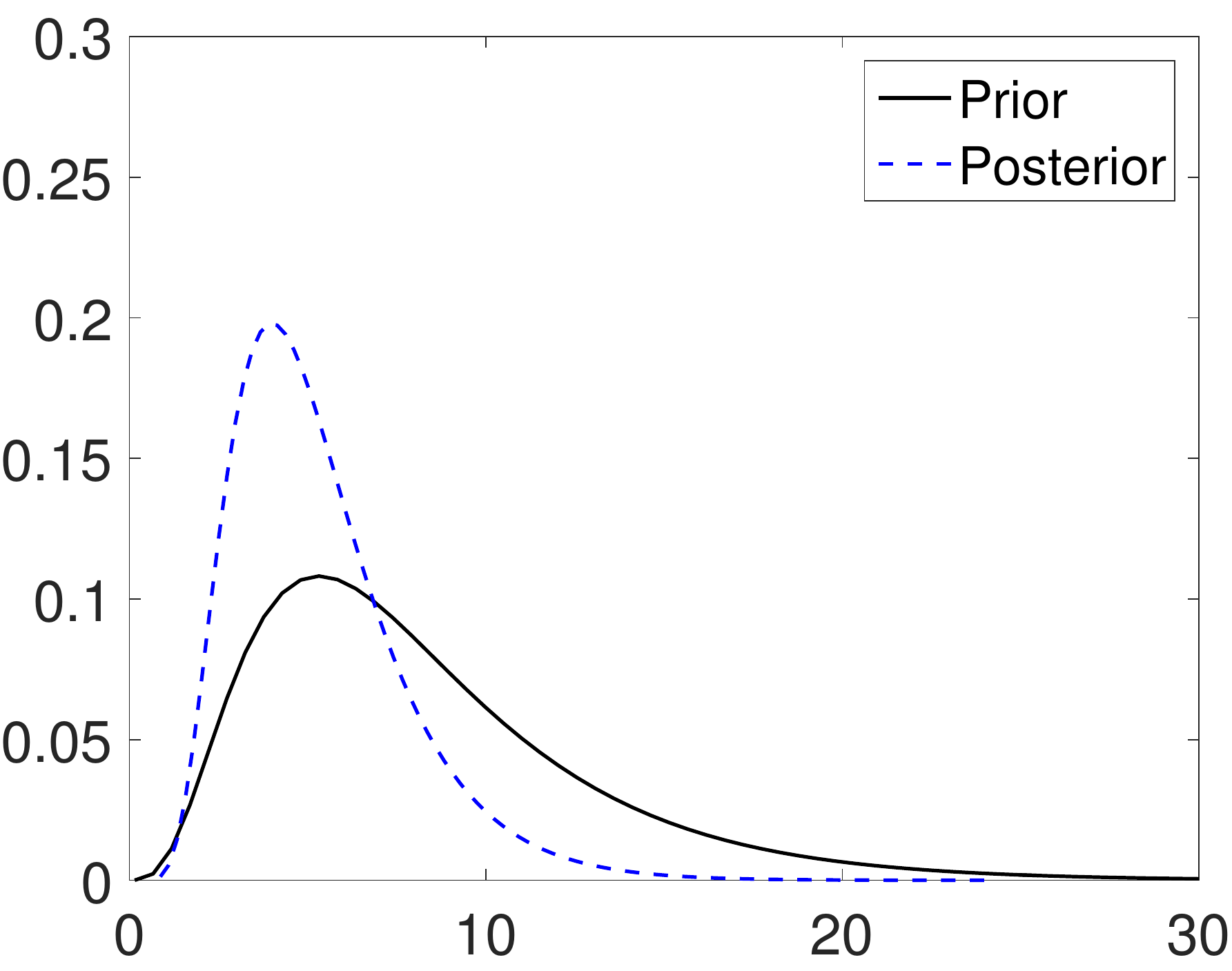} &
\includegraphics[width=3.5cm]{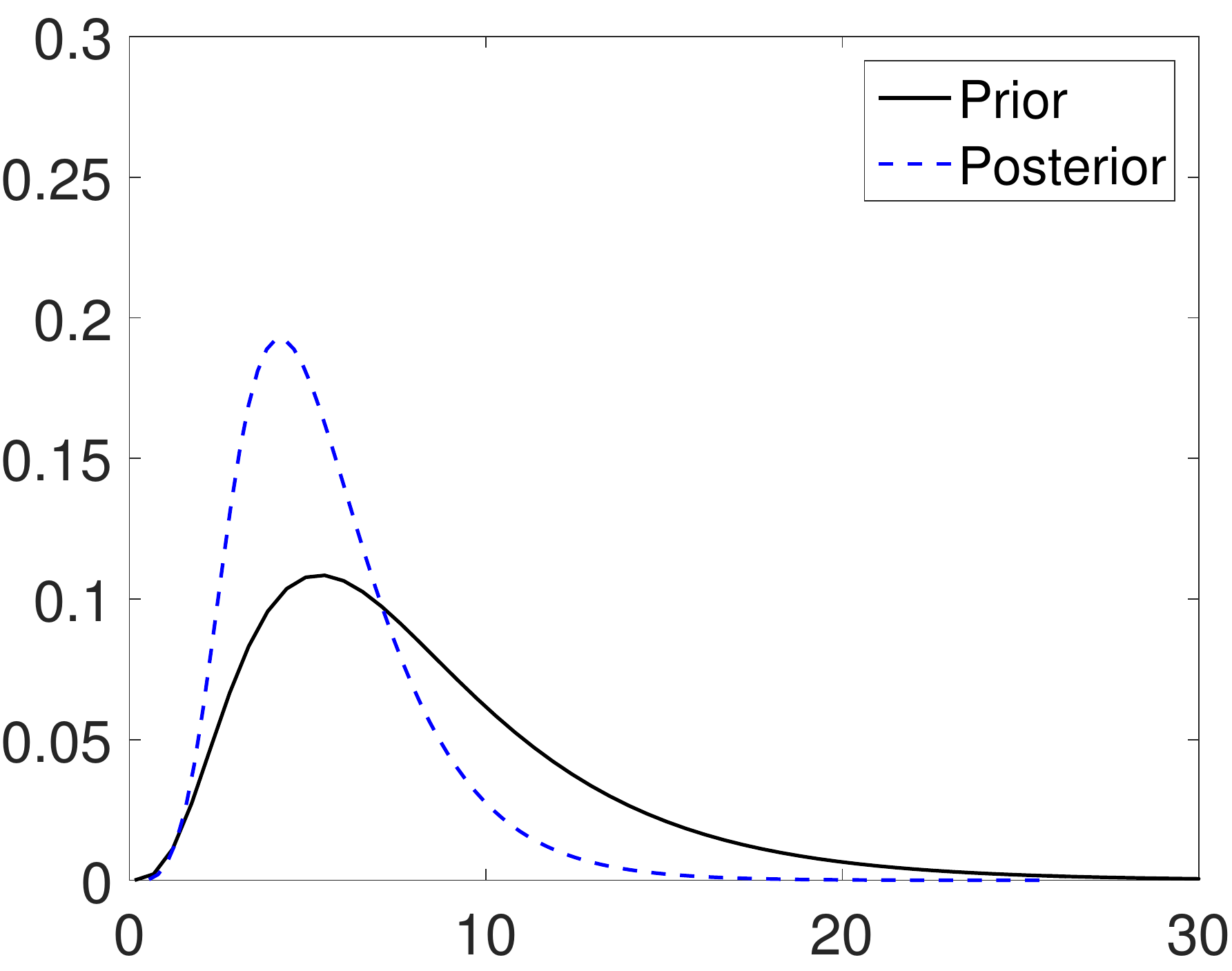} \\
%%%%%%%%%
\multicolumn{3}{c}{\scriptsize{(b) High persistence ($\phi=0.99$)}}\vspace{0pt}\\
\scriptsize{$\varrho=-1$}&\scriptsize{$\varrho=-3$}&\scriptsize{$\varrho=-5$}\vspace{0pt}\\
\includegraphics[width=3.5cm]{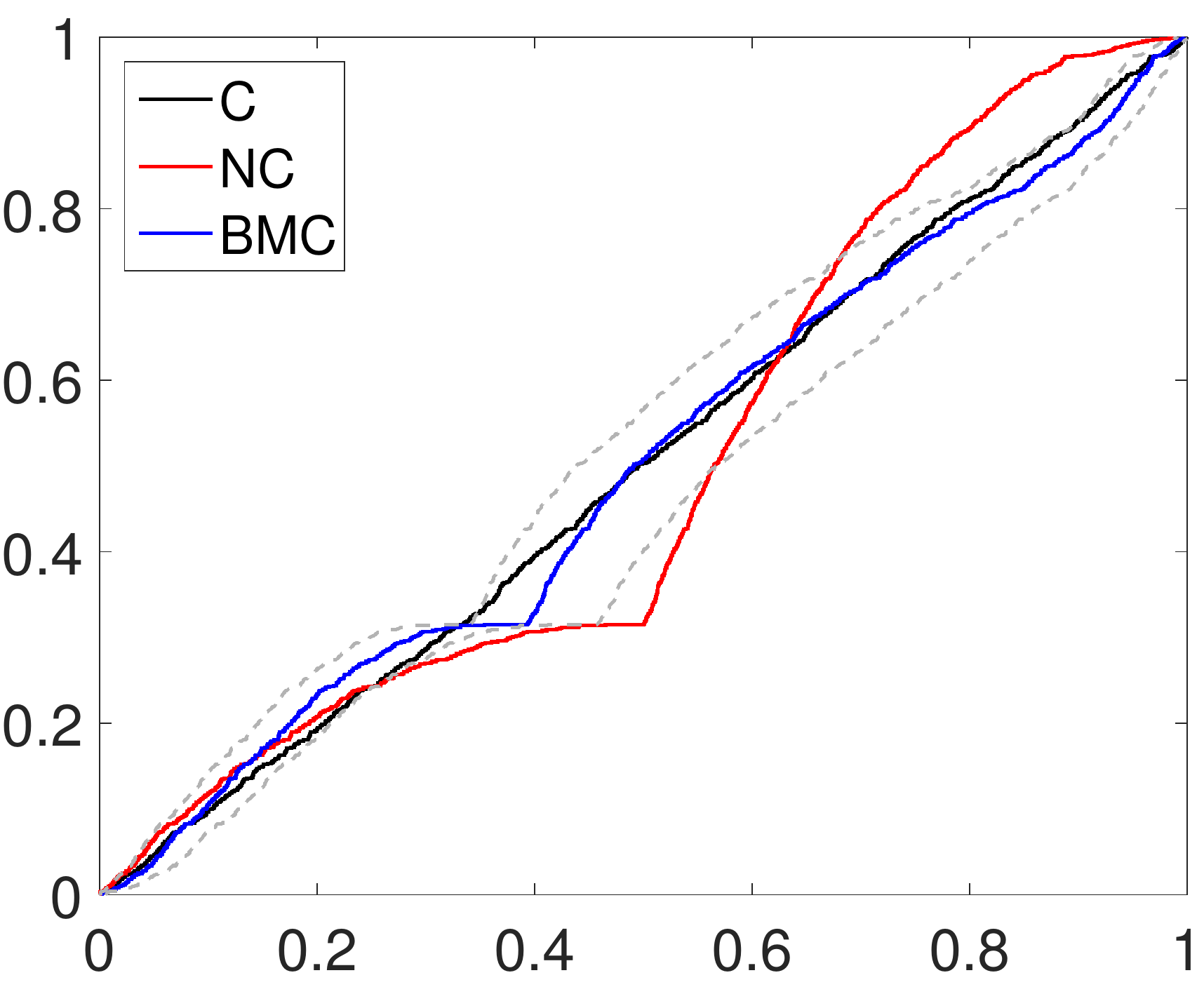} &
\includegraphics[width=3.5cm]{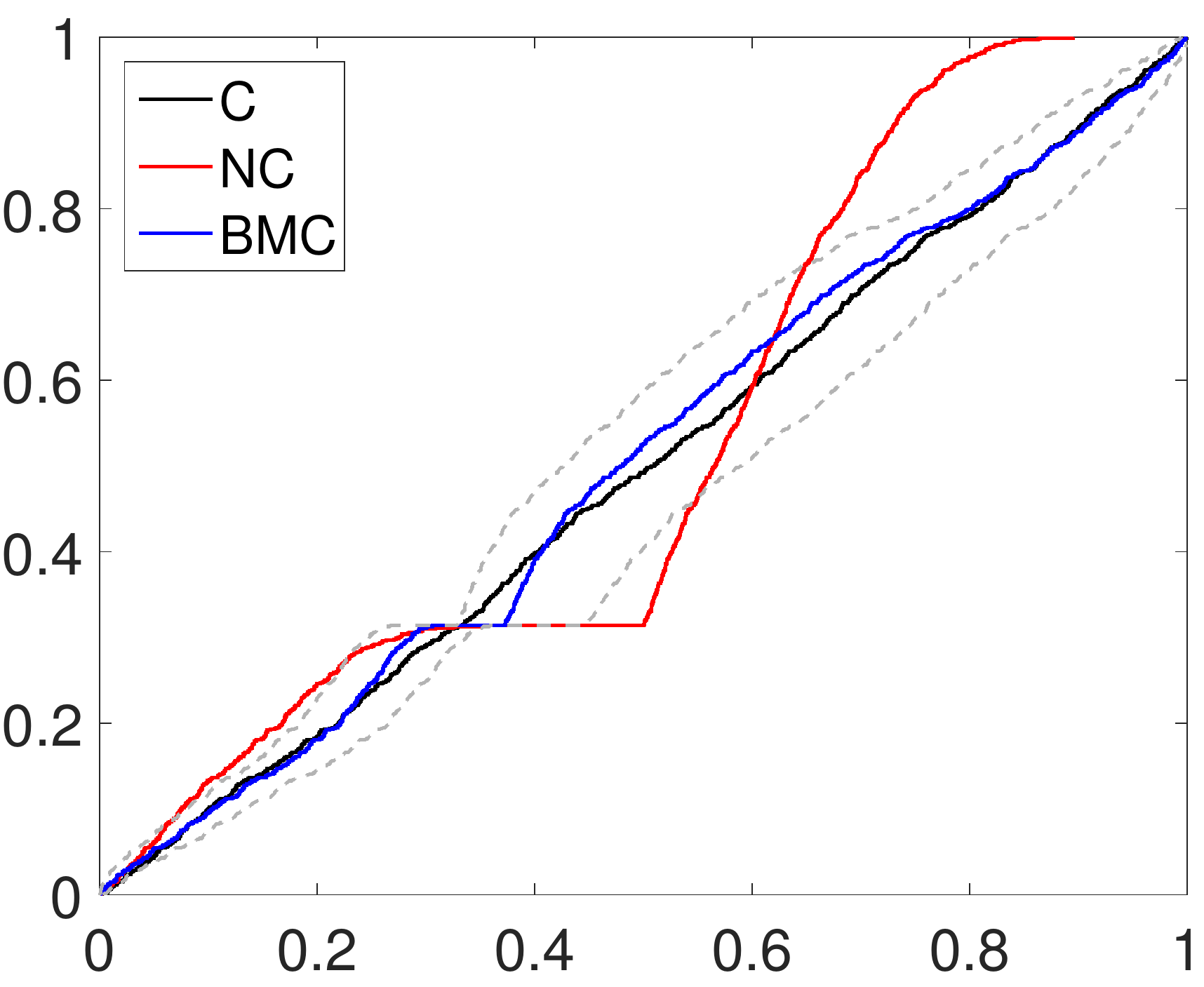} &
\includegraphics[width=3.5cm]{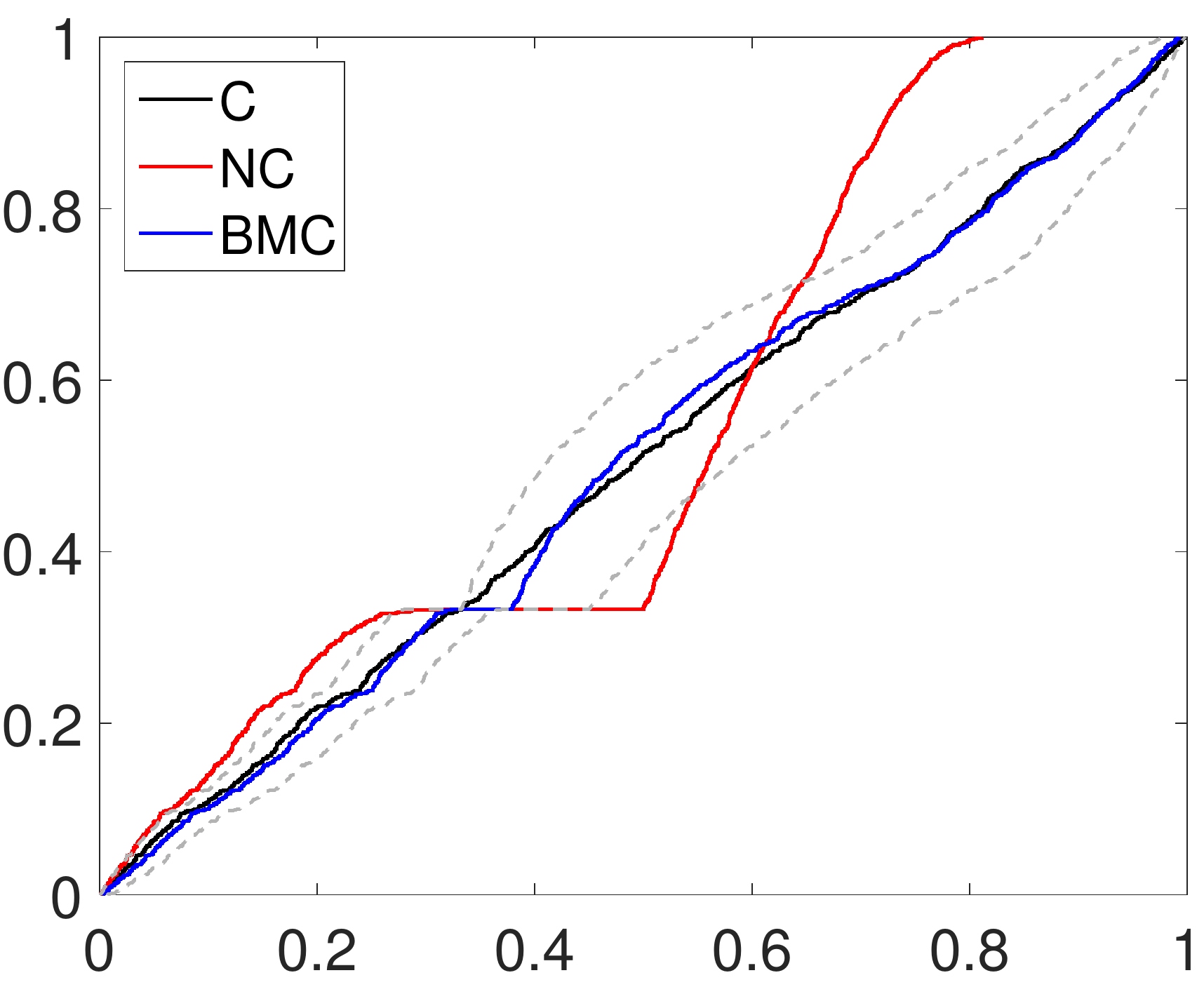} \\
\includegraphics[width=3.5cm]{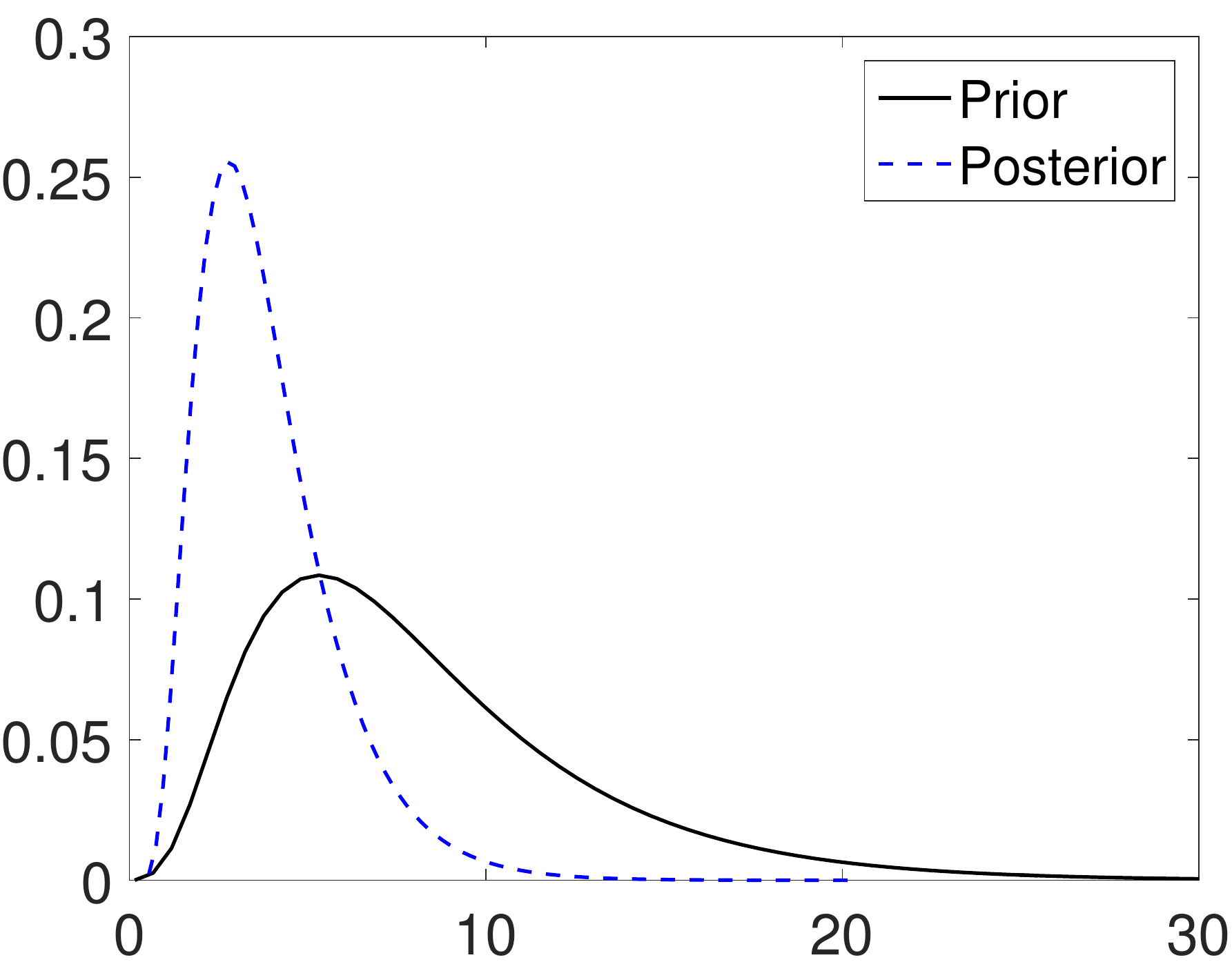} &
\includegraphics[width=3.5cm]{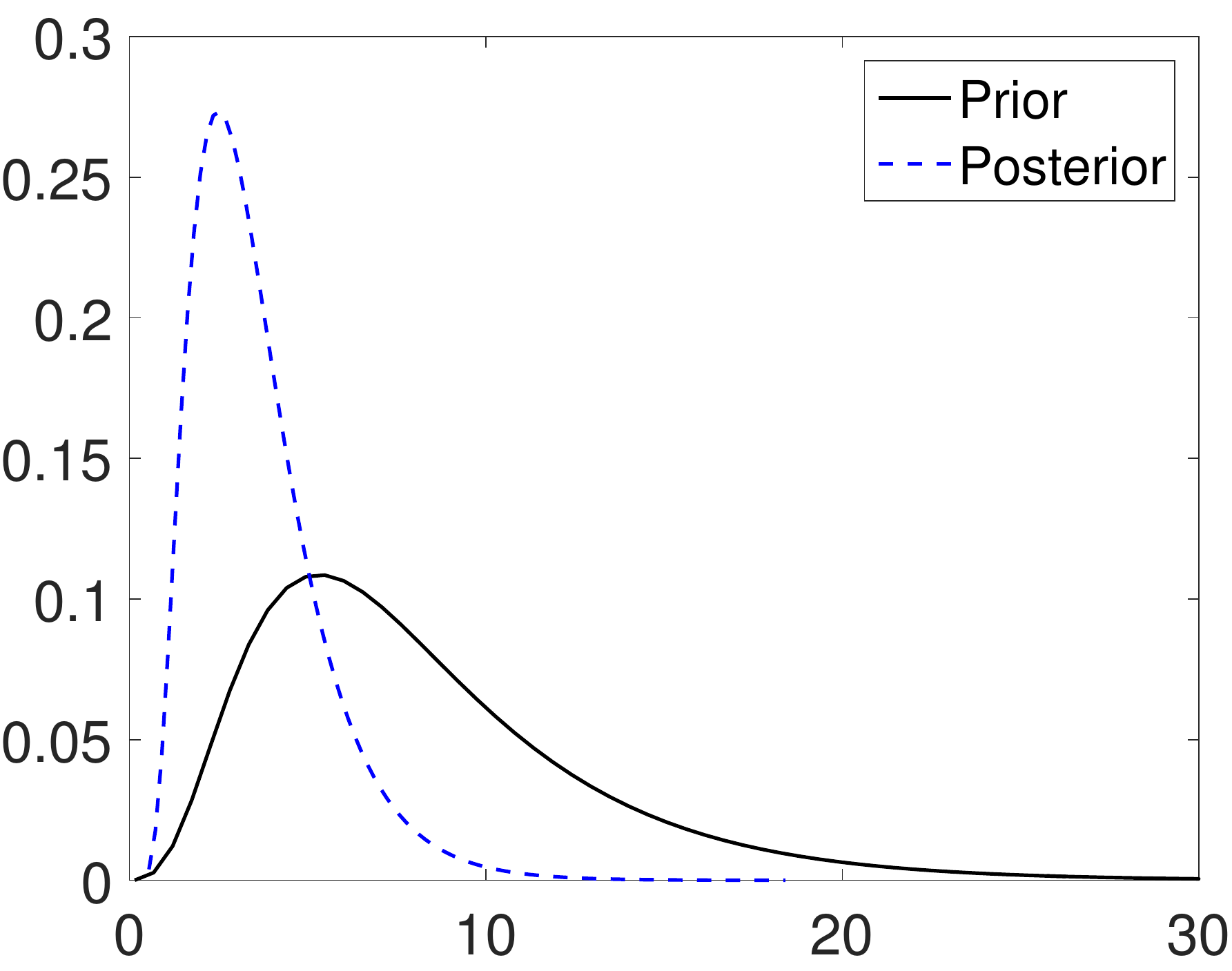} &
\includegraphics[width=3.5cm]{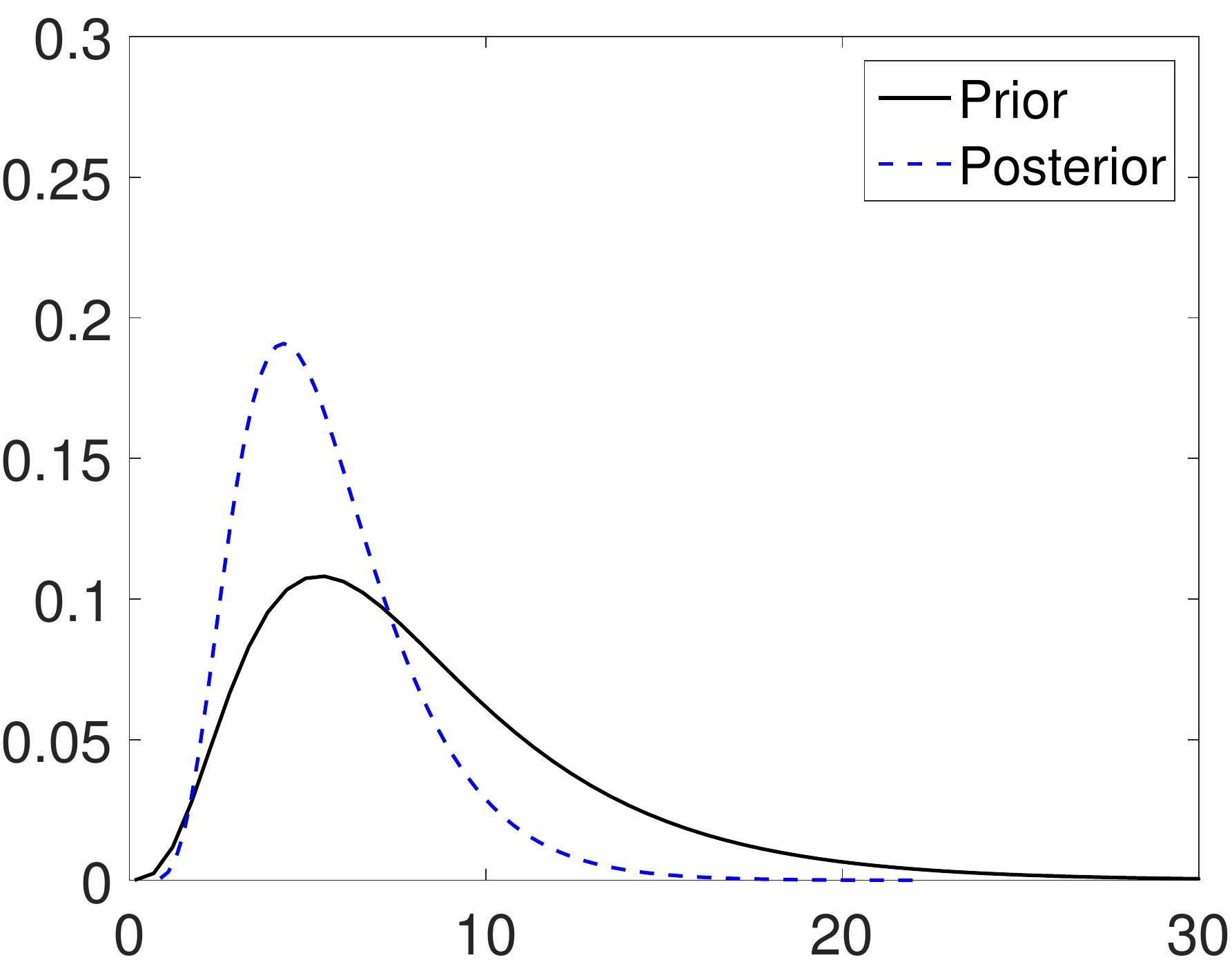} \\
\end{tabular}
\caption{Infinite beta mixture calibrated (BMC), calibrated (C) and
  non calibrated (NC) combinations for a dataset of 1,000 samples from the autoregressive mixture model mAR. Top: PITs cdf of the combination models C (black),
  NC (red) and BMC (blue) and BMC 99\% HPD (gray). Bottom: prior (black) and
  posterior (blue) number of components of the random BMC model.}  \label{fig:AR}
\end{center}
\end{figure}

\section{Empirical applications}\label{sec:empirical}

Then, we investigate relative
predictability accuracy for the out-of-sample period.  Precisely, as
in \cite{GewekeAmisano2010}, \cite{GewekeAmisano2011} and
\cite{MitKap13}, we evaluate the predictive densities using the
Kullback Leibler Information Criterion (KLIC) based measure, utilizing
the expected difference in the Logarithmic Scores of the candidate
forecast densities.  The KLIC computes the distance between the true
density of a random variable and some candidate density.  Even though
the true density is not known, for the comparison of two competing
models, it is sufficient to consider the average Logarithmic Score (AvLS).
The continuous ranked probability score (CRPS) at time $t$ for
model $k$ is defined as:
\begin{equation*}
\hbox{CRPS}_{t,k}
= \int \left(F_{t,k}(z)-\one_{[y_t,+\infty)}(z)\right)^2\text{d}z
\end{equation*}
where $F_{t,k}(y)$ and $f_{t,k}(y)$ are the predictive cdf and
pdf, respectively, for model $k$, conditional to the information 
available up to time $t-1$.

\subsection{Stock returns}  \label{sec:stocks}

The first application considers S\&P500 daily percent log returns data
from 3 January 1972 to 31 December 2008, an updated version of the
database used in studies such as \cite{GewekeAmisano2010},
\cite{GewekeAmisano2011} and \cite{MitKap13}.%BLIND
\footnote{We thank James Mitchell for providing data.}  
We estimate a Normal GARCH(1,1) model
and a $t$-GARCH(1,1) model via maximum likelihood (ML) using rolling
samples of 1250 trading days (about five years) and produce one day
ahead density forecasts.  The first one day ahead forecast refers to
December 15, 1975.  The predictive densities are formed by
substituting the ML estimates for the unknown parameters.  We combine
the two predictive densities using a linear pooling with recursive log
score weights, see description in Section
\ref{sec:simulation}.\footnote{More flexible weighting schemes, such
as time-varying weights, can also be computed.}  Also in this section,
we refer to it as the non-calibrated model.  Furthermore we consider
our mixture of beta probability density functions (BMC) to achieve
better calibration properties.  We split the sample in two periods.
The data from December 15, 1975 to December 31, 2006 are used for an
in-sample calibration of our method to investigate its properties over
a long period. The data from January 3, 2007 to December 31, 2008 for
a total of 504 observations, are used for our out-of-sample analysis.\footnote{In
the linear pooling, we use equal weights for the first forecast for the value January 3, 2007; then
weights are updated summing recursively previous realized log scores of both models.
Using the forecasts from December 15, 1975 to December 31, 2006 as training sample for log score weights
reduces drastically forecast accuracy.}
Therefore, we extend evidence in \cite{GewekeAmisano2010} and
\cite{GewekeAmisano2011} by focusing on the period related to Great
Financial Crisis, with the first semester of 2007 considered a
tranquil period and the remaining part of the sample corresponding to
the most turbulent times.  In this experiment, we fit the calibration
over a moving window of 250 days and produce one-day ahead
forecasts.\footnote{We also investigated out-of-sample performance over
the period from 15 December 1976 to 16 December 2002, the same sample
applied in \cite{GewekeAmisano2010} and \cite{GewekeAmisano2011}.
Superior performance of our BMC are confirmed in this longer sample.}

\begin{figure}[t]
\begin{center}
\begin{tabular}{cc}
\includegraphics[width=4cm]{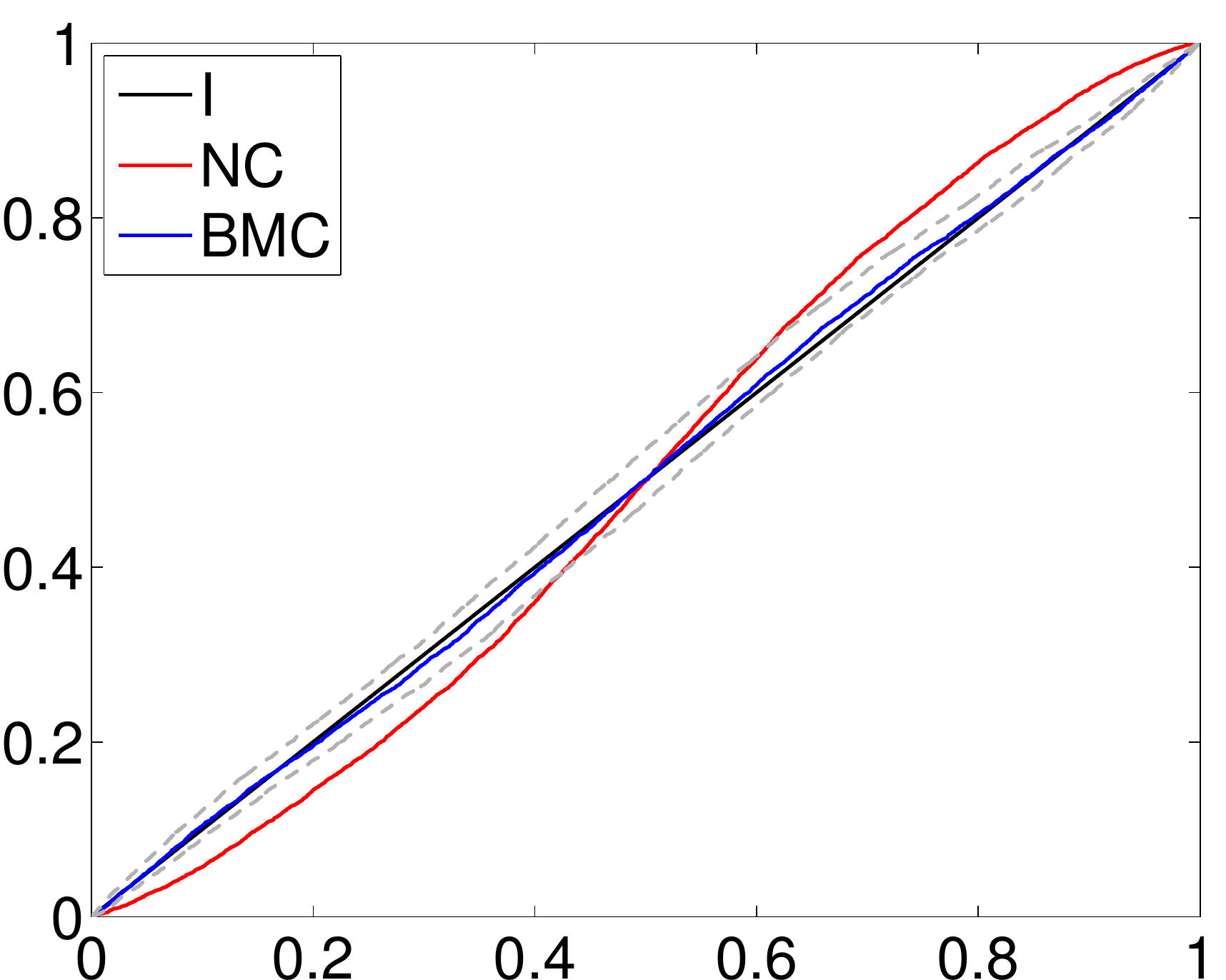}&
\includegraphics[width=4cm]{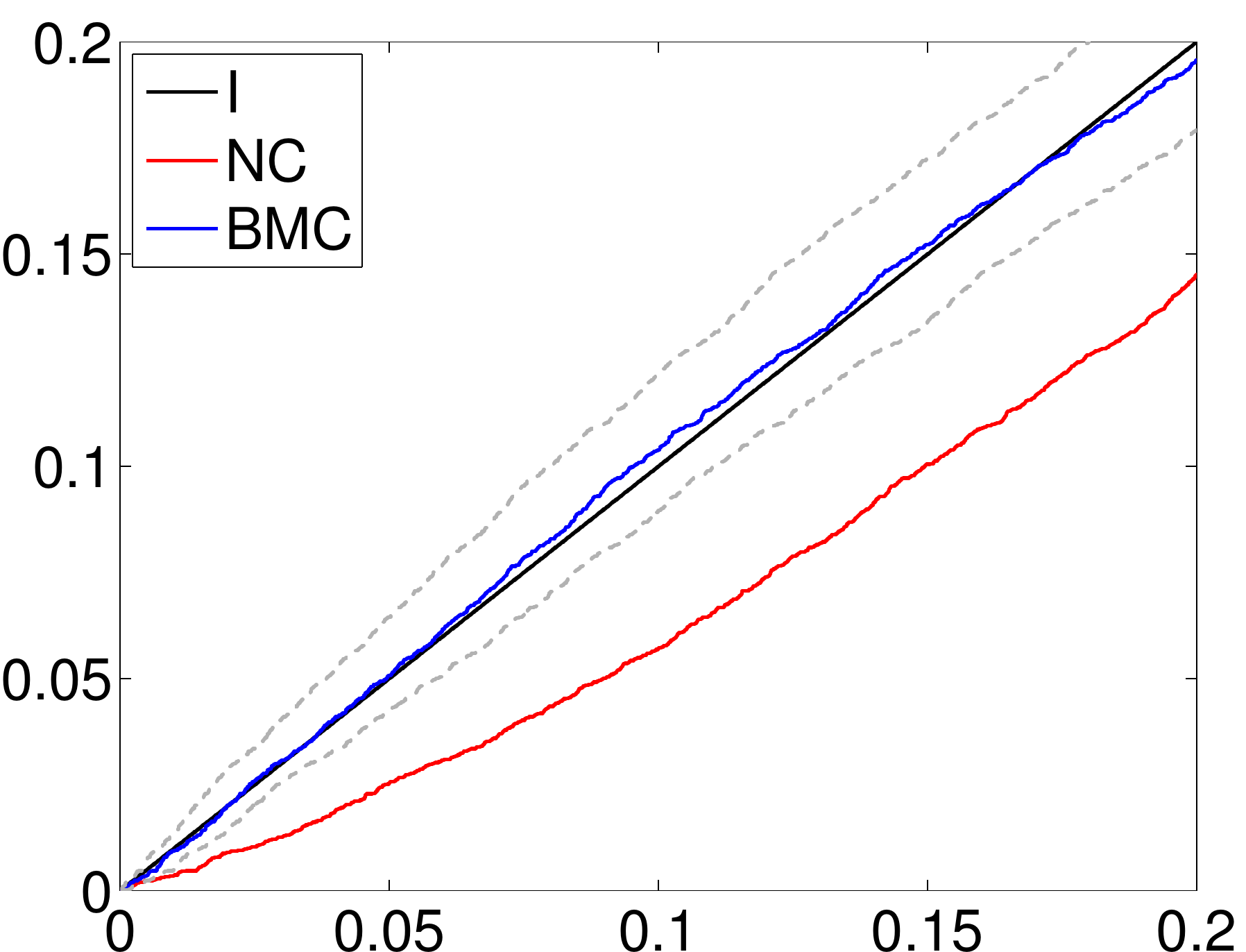}\\
\includegraphics[width=4cm]{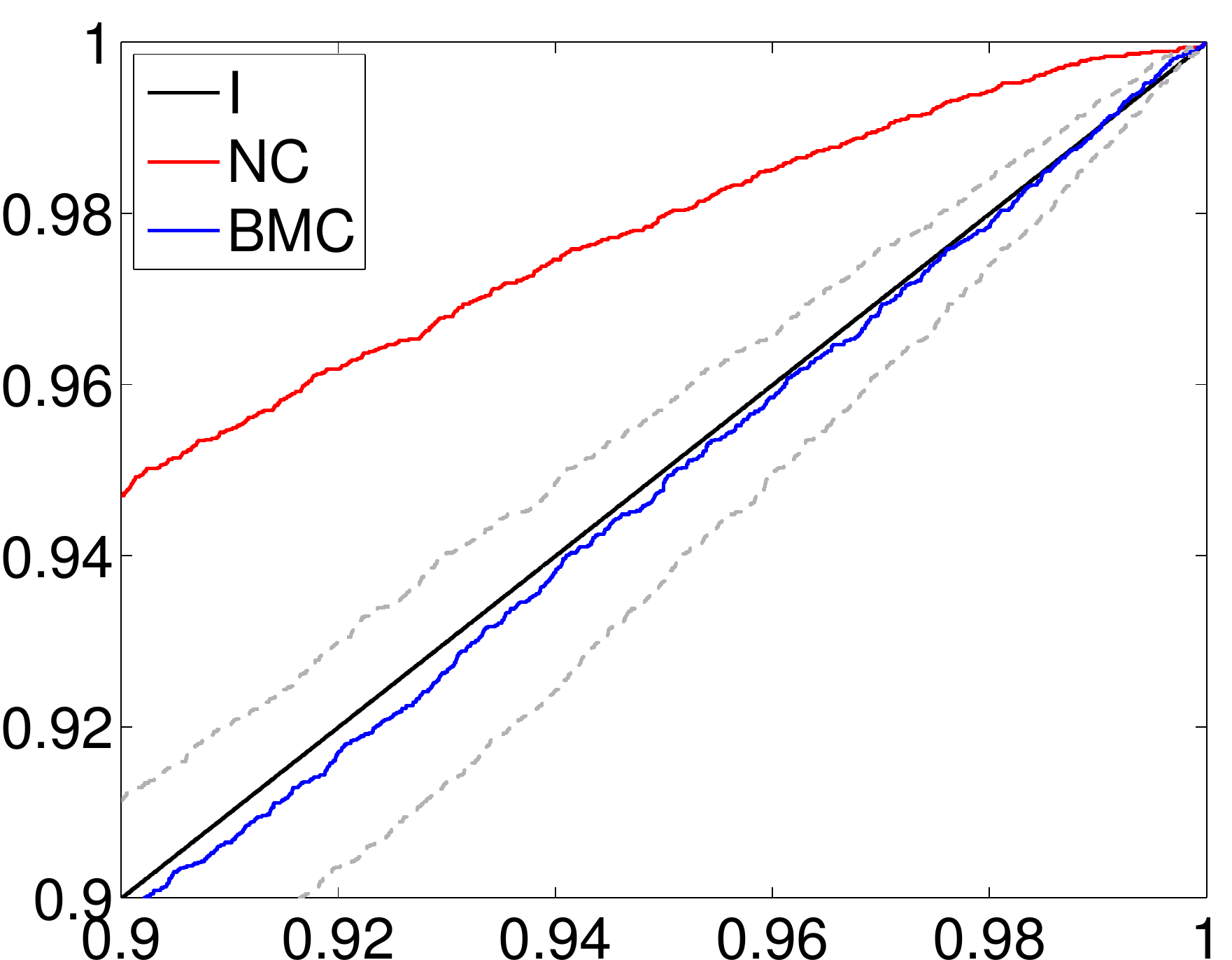}&
\includegraphics[width=4cm]{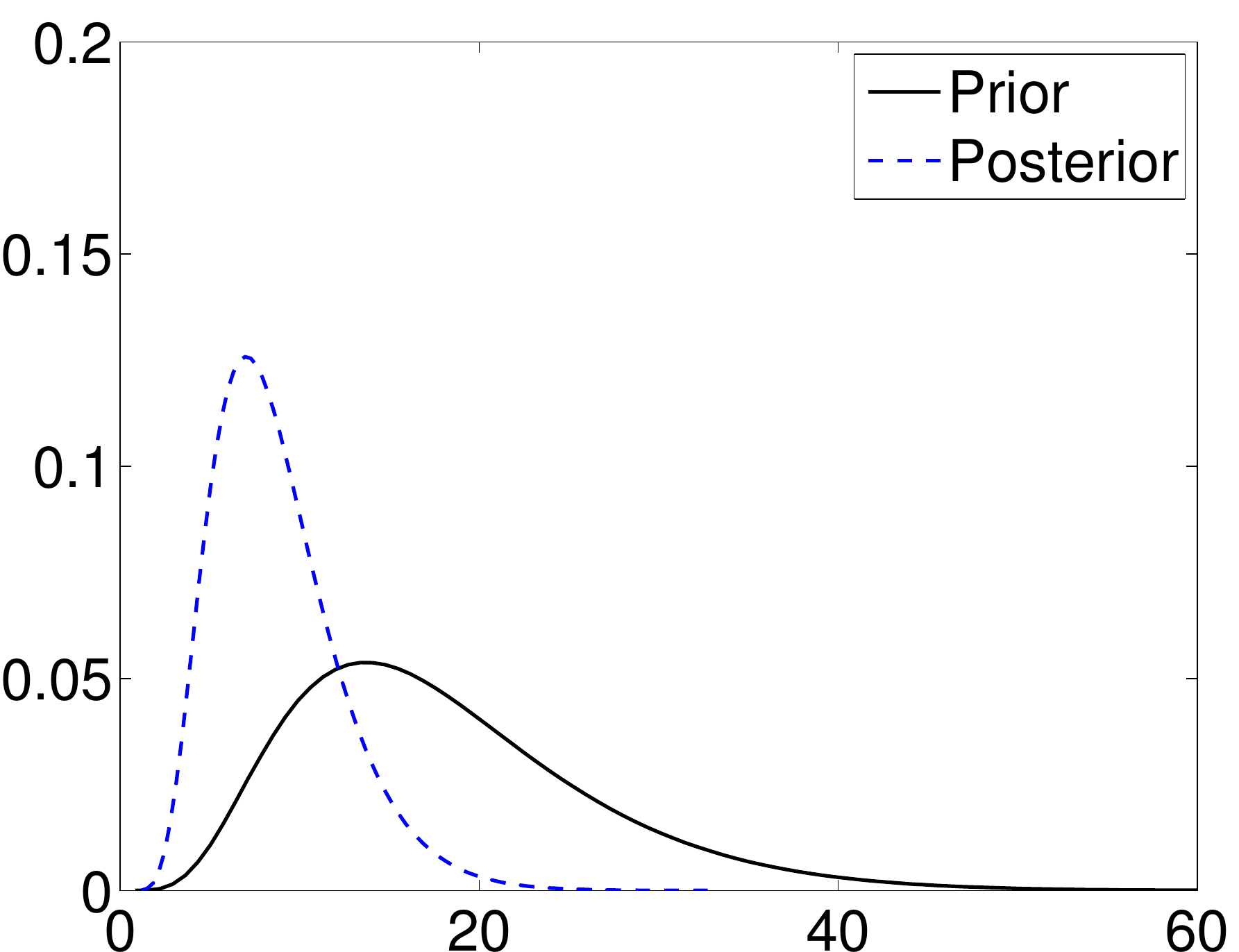}\\
\end{tabular}
\caption{Infinite beta mixture calibrated (BMC), calibrated (C) and
  non calibrated (NC) combinations for the S\&P500 daily percent log
  returns data. PITs cdf (lines from 1 to 3) of the idea model I C
  (black), combination models NC (red) and BMC (blue) and BMC 99\% HPD
  (gray). Prior (black) and posterior (blue) number of components of
  the random BMC model (bottom).}  \label{fig:inf.mix.real1}
\end{center}
\end{figure}

First, we compare the two individual models and the two combinations
in terms of calibration, measured as PIT, over the full sample period
(in-sample and out-of-sample periods).\footnote{Figures focusing only
on the out-of-sample period provide similar evidence and available
upon request from the authors.}  Figure \ref{fig:inf.mix.real1}
reports calibration results for the in-sample analysis.  The BMC line
is the closer to the 45 degree line, which represents the PIT plot for
the unknown true/ideal model.  This 45 degree line always belongs to
the confidence interval of the BMC.  NC is not calibrated for all
quantiles.  In particular, on the upper and lower tails, the NC
differs substantially from the BMC.  As in the simulation exercises
the posterior density for the numbers of beta components in BMC is
more concentrated than the prior.

\begin{table}[t]
\caption{Average log scores for the Normal GARCH model (Normal),
  $t$-GARCH model (Student-$t$), linear pooling (NC) and beta
  mixture calibration (BMC) over the sample period from January 1,
  2007 to December 31, 2008.}  \label{table:LS}
\begin{center}
\begin{tabular}{lrrrr}
  \hline\hline
   & Normal& Student-$t$ & NC& BMC\\
  AvLS & -2.311 & -1.650 & -1.827 & -1.450 \\
  \hline\hline
\end{tabular}
\end{center}
\end{table}

We now turn our attention to the out-of-sample analysis and the log
score results.  Table \ref{table:LS} reports average log scores for
the 4 forecasting methods.  BMC provides the highest score.  Figure
\ref{fig:cum.LS} in Supplementary Materials \ref{appendix_real} shows that after 
the initial weeks of January 2007
where models perform similarly, BMC outperforms the other three
approaches.  The $t$-model provides higher scores than the
normal one and the non-calibrated combination.  The accuracy of the
normal Garch model is very low during our OOS period, in particular on
the extreme events, which results in deteriorating NC performance
after August 2007, the beginning of the turbulent times.  Just
selecting the $t$-GARCH version or, even better, applying local
weights as in our BMC improves accuracy.  Figure \ref{fig:fanchart} in 
\ref{appendix_real}
shows the BMC-based predictive density.

\subsection{Wind speed}  \label{sec:wind}

The second empirical example considers the dataset used in
\cite{LerTho2013}.%BLIND 
\footnote{We thank Sebastian Lerch and Thordis Thorarinsdottir for providing data and forecasts.}  
It consists of 50
ensemble member predictions \citep{MolEt1996} of wind speed at
ten meters above the ground, obtained from the global ensemble
prediction system of the European Centre for Medium-Range Weather
Forecasts (ECMWF).  We restrict our attention to the ensemble
predictions for the maximum wind speed at the station at Frankfurt
airport.  The station ensemble forecasts are obtained by bilinear
interpolation of the gridded model output.

\begin{figure}[t]
\begin{center}
\begin{tabular}{cc}
\includegraphics[width=4cm]{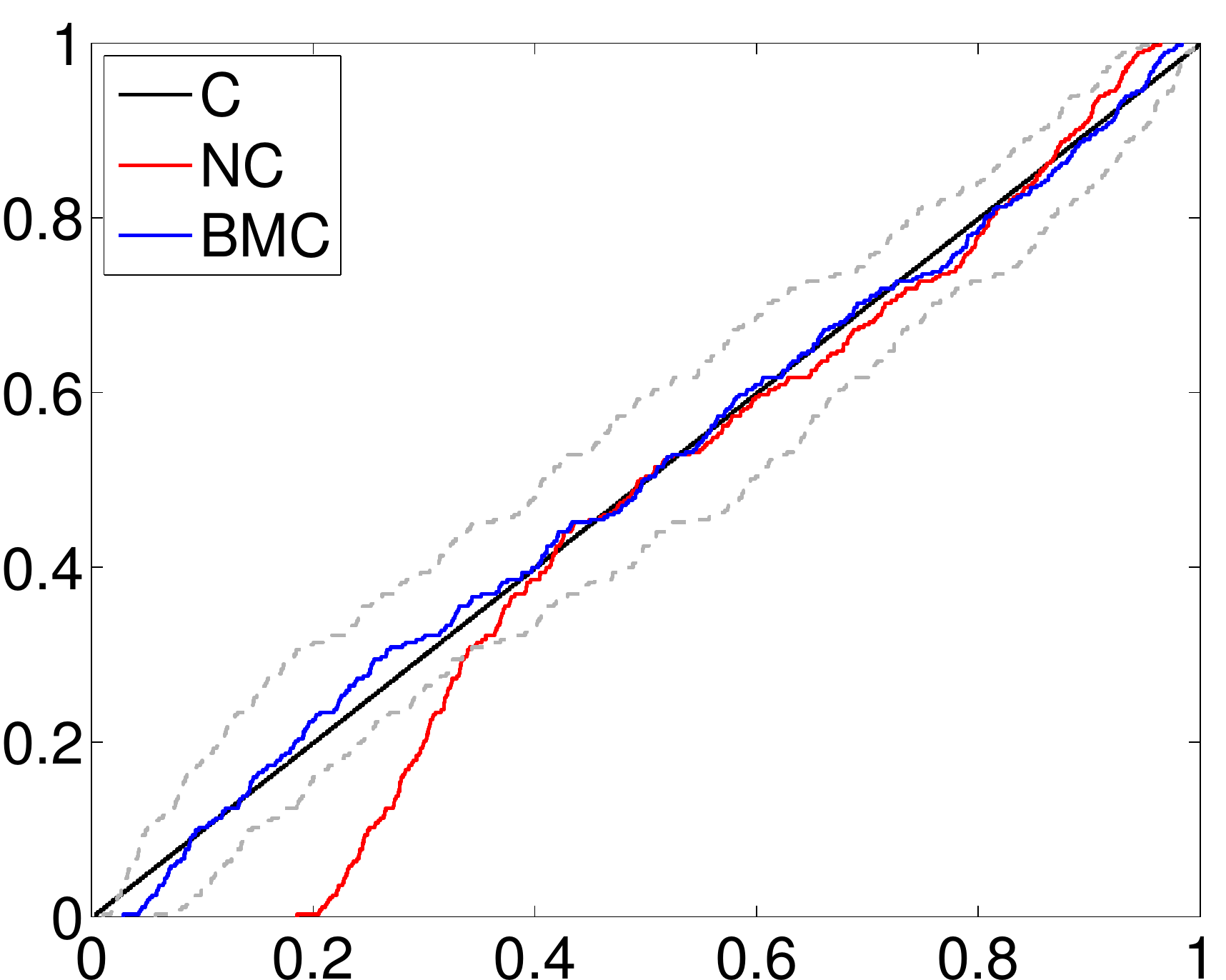}&
\includegraphics[width=4cm]{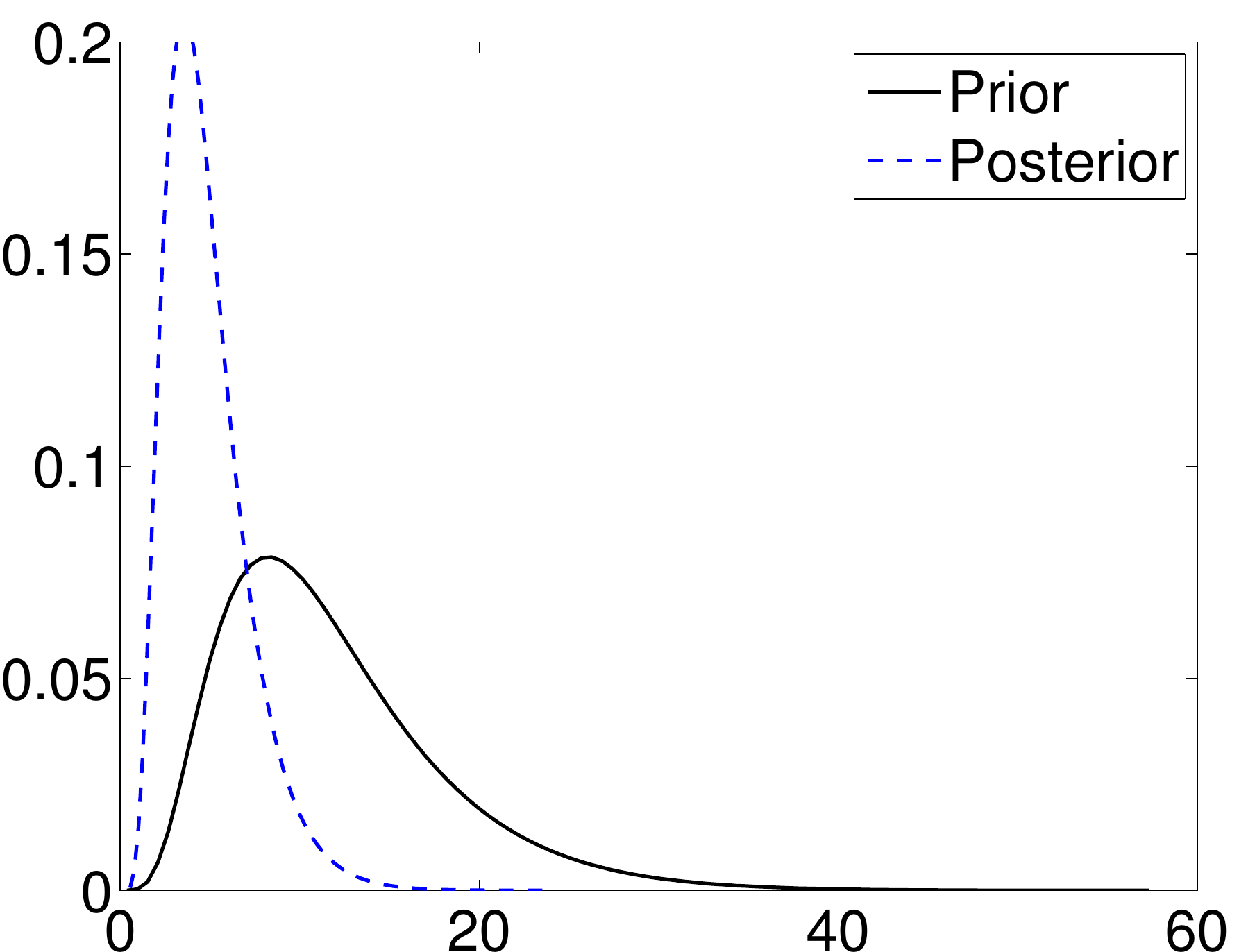}\\
\end{tabular}
\caption{Infinite beta mixture calibrated (BMC), calibrated (C) and
  non calibrated (NC) combinations for the maximum wind speed
  at the station at Frankfurt airport. PITs (left top) of the
  combination models C (black), NC (red) and BMC (blue) and BMC 99\%
  HPD (gray). Prior (black) and posterior (blue) number of components
  of the random BMC model (right top).}  \label{fig:inf.mix.real2}
\end{center}
\end{figure}

We consider the ECMWF ensemble run initialized at 00 hours UTC with a
horizontal resolution of about 33 km, a temporal resolution of 3-6
hours and lead times of 3, 6 and 24 hours.  To obtain predictions of
daily maximum wind speed, we take the daily maximum of each ensemble
member at the Frankfurt location.  One day ahead forecasts are given
by the maximum over lead times.  The observations are hourly
observations of 10-minute average wind speed which is measured over
the 10 minutes before the hour.  To obtain daily maximum wind speed,
from 1 May 2010 to 30 of April 2011, we take the maximum over the 24
hours corresponding to the time frame of the ensemble forecast.

The results presented below are based on a verification period from 9
August 2010 to 30 April 2011, consisting of 263 individual forecast
cases.  Additionally, we use data from 1 February 2010 to 30 April
2011 to obtain training periods of equal length for all days in the
verification period and for model selection purposes and forecasts
from May 1, 2010 to 8 August 2010 (100 observations) as initial
training period for the combination methods.

Following \cite{LerTho2013}, we consider two competing models: the
truncated normal distribution (TN) and the generalized extreme value
distribution (GEV).  The TN model is estimated by minimizing the CRPS.
The GEV model is estimated by maximum likelihood estimation.

First, we report in-sample results over the sample from May 1, 2010 to
30 April 2011.  Then, we implement an out-of-sample exercise for the
period from August 9, 2010 to 30 April 2011.  We report both log score
and CRPS results.

\begin{table}[t]
\caption{Average log scores (AvLS) and average CRPS (AvCRPS) for the
  truncated normal (TN), the generalized extreme value
  distribution (GEV), linear pooling (NC) and beta mixture calibration
  (BMC) over the sample period from August 9, 2010 to 30 April 2011.}
  \label{tab:wind}
\begin{center}
\begin{tabular}{lrrrr}
  \hline\hline
   & TN& GEV& NC& BMC\\
  AvLS &     -2.812 & -2.904 & -2.433 & -1.997 \\
  AvCRPS & 1.346 & 1.802 & 1.314 & 0.982 \\
 \hline\hline
\end{tabular}
\end{center}
\end{table}

Figure \ref{fig:inf.mix.real2} reports in-sample calibration
results. The BMC line is close to the ideal model and always includes
the 45 degree line in the confidence interval.  The NC performs poorly
for small quantiles.  The posterior density for the numbers of beta
components in BMC is more concentrated than the prior, confirming also
in this exercise that data are informative on the number of mixture
components.  When focusing on the OOS exercise, the BMC predictive
distribution predicts accurately and provides the highest average LS
and the lowest average CRPS in Table \ref{tab:wind}.  Gains are
substantial, as Figure \ref{fig:wind} in Supplementary Materials \ref{appendix_real} shows. 
The distribution is often multimodal, see Figure \ref{fig:fanchartwind} in Supplementary Materials \ref{appendix_real}, with the highest mode at low values of wind speed, and a second mode concentrated around values of wind speed greater than 5 meters per second.  The
truncated normal has too many values in the lower and upper tails; the
GEV is too skewed to the upper tail, thus predicting on average too
high values.  The NC is also upper biased by the GEV.  The BMC shifts
the probability mass of the predictive distribution from the upper
tail to the central part and the left tail, thus producing better
calibrated forecasts.

\section{Discussion}  \label{sec:discussion}

We propose a Bayesian approach to predictive density combination and calibration. We build on the predictive density calibration framework of \cite{RanGne10} and \cite{GneRan13} and propose infinite beta mixtures as prior distributions for the calibration function.  We rely upon the flexibility of the infinite beta mixtures to achieve a continuous deformation of the linear density combination.  Each component of the
beta mixture calibrates different parts of the predictive cdf and
uses component-specific combination weights.  Thanks to these
features, our calibration model can also be viewed as a mixture of
local combination models.  Furthermore, our Bayesian framework allows
for including various sources of uncertainty in the predictive density. 

We provide suitable sufficient conditions for weak posterior consistency of our probabilistic calibration. The results imply uniformity of the PITs in the limit, when the number of observations goes to infinity, under both assumptions of i.i.d. and Markovian observations. 

We discuss finite sample properties of our methodology in simulation exercises, showing how the infinite beta components are adequate in applications with multimodal densities, skewness and heavy tails.  In empirical applications to stock returns and wind speed data, our approach provides well-calibrated and accurate density forecasts.

\section*{Acknowledgments}
The authors would like to thank two anonymous referees for their comments and careful reading of the paper.

%\bibliographystyle{apalike}
%\bibliography{20142711}

\newpage
\setcounter{page}{1}

\appendix
\begin{center}

{\Large{Supplementary Materials for:  "Bayesian Nonparametric Calibration and 
       Combination of Predictive Distributions"}\\
      %%%BLIND %%% 
      by F.~Bassetti, R.~Casarin, F.~Ravazzolo.
      }

\end{center}
%
%\author{$\,$}
%
%\date{\today}
%
%\maketitle
%

%\appendix
\renewcommand{\thesection}{S\arabic{section}}
\renewcommand{\theequation}{S\arabic{equation}}
\renewcommand{\thepage}{S\arabic{page}}

\section{Computational details}  \label{sec:details}

\subsection{Gibbs sampler for the finite beta mixture model}  \label{sec:details.finite}

\par\noindent \textit{1. Full conditional distribution of $D$.}
Samples from the full conditional of $D$ given $(\btheta,Y)$ are
obtained by drawing sequentially over $t$, vectors
$d_t = (d_{1t}, \ldots, d_{Kt})$ from multinomial distributions with
probabilities
\begin{equation*}
\pi(d_{kt}=1|\btheta,Y) \propto
w_kb^{*}_{\mu_k,\nu_k}\left(H(y_t|\bomega_k)\right)h(y_t|\bomega_k)
\end{equation*}
for $k=1, \ldots, K$.

\bigskip
\par\noindent \textit{2. Full conditional distribution of $(\bmu,\bnu)$.}
Samples from the full conditional of $(\bmu,\bnu)$ given
$(\bw,\bomega,D,Y)$ are obtained in a sequence of Metropolis-Hastings (MH) steps on
a transformed space.  Following \cite{BouZioMon06}, we let: $\mu_k=1/(1+\exp\{-\gamma_k\})$ and
$\nu_k=\exp\{\lambda_k\}$, $k=1,\ldots,K$ and draw iteratively
from
\begin{eqnarray*}
&&\pi(\gamma_k,\lambda_k|\bw,\bomega,D,Y)\propto\\
&&\prod_{t\in\cD_k}b^{*}_{\mu_k,\nu_k}
\left(H(y_t|\bomega_k)\right)\mu_k^{\xi_{1\mu}-1}(1-\mu_k)^{\xi_{2\mu}-1}
\nu_k^{\xi_{1\nu}-1}\exp\{-\xi_{2\nu}\nu_k\}J(\mu_k,\nu_k),
\end{eqnarray*}
where $J(\mu_k,\nu_k) = \exp\{-\gamma_k-\lambda_k\}
(1+\exp\{-\gamma_k\})^{-2}(1+\exp\{-\lambda_k\})^{-2}$ is the Jacobian
of the transform.  In the MH we use a Gaussian random walk proposal
distribution with covariance matrix $\Sigma=0.05I_2$, which yields
acceptance rates of about 0.4.

\bigskip
\par\noindent \textit{3. Full conditional distribution of $\bomega$.}
Samples from the full conditional of $\bomega$ given
$(\bw,\bmu,\bnu,D,Y)$ are obtained by drawing iteratively
$\bomega_k$, $k=1,\ldots,K$.  At each step we apply a MH with the
prior distribution as proposal.  The acceptance probability of each MH
step is:
\begin{equation*}
\min
\left\{\prod_{t\in\cD_k}\frac{b^{*}_{\mu_k,\nu_k}(H(y_t|\bomega^*))
h(y_t|\bomega^*)}{b^{*}_{\mu_k,\nu_k}(H(y_t|\bomega_k))h(y_t|\bomega_k)},1\right\},
\end{equation*}
where $\bomega^*\sim\Di(\xi_{\omega},\ldots,\xi_{\omega})$.

\bigskip
\par\noindent \textit{4. Full conditional distribution of $\bw$.}
Samples from the full conditional of $\bw$ given
$(\bomega,\bmu,\bnu,D,Y)$ are obtained by exploiting the conjugacy of
the prior distribution, in that
\begin{equation*}
\pi(w_{1},\ldots,w_k|\bomega,\bmu,\bnu,D,Y)\propto\Di(\xi_{w}+T_{1},\ldots,\xi_{w}+T_k).
\end{equation*}

When $K = 1$, we replace the single-move Gibbs sampler with a global MH sampler with
target distribution obtained by applying to the joint posterior
\begin{eqnarray*}
&& \pi(\mu,\nu,\omega|Y) \propto
   \prod_{t=1}^{T}b^{*}_{\mu,\nu}\left(H(y_t|\omega)\right)
   h(y_t|\omega)\mu^{\xi_{1\mu}-1}(1-\mu)^{\xi_{2\mu}-1} \\
&& \times \; \nu^{\xi_{1\nu}-1} \exp\{-\xi_{2\nu}\nu\}
   \omega^{\xi_{1\omega}-1}(1-\omega)^{\xi_{2\omega}-1}
\end{eqnarray*}
where $Y=(y_{1},\ldots,y_t)$, the change of variable
$\mu=1/(1+\exp\{-\theta_1\})$, $\nu=\exp\{\theta_{2}\}$ and
$\omega=1/(1+\exp\{-\theta_{3}\}$. We consider a random walk proposal
on the transformed parameter space accounting for the Jacobian of the
transformation, that is, $J =
\exp\{\theta_2-\theta_1-\theta_3\}(1+\exp\{-\theta_1\})^{-2}(1+\exp\{-\theta_3\})^{-2}$. Setting the covariance matrix to $\Sigma=\hbox{diag}\{0.1,0.05,0.1\}$, we achieve acceptance rates of about 0.4.

\subsection{Gibbs sampler for the infinite beta mixture model}  \label{sec:details.infinite}
Let $\cD_k=\{t=1,\ldots,T|d_t=k\}$ denote the set of indexes of the observations allocated to the $k$-th component of the mixture and
with $\cD=\{k|\cD_k\neq\emptyset\}$ the set of indexes of the
non-empty mixture components.  Then the cardinality of $\mathcal{D}$, $Card(\cD)$,
gives the number of mixture components and $D^*=\sup\cD$ can be interpreted as
the number of stick-breaking components used in the mixture.  As noted by \cite{walker2011}, the sampling of an infinite numbers of $\Theta$ and
$V$ is not necessary, since only the elements in the full conditional pdfs of $D$ are needed.  The maximum number of atoms and stick-breaking components to sample is
$N^*=\max\{t=1,\ldots,T|N^*_t\}$, where $N^*_t$ is the
smallest integer such that $\sum_{j=1}^{N^*_t}w_{j}>1-u_t$.  Thus
sampling from the joint $\pi(V,U|\Theta,D,Y,\psi)$ is obtained by
splitting $V=(V^*,V^{**})$, where $V^*=(v_{1},\ldots,v_{D^*})$
and $V^{**}=(v_{D^{*}+1},\ldots,v_{N^*})$, and by further collapsing the
Gibbs, that is by sampling from $\pi(V^*|\Theta,D,Y,\psi)$ and
$\pi(U|V^*,\Theta,D,Y,\psi)$ and then from
$\pi(V^{**}|V^*,U,\Theta,D,Y,\psi)$.

\par\noindent \textit{1. Full conditional distribution of $V^*$.}
Sampling from the full conditional of $V^*$ given
$(D,\Theta,Y,\psi)$ is obtained by drawing $v_k$, with $k\leq D^*$,
from the full conditionals
\begin{equation*}
\pi(v_k|D,Y)\propto (1-v_k)^{\psi+b_k-1} v_k^{a_k},
\end{equation*}
that is, the PDF of a $\Beta(a_k+1,b_k+\psi)$ with
$a_k=\sum_{t=1}^{T}\one_{\{d_t=k\}}$ and
$b_k=\sum_{t=1}^{T}\one_{\{d_t>k\}}$.

\bigskip
\par\noindent \textit{2. Full conditional distribution of $U$.}
Samples from the full conditional of $U$ given
$(V,D,\Theta,Y,\psi)$ is obtained by simulating from the uniform
\begin{equation*}
\pi(u_t|V,D,Y)\propto \frac{1}{w_{d_t}}\one_{\{u_t<w_{d_t}\}}
\end{equation*}
for $t=1,\ldots,T$.

\bigskip
\par\noindent \textit{3. Full conditional distribution of $V^{**}$.}
Sampling from the full conditional of $V^{**}$ given
$(V^*,U,D,\Theta,Y,\psi)$ is obtained by sampling from
\begin{equation*}
\pi(v_k|U,D,Y)\propto (1-v_k)^{\psi-1},
\end{equation*}
that is, the PDF of a $\Beta(1,\psi)$, with $k=D^*+1,\ldots,N^*$.

\bigskip
\par\noindent \textit{4. Full conditional distribution of $\Theta$.}
Sample the elements $k$, $k=1,\ldots,N^*$, of $\Theta$
given $(U,D,V,Y,\psi)$, from the full conditional
\begin{eqnarray*}
&& \!\!\!\!\!\!\!\!\!\pi(\btheta_k|U,D,V,Y) \propto
\prod_{t\in\cD_k}\!\!b^{*}_{\mu_k,\nu_k}\!\!\left(H(y_t|
\bomega_k)\right)h(y_t|\bomega_k)\\
&& \times \; \mu_k^{\xi_{\mu}-1}(1-\mu_k)^{\xi_{\mu}-1}
\nu_k^{\xi_{\nu}/2}\exp\{-\xi_{\nu}\nu_k/2\}
\prod_{i=1}^{M}\omega_{ik}^{\xi_{\omega}-1}
\one_{ \{ \bomega_k \in \DM \} }
\end{eqnarray*}
for $k\in\cD$, and from the prior $G_{0}$ for $k\notin\cD$.  We sample
from the full conditional by iterating over the following steps:
\begin{enumerate}
\item[(a)] $\pi(\mu_k,\nu_k|\bomega_k,U,D,V,Y,\psi)$
\item[(b)] $\pi(\bomega_k|\mu_k,\nu_k,U,D,V,Y,\psi)$.
\end{enumerate}
We apply here the same sampling strategy described for the parameter
of the finite beta mixture model.

\bigskip
\par\noindent \textit{5. Full conditional distribution of $D$.}
Samples from the full conditional of $D$ given
$(V,U,\Theta,Y,\psi)$ are obtained by sampling from
\begin{equation*}
\pi(d_t|V,U,Y) \propto
\one_{\{u_t<w_{d_t}\}}f_{\mu_{d_t},\nu_{d_t}}
\left(H(y_t|\bomega_{d_t})\right)h(y_t|\bomega_{d_t}),
\end{equation*}
with $d_t\in\{1,\ldots,N_t^*\}$, where $N_t^*$ is
defined above.

\bigskip
\par\noindent\textit{6. Full conditional distribution of $\psi$.}
If the dispersion parameter $\psi$ is assumed to be random with $\mathcal{G}a(c,d)$ prior, then an extra step is needed in the Gibbs sampler.  More specifically, the full conditional distribution of $\psi$ given $U$, $D$, $V$ and $\Theta$ has density
\begin{equation*}
\pi(\psi|K,T)\propto B(\psi,T)\psi^{K+c-1}\exp\{-d \psi\}\one_{\psi\in (0,+\infty)},
\end{equation*}
which depends only on the number of observations $T$ and the number of
mixture components $N^*$, which has been defined above.

\bigskip
The Gibbs sampler can used to generate draws from the predictive distribution $\hat{F}_{T+1}(y_{T+1})$. At each iteration a uniform random variable $u^{(i)}$ is sampled from the unit interval and $\btheta_j^{(i)}$ is used such that $w_{j-1}^{(i)}<u^{(i)}<w_{j}^{(i)}$. If $j>N^{*\,(i)}$, then more weights are required than currently exist, and they can be
sampled from $\mathcal{B}e(1,\psi)$ and the additional $\btheta_j^{(i)}$ from $G_0$. Having taken $\btheta_j^{(i)}$, $y_{T+1}^{(i)}$ can be sampled from $B^{*}_{\mu_{j}^{(i)},\nu_j^{(i)}}\left(H(y_{T+1}|\bomega_{j}^(i))\right)$.

\section{Calibration consistency }\label{additionalconsistency}
%Let us consider the case in which the pooling parameters $\boldsymbol{\theta}_{p}$ are known.
If the pooling parameters  are fixed, say $\boldsymbol{\omega}=\boldsymbol{\omega}_0$,
the inference is   necessarily  limited to  the calibration parameters $\boldsymbol{\theta}_{c}=(\mu,\nu)$, hence $\Theta=[0,1]\times\RE^+$
and $G$ is a DP process on $\mathcal{M}([0,1]\times\RE^+)$ with base measure $G_0$ and (known) concentration parameter $\psi$.

In this special case  $\Pi^{*}$ turns out to be the prior induced by
\[
G \mapsto  \int b^*_{\boldsymbol{\theta}_{c}}(H(y|\boldsymbol{\omega_0})) G(d\boldsymbol{\theta}_{c}) h(y|\boldsymbol{\omega_0})
\]
when $G \sim DP(\psi,G_0)$. Then, the analogous of Theorem \ref{corollario} is given below.

\begin{theorem}\label{teocons2}
Let $\boldsymbol{\omega_0}$ be a given point in $\Delta_M$ such that
$h(\cdot|\boldsymbol{\omega_0})$ is continuous and
for every compact set $C \subset \mathcal{Y}$,
\begin{equation}\label{lowerboundh0}
 \inf_{y \in C} h(y|\boldsymbol{\omega_0})>0.
\end{equation}
Assume that  $f_0$ is a continuous density on $\mathcal{Y}$
such that \eqref{entropy2} holds with $\boldsymbol{\omega_0}$ in place of  $\boldsymbol{\omega^*}$.  If $G_0$ has full support, then
$f_{0}\in KL(\Pi^{*})$.
% and let  $f_{0}(x)=u_0(H(y|\boldsymbol{\omega_0}))h(y|\boldsymbol{\omega_0})$ where
%$u_0$ is a continuous probability density on $(0,1)$ such that  \eqref{entropyu0} is satisfyied.
%If $G_0$ has full support,  then $f_{0}\in KL(\Pi^{*})$.
\end{theorem}

In the previous theorem, contrary to Theorem \ref{corollario},  $\boldsymbol{\omega_0}$ represents the true parameter. We do not assume $\boldsymbol{\omega_0}$ is in the interior of $\Delta_M$, which means that the set of models in the combination scheme can be complete. In some situations, it is useful to consider a base measure $G_0$ without full support. In this spirit,  following the techniques of \cite{GhosalRoyTang2007}, we can prove the next result.

\begin{theorem}\label{thm_cons1}  
% \begin{proposition}
 Let $\boldsymbol{\omega_0}$ be a given point in
  $\Delta_M$ and let
  \[
  f_{0}(y)=u_0(H(y|\boldsymbol{\omega_0}))h(y|\boldsymbol{\omega_0})
  \] with
$u_0(x)=w_0 b_{\mu_0,\nu_0}^*(x)+(1-w_0) \int_{(0,1)\times \RE^+} b_{\mu,\nu}^*(x) P_0(d\mu d\nu)$,
$P_0$ being a probability measure on $(0,1)\times \RE^+$.
If $(\mu_0,\nu_0)$ belong to $supp(G_0)$,
$\text{supp}(P_0) \subset \text{supp}(G_0)$, and for some $\zeta>0$
and $0<\eta<\min(\mu_0,1-\mu_0,\nu_0,w_0)$ one has
\begin{equation}\label{cond_moment}
\int_0^1 \frac{u_0(x)^{\zeta+1}}{x^{\zeta A }(1-x)^{\zeta B}}dx<+\infty,
\end{equation}
for $A=(\mu_0+\eta)(\nu_0+\eta)-1$ and $B=(1-\mu_0+\eta)(\nu_0+\eta)-1$,
 then $f_{0}\in KL(\Pi^{*})$.
%The same result holds for every $u_0(x)=\int_{(0,1)\times \RE^+} b_{\mu,\nu}^*(x) P_0(d\mu d\nu)$,
%if $supp(G_0)$ is the compact set $[a,b]\times [m,M] \subset (0,1)\times \RE^+$,
%\eqref{cond_moment} holds with $A=bM-1,B=(1-a)M-1$
%and $supp(P_0)\subset supp(G_0)$.
\end{theorem}
We shall notice that the specification of the calibration function $u_0$ in the true density $f_0$ can be used to build a statistical procedure for testing the hypothesis of well calibrated PITs against the alternative of not well-calibrated PITs.

%\section{Proofs of the results of Section \ref{sec:consistency}}\label{appendix_consistencyl}

\section{Proofs}\label{appendix_consistencyl}
 % of the results of Section \ref{sec:consistency}}\label{appendix_consistencyl}

\subsection{Proofs of the results of Sections \ref{sec:consistency_iid} and \ref{additionalconsistency}} %\ref{ConMarkov}

The proof of Theorem \ref{corollario} is based on an application of Theorem 1 and Lemma 3 of \cite{WuGhosal2008,WuGhosal2008correction}.
For the shake of clarity we report the
statements of these results in Theorems %\ref{ThmWG1} and 
 \ref{ThmWG2} below.
%in Appendix \ref{appendix_consistencyl}.
%
%
%To prove that $f_0 \in KL(\Pi^{*})$,  \cite{WuGhosal2008} suggest to split the problem in
%two steps, as shown  in the next very simple theorem.
%
%\begin{theorem}[Thm. 1 of \cite{WuGhosal2008}]\label{ThmWG1}
%If for any $\eps$ there is a probability measure $G_\eps$ and a measurable set $\mathcal{W}\subset \mathcal{M}(\Theta)$,
%with $G_\eps \in\mathcal{W} $ and $\Pi(\mathcal{W})>0$, such that
%\begin{itemize}
%\item[(H1)] $\int \log(f_0/f_{G_\eps})f_0<\eps$,
%\item[(H2)]  $\int \log(f_{G_\eps}/f_{G})f_0<\eps$ for every $G$ in $\mathcal{W}$;
%\end{itemize}
%then $f_0 \in KL(\Pi^{*})$.
%\end{theorem}
%
%\begin{proof} Since
%\[
%\int \log\big(\frac{f_0(y)}{f_{G}(y)}\big)f_0(y)=\int \log\big(\frac{f_0(y)}{f_{G_\eps}(y)}\big)f_0(y)dy+\int \log\big(\frac{f_{G_\eps}(y)}{f_{G}(y)}\big)f_0(y)dy,
%\]
% one has
%%$$\Pi^* \big \{ \int \log(f_0(y)/f_{G}(y))f_0(y)dy<2 \eps  \big \}
%%\geq \Pi^* \big  \{\int \log(f_{G_\eps}(y)/f_{G}(y))f_0(y)dy <\eps  \big \} \geq \Pi(\mathcal{W})>0.$$
%\[
%\begin{split}
%\Pi^* & \Big \{ \int \log(f_0(y)/f_{G}(y))f_0(y)dy<2 \eps  \Big \} \\
%& \geq \Pi^* \Big  \{\int \log(f_{G_\eps}(y)/f_{G}(y))f_0(y)dy <\eps  \Big \} \geq \Pi(\mathcal{W})>0. \\
%\end{split}
%\]
%\end{proof}
%
%Given a probability $G_\eps$ satisfying (H1), in general some work is need to verify  assumption (H2).
%A useful sufficient condition is contained in Lemma 3 of of \cite{WuGhosal2008}.

\begin{theorem}[Theorem 1 and Lemma 3 of \cite{WuGhosal2008}]\label{ThmWG2} 
Let $\Theta$ be a separable metric space. If for any $\eps>0$
there is a probability measure $G_\eps\in supp(\Pi)$ and a closed set $D_\eps$  such that 
\begin{itemize}
\item[(H1)] $\int \log(f_0/f_{G_\eps})f_0<\eps$,
%\end{itemize}
%\begin{itemize}
\item[(H2)] $D_\eps$ contains $\sup(G_\eps)$ in its interior and
\[
 \int \log\Big( \frac{f_{G_\eps}(y)}{\inf_{\theta \in D_\eps} K(y;\theta)} \Big ) f_0(y)dy<+\infty,
\]
\item[(H3)] $\inf_{y \in C} \inf_{\theta \in D_\eps} K(y;\theta)>0$ for every compact set $C \subset \mathcal{Y}$,
\item[(H4)] $\{ \theta \mapsto K(y;\theta): y \in C\}$
is uniformly equicontinuous on $D_\eps$,
\end{itemize}
then 
%{\rm (H2)} holds for a suitable $\mathcal{W} \subset \mathcal{M}(\Theta)$ with $\Pi(\mathcal{W})>0$, and hence 
$f_0 \in KL(\Pi^{*})$.
\end{theorem}

Assumption (H1)  corresponds to (A1) in Theorem 1 of
\cite{WuGhosal2008}. %Note the Theorem  \cite{WuGhosal2008}
%is stated for Type II prior and it has an additional assumption (A2), which is not needed
%for Type I prior. Note also  that 
Assumptions (H2)-(H3) correspond to assumptions (A7)-(A8) of Lemma 3  of \cite{WuGhosal2008},
while (H5) is slightly different from the original assumption (A9), see \cite{WuGhosal2008correction}. In order to apply the results given above, 
it is useful to reformulate Theorem \ref{corollario} as follows.
\begin{theorem}\label{teocons3}
Assume that the functions $f_m(\cdot)$ are continuous on $\mathcal{Y}$.
Let $u_0$ be a continuous density on $(0,1)$
such that
\begin{equation}\label{entropyu0bis}
\begin{split}
& \int_0^1 [|\log(x)|+|\log(1-x)|]u_0(x)dx<+\infty \\
& \text{and} \quad \int_0^1 \log(u_0(x)) u_0(x)dx <+\infty.\\
\end{split}
\end{equation}
Let
$f_0(y)=u_0(H(y|\boldsymbol{\omega^*}))h(y|\boldsymbol{\omega^*})$ with $\boldsymbol{\omega^*}$ in the interior of $\Delta_M$
and assume that, for every compact set $C \subset \mathcal{Y}$,
\begin{equation}\label{lowerboundh0bis}
 \inf_{y \in C} h(y|\boldsymbol{\omega^*})>0.
\end{equation}
Then
$f_{0}\in KL(\Pi^{*})$ whenever $G_0$ has full support.
\end{theorem}

\begin{proof}[Proof of Theorem \ref{teocons3}]
Here we need to think $\Delta_M$ as the set
$\{(\omega_1,\dots,\omega_{M-1}) \in [0,1]^{M-1}: \sum_{i=1}^{M-1}\omega_i\leq 1\}$ endowed with
the topology induced by the euclidean norm. Clearly, $\omega_{M}$  will denote $1- \sum_{i=1}^{M-1}\omega_i$.

{\it Verification of H1 of Theorem  \ref{ThmWG2}.}
Since $u_0$ is continuous on $(0,1)$ and $\int_0^1 \log(u_0(x)) u_0(x)dx <+\infty$
by Theorem 1 in \cite{RobRou02} there is
\[
 u_{\eps}(x)=u_{\tilde P_\eps}(x)=\sum_{i=1}^{K_\eps} w_{i,\eps} b^*_{\mu_{i,\eps},\nu_{i,\eps}}(x),
\]
where $\tilde P_\eps(d\mu d\nu):=\sum_{i=1}^{K_\eps} w_{i,\eps} \delta_{\mu_{i,\eps},\nu_{i,\eps}}(d\mu d\nu)$,
such that $KL(u_0,u_\eps) \leq \eps$.
If $G_\eps(d\boldsymbol{\omega} d\mu d\nu):=\delta_{\boldsymbol{\omega^*}}(d\boldsymbol{\omega})\times \tilde P_\eps(d\mu d\nu)$,
then
\[
 f_{G_\eps}(y)=\int b^*_{\mu,\nu}(H(y|\boldsymbol{\omega} )) h(y|\boldsymbol{\omega})G_\eps(d\boldsymbol{\omega} d\mu d\nu) =
u_{\eps}(H(y|\boldsymbol{\omega^*}))h(y|\boldsymbol{\omega^*}).
\]
By a simple change of variables,
\[
\begin{split}
KL(f_0,f_{G_\eps})& =\int u_0(H(y|\boldsymbol{\omega^*}))h(y|\boldsymbol{\omega^*}) \log\Big( \frac{u_0(H(y|\boldsymbol{\omega^*}))h(y|\boldsymbol{\omega^*})}{u_\eps(H(y|\boldsymbol{\omega^*}))h(y|\boldsymbol{\omega^*})} \Big) dy\\
& %=\int u_0(H(y))h(y) \log\Big( \frac{u_0(H(y))}{u_G(H(y))} \Big) dy
=\int_0^1 u_0(z) \log\Big( \frac{u_0(z)}{u_\eps(z)} \Big)dz. \\
\end{split}
\]
That is
\[
 KL(f_0,f_{G_\eps})= KL(u_0,u_\eps) \leq \eps.
\]
 Note that $supp(G_\eps)=\{ \boldsymbol{\omega^*}\}\times \cup_{i=1}^{K_\eps} \{(\mu_{i,\eps},\nu_{i,\eps})\}$
and, since $G_0$ has full support, $G_\eps \subset supp(Dir(\psi,G_0))$.

{\it Verification of H2 of Theorem  \ref{ThmWG2}.}
One can find a compact set $D_\eps^*$  in $(0,1)\times (0,+\infty)$ such that
$D_\eps^*$ contains $\cup_{i=1}^{K_\eps} \{(\mu_{i,\eps},\nu_{i,\eps})\}$ in its interior. Moreover,
recalling that $h(y|\boldsymbol{\omega})=\sum_{i=1}^M \boldsymbol{\omega}_i f_i(y)$ and that
$\boldsymbol{\omega^*}$ is in the interior of $\Delta_M$,
%$\Delta_M=\{(\omega_1,\dots,\omega_{M-1}) \in [0,1]^{M-1}: \sum_{i=1}^{M-1}\omega_i\leq 1\}$,
one can find
a (sufficiently small) compact set $\Delta_\eps^* \subset
\Delta_M$ containing $\boldsymbol{\omega^*}$ in its interior such that if $\boldsymbol{\omega} \in \Delta_\eps^*$
then $ C_{1,\eps} h(y|\boldsymbol{\omega^*}) \leq h(y|\boldsymbol{\omega})\leq C_{2,\eps} h(y|\boldsymbol{\omega^*})$ for every $y$.
It follows that $D_\eps=\Delta_\eps^* \times D_\eps^*$ is a compact set containing $supp(G_\eps)$ in its interior.
Noticing that if $\boldsymbol{\omega} \in \Delta_\eps^*$ then $ C_{2,\eps}  H(y|\boldsymbol{\omega^*}) \geq H(y|\boldsymbol{\omega}) \geq C_{1,\eps}  H(y|\boldsymbol{\omega^*})$
and $ C_{2,\eps} (1- H(y|\boldsymbol{\omega^*})) \geq (1-H(y|\boldsymbol{\omega})) \geq C_{1,\eps} (1- H(y|\boldsymbol{\omega^*}))$, one can write
\[
\begin{split}
I_\eps(y)&:=\inf_{(\boldsymbol{\omega},\mu,\nu) \in D_\eps} K(y;\boldsymbol{\omega},\mu,\nu)  \\
& =\inf_{(\boldsymbol{\omega},\mu,\nu) \in D_\eps} h(y|\boldsymbol{\omega}) \frac{H(y|\boldsymbol{\omega})^{\mu\nu-1}(1-H(y|\boldsymbol{\omega}))^{(1-\mu)\nu-1}}{B(\mu\nu,(1-\mu)\nu)} \\
&
\geq
C_{3,\eps} h(y|\boldsymbol{\omega^*}) H(y|\boldsymbol{\omega^*})^{A_\eps-1}(1-H(y|\boldsymbol{\omega^*}))^{B_\eps-1}=:I^*_\eps(y)\\
\end{split}
\]
where $C_{3,\eps}=C_{1,\eps} C_{2,\eps}^{-2} \inf\{ C_{1,\eps}^{\mu\nu+(1-\mu)\nu} /B(\mu\nu,(1-\mu)\nu)
:(\mu,\nu) \in D_\eps^* \}  >0$, $A_\eps=\sup\{\mu\nu: (\mu,\nu) \in D_\eps^*\}>0$ and $B_\eps=\sup\{(1-\mu)\nu: (\mu,\nu) \in D_\eps^*\}>0$.
Hence, one the one hand $f_{G_\eps}(y) \geq  I_\eps(y)$ and hence $\log( {f_{G_\eps}(y)}/I_\eps(y))\geq 0$, on the other hand
\[
\begin{split}
\int \log\Big( \frac{f_{G_\eps}(y)}{I_\eps(y)}  \Big ) f_0(y)dy  &
\leq \int \log\Big( \frac{f_{G_\eps}(y)}{I^*_\eps(y)}  \Big ) f_0(y)dy \\
& \leq  \int_0^1 \log\Big( \frac{u_{\eps}(x)}{x^{A_\eps-1}(1-x)^{B_\eps-1}}\Big ) u_0(x)dx+|\log(C_{3,\eps})|. \\
\end{split}
\]
Since $C_{4,\eps} x^{A_\eps'-1}(1-x)^{B_\eps'-1} \leq u_{\eps}(x) \leq C_{5,\eps} x^{A_\eps''-1}(1-x)^{B_\eps''-1}$ for suitable constants, it follows that
\[
%\begin{split}
 \int \Big|\log\Big(   \frac{u_{\eps}(x)}{x^{A_\eps-1}(1-x)^{B_\eps-1}}\Big )\Big| u_0(x)dx \leq
% \int \Big|\log\Big( \frac{C_{4,\eps} x^{A_\eps'-1}(1-x)^{B_\eps'-1}}{x^{A_\eps-1}(1-x)^{B_\eps-1}}\Big )\Big| u_0(x)dx   \\
%& \leq
C_{6,\eps} \int [|\log(x)|+|\log(1-x)|]u_0(x)dx<+\infty
%\\
%\end{split}
\]
by assumption \eqref{entropyu0bis}. Hence
\[
 0< \int \log\Big( \frac{f_{G_\eps}(y)}{\inf_{(\boldsymbol{\omega},\mu,\nu) \in D_\eps} K(y;\boldsymbol{\omega},\mu,\nu)} \Big ) f_0(y)dy<+\infty.
\]

{\it Verification of H3 of Theorem  \ref{ThmWG2}.} It follows immediately that, for every compact set $C$,
\[
 \inf_{y \in C} \inf_{(\boldsymbol{\omega},\mu,\nu) \in D_\eps} K(y;\boldsymbol{\omega},\mu,\nu)
\geq \inf_{y \in C}  I^*_\eps(y)
\]
and the right hand side is strictly positive by  \eqref{lowerboundh0bis}.

{\it Verification of H4 of Theorem  \ref{ThmWG2}.}
The function $(\boldsymbol{\omega},\mu,\nu,y) \mapsto K(y;\boldsymbol{\omega},\mu,\nu)$
is continuous and hence uniformly continuous on the compact set $C\times D_\eps$. It follows
that the family $\{ (\boldsymbol{\omega},\mu,\nu) \mapsto K(y;\boldsymbol{\omega},\mu,\nu): y \in C\}$
is uniformly equicontinuous on $D_\eps$.
\end{proof}

\begin{proof}[Proof of Theorem \ref{corollario}]
%Let $F_0$ be the cumulative distribution function of $f_0$
Write $H_0$ and $h_0$ for $H(\cdot|\boldsymbol{\omega^*})$ and its density.
By assumptions one gets that $H_0$ is continuous and strictly increasing.
Hence, if one defines
\[
 u_0(x):=\frac{f(H_0^{-1}(x))}{h_0(H_0^{-1}(x))},
\]
it follows that $f_0(y)=u_0(H_0(y))h_0(y)$. Note that $u_0$ turns out to be a continuous function on $(0,1)$. It remains to check that
assumption \eqref{entropy2} yields \eqref{entropyu0bis}.
Now, a change of variable gives
\[
\int |\log(H(y|\boldsymbol{\omega^*}))|f_0(y)dy=\int |\log(H(y|\boldsymbol{\omega^*}))|u_0(H_0(y))h_0(y)dy=
\int |\log(x)|u_0(x)dx.
\]
Similarly for $\int |\log(1-H(y|\boldsymbol{\omega^*}))|]f_0(y)dy$. Finally
\[
 KL(f_0,h(\cdot|\boldsymbol{\omega^*})) =\int  \log(u_0(H_0(y))u_0(H_0(y))h_0(y)  dy=\int u_0(x) \log(u_0(x))dx.
 \]
\end{proof}

\begin{proof}[Proof of Theorem \ref{teocons2}]
The proof is a simple modification of the proof of Theorem \ref{teocons3}.
% Note that here the assumption
%that $\boldsymbol{\omega_0}$ belongs to the interior of $\Delta_M$ is not needed.
\end{proof}

\begin{proof}[Proof of Theorem \ref{thm_cons1}]
Given any measure $Q$ on $(0,1)\times \RE^+$ recall that
$f_Q(x)=u_Q(H(y|\boldsymbol{\omega_0}))h(y|\boldsymbol{\omega_0})$
where $u_Q(x)=\int b_{\mu,\nu}^*(x) Q(d\mu d\nu)$.
%Let $KL(f,g)=\int f \log\big(\frac{f}{g}\big)$.
%
Again, by a simple change of variables,
$KL(f_G,f_0)=KL(u_G,u_0)$.
Hence to prove that $f_{0}\in KL(\Pi^{*})$ it suffices to prove that
for every $\epsilon>0$, $P\{KL (u_G,u_0)\leq \epsilon \}>0$.

Now recall that if $G \sim DP(\psi,G_0)$ then
$G$ admits the representation
$G=w_1 \delta_{\theta_1}+(1-w_1) G_1$ where $w_1,\theta_1=(\mu_1,\nu_1)$ and $G_1$ are stochastically
independent, $G_1  \sim DP(\psi,G_0)$, $w_1 \sim Beta(1,\phi_1)$ and $\theta_1 \sim G_0$.

Given $\eta,\eta'>0$ define $\mathcal{U}_\eta:=\{ (w,\mu,\nu) \in (0,1)^2\times \RE^+: |w-w_0|\leq \eta, |\mu-\mu_0|\leq \eta,
|\nu-\nu_0|\leq \eta\}$ and $\mathcal{U}_{\eta,\eta'}^*:=
\{ G=w_1 \delta_{\theta_1}+(1-w_1)G_1: (w_1,\theta_1)\in \mathcal{U}_{\eta}, |u_{G_1}-u_{P_0}|_1\leq\eta'  \}$,
where we denote by $|u_1-u_2|_1=\int|u_1-u_2|dx$ be the $L_1$ distance between two densities $u_1$ and $u_2$.

Note that if $G \in \mathcal{U}_{\eta,\eta'}^* $ then
\[
 u_G(x)\geq w_1 b_{\mu_1,\nu_1}^*(x)\geq  c_\eta x^{A_\eta}(1-x)^{B_\eta}
\]
where
\[
\begin{split}
& c_\eta:=\frac{w_0-\eta}{B((\mu_0-\eta)(\nu_0-\eta),(1-\mu_0-\eta)(\nu_0-\eta))},  \\
& A_\eta:=(\mu_0+\eta)(\nu_0+\eta)-1, \qquad B_\eta:=(1-\mu_0+\eta)(\nu_0+\eta)-1, \\
\end{split}
\]
provided that $\mu_0-\eta,1-\mu_0-\eta,\nu_0-\eta,w_0-\eta$ are positive.
Hence, for any $\zeta>0$,
\[
 \Big [ \frac{u_0(x)}{u_G(x)}  \Big ]^{\zeta} \leq c_{\eta,\zeta}^*
\frac{u_0(x)^{\zeta}}{x^{A_\eta\zeta}(1-x)^{B_\eta\zeta}}=:g^*(x)
\]
for a suitable constant $c_{\eta,\zeta}^*$.
By assumption \eqref{cond_moment}, there is $\zeta$ such that $C_0:=\int g^*(x)u_0(x)dx<+\infty$.
Hence for such $(\eta,\zeta)$,  by Lemma
7 of \cite{GhosalvanderVaart2007},
\[
KL(u_0,u_G)\leq C_1 d_H^2(u_G,u_0)[1+ \max(0,\log(d_H^{-1}(u_G,u_0)))]
\]
where $d_H(u_G,u_0)=(\int (\sqrt{u_G}-\sqrt{u_0})^2dx)^{1/2}$ is the Hellinger distance between
$u_0$ and $u_G$. Note that the constant $C_1$
depends on $C_0,\eta$ and $\zeta$ only.
Since $d_H(u_G,u_0)^2\leq |u_G-u_0|_1$ (see, e.g.,  Corollary 1.2.1 in \cite{GhoshRamamoorthiNP})
it follows that
\[
KL(u_0,u_G)\leq C_2 |u_G-u_0|_1^{1/2}
\]
for every $G \in \mathcal{U}_{\eta,\eta'}^*$ when $\eta'$ is small enough.
Now, it is easy to check that
\[
 |u_G-u_0|_1 \leq 2|w_1-w_0|+|u_{\delta_{\theta_1}}-u_{\delta_{\theta_0}}|_1+|u_{G_1}-u_{P_0}|_1
\]
and that $|u_{\delta_{\theta_1}}-u_{\delta_{\theta_0}}|_1$ goes to zero as $|\theta_1-\theta_0| \to 0$.
Since if $\eta''\leq \eta$ then
$\mathcal{U}_{\eta'',\eta'}^* \subset \mathcal{U}_{\eta,\eta'}^*$,
using the previous results, for every $\epsilon>0$, one can find sufficiently small $\eta'$ and $\eta''\leq \eta$ such that
if $G \in \mathcal{U}_{\eta'',\eta'}^*$  then
$KL(u_G,u_0)\leq \epsilon$.
By standard argument (see e.g. the proof of Thm.2 in \cite{GhosalRoyTang2007}) if $G_1$ is in a sufficiently small  weak neighbourhood $V_{\eta'}$ of $P_0$ then
$|u_{P_0}-u_{G_1}|_1\leq \eta'$, hence
\[
\{G=w_1\delta_{\theta_1}+(1-w_1)G_1: G_1 \in V_{\eta'},(w_1,\theta_1)\in \mathcal{U}_{\eta''}\} \subset
  \mathcal{U}_{\eta'',\eta'}^* \subset \{KL(u_G,u_0)\leq \epsilon\}.
\]
Moreover, $supp(P_0) \subset supp(G_0)$ yields that
$P_0$ belongs to the support of $Dir(\phi,G_0)$ and hence $P(G_1 \in V_{\eta'})>0$,
while the fact that $\theta_1 \in supp(G_0)$ yields that $P((w_1,\theta_1)\in \mathcal{U}_{\eta''})>0$.
Using the independence of $(w_1,\theta_1)$ and $G_1$, one concludes
\[
 P(KL(u_0,u_G)\leq \epsilon)\geq P(G \in \mathcal{U}_{\eta'',\eta'}^*)
\geq P(G_1 \in V_{\eta'}) P((w_1,\theta_1)\in \mathcal{U}_{\eta''})>0.
\]
\end{proof}

\subsection{Proofs of the results of Section \ref{ConMarkov}} %

As for the proof of the results for Markovian observations, the starting 
point is a suitable  adaptation  of Theorem 1 and Lemma 3 of \cite{WuGhosal2008}.
In the Markovian setting,  the prior  $\Pi^*$ on $\mathcal{F}$ is induced via the map
$G\mapsto f_{G}(y|x) :=\int_\Theta K(y| x,\boldsymbol{\theta})G(d\boldsymbol{\theta})$,
where
$\Theta$ is  the mixing parameter space, $K(y| x,\boldsymbol{\theta})$ is a transition kernel and
$G$ has  a prior $\Pi$ on the  space  $\mathcal{M} (\Theta)$ of probability measures on $\Theta$.
Recall that in our model $\Theta=\Delta_M \times \Theta_1$, 
where $\Delta_M$ and $\Theta_1=(0,1) \times  \mathbb{R}^+$   are the combination  and  calibration
parameter spaces, respectively. 

\begin{theorem}\label{ThmWGMarkov2bis} Let $\Theta=\Delta_M \times \Theta_1$, $\Theta_1$  being a separable metric space. If for any $\eps>0$
there is a probability measure $G_\eps\in supp(\Pi)$ such that
\begin{itemize}
\item[(M1)] $\int \int \log(f_0(y|x)/f_{G_\eps}(y|x))f_0(y|x) dy \pi(dx)<\eps$,
\item[(M2)] there is a closed set $D^*_\eps$ which contains $\sup(G_\eps)$ in its interior or   
$D_\eps^*=\Delta_M \times  D_\eps$ with $D_\eps$ is closed in $\Theta_1$ 
and  $\sup(G_\eps) \subset 
\Delta_M \times  D_\eps^\circ $ ($ D_\eps^\circ$ being the interior of $D_\eps$), for which  
\[
 \int \int \log\Big( \frac{f_{G_\eps}(y|x)}{\inf_{\theta \in D_\eps^*} K(y|x,\theta)} \Big ) f_0(y|x)dy\pi(dx)<+\infty,
\]
\item[(M3)] $\inf_{ (x,y) \in C} \inf_{\theta \in D_\eps^*} K(y|x,\theta)>0$ for every compact set $C \subset \mathcal{Y} \times \mathcal{Y}$,
\item[(M4)] $\{ \theta \mapsto K(y|x,\theta): (x,y) \in C\}$
is uniformly equicontinuous on $D_\eps^*$,
\end{itemize}
then $f_0 \in KL(\Pi^{*})$.
\end{theorem}

\begin{proof}
The proof follows the same lines of the  proofs of Theorem 1 and Lemma 3 of \cite{WuGhosal2008}.
\end{proof}

\begin{theorem}\label{teoconsMarkov2}
Under the assumption of Theorem \ref{teoconsMarkov3}
 then  $f_{0}\in KL(\Pi^{*})$.
  \end{theorem}

\begin{proof} We apply   Theorem \ref{ThmWGMarkov2bis} with 
 $\Theta_1=(0,1) \times  \mathbb{R}^+$, 
 $
K(y|x, \boldsymbol{\theta})=b^*_{\mu,\nu}(H(y|x,\boldsymbol{\omega}))h(y|x,\boldsymbol{\omega})
$
with $\boldsymbol{\theta}=(\boldsymbol{\theta}_{p},\boldsymbol{\theta}_{c})$ 
  and 
$G_\eps(d\boldsymbol{\omega}d\mu d\nu):=G^*(d\boldsymbol{\omega}d\mu d\nu)$. 
%$G_\eps(d\boldsymbol{\omega}d\mu d\nu):=\delta_{\boldsymbol{\omega^*}}(d\boldsymbol{\omega})\times P_0(d\mu d\nu)$.
In this case, since $f_{G_\eps}=f_{G^*}=f_0$,  M1 is automatically satisfied, so it remains to check the other assumptions. 

{\it Verification of M2 of Theorem  \ref{ThmWGMarkov2bis}.}
%We assume first that $\boldsymbol{\omega}$ is on the boundary of $\Delta_M$, which is the more delicate case. 
One can find a compact set $D_\eps$  in $(0,1)\times (0,+\infty)$ such that
%$D_\eps$ contains support of $P_0$ in its interior and set
the support of $G_0$ is contained in 
$\Delta_M \times  D_\eps^\circ$. At this stage one can check that
\[
\inf_{(\mu,\nu) \in D_\eps} 
H(y|x,\boldsymbol{\omega})^{\mu\nu}(1-H(y|x,\boldsymbol{\omega}))^{(1-\mu)\nu} 
\geq H(y|x,\boldsymbol{\omega})^{A}(1-H(y|x,\boldsymbol{\omega}))^{A}
\geq R_-(y|x)^A 
\]
where $A:=\max\{ \sup\{\mu\nu: (\mu,\nu) \in D_\eps\}, \sup\{(1-\mu)\nu: (\mu,\nu) \in D_\eps\}\}>0$.
Hence, setting $D_\eps^*:=\Delta_M \times  D_\eps$,
\[
%\begin{split}
I(y|x) :=\inf_{(\boldsymbol{\omega},\mu,\nu) \in D_\eps^*} K(y|x,\boldsymbol{\omega},\mu,\nu)  
\geq
\frac{C_1r_-(y|x)  R_-(y|x)^{A} }{R_+(y|x)} =:I^*(y|x)%\\
%\end{split}
\]
where $C_{1}= \inf\{ 1/B(\mu\nu,(1-\mu)\nu)
:(\mu,\nu) \in D_\eps \}  >0$.
Clearly $f_{G_\eps}(y|x) \geq  I(y|x)$ and hence $\log( {f_{G_\eps}(y|x)}/I_\eps(y|x))\geq 0$, moreover, 
recalling that $f_{G_\eps}(y|x)= f_0(y|x)$, 
one can write
\[
\begin{split}
\int \log & \Big( \frac{f_{G_\eps}(y|x)}{I(y|x)}  \Big ) f_0(y|x)dy  
\leq \int \log\Big( \frac{f_{0}(y|x)}{I^*(y|x)}  \Big ) f_0(y|x)dy \\
%&\leq  \int \log\Big( \frac{u_{0}(H(y|x,\boldsymbol{\omega}^*))r^+(y|x) }{I^*(y|x)}  \Big ) 
%H(y|x,\boldsymbol{\omega}^*))h(y|x,\boldsymbol{\omega})^* dy \\
& \leq  \int  |\log (f_0(y|x))| f_0(y|x)dy + \int  |\log(I^*(y|x)) |  f_0(y|x)dy . \\
\end{split}
\]
Combining this with the assumptions, 
one gets 
\[
0< \int \int \log\Big( \frac{f_{G_\eps}(y|x)}{\inf_{(\boldsymbol{\omega},\mu,\nu) \in D_\eps} K(y|x,\boldsymbol{\omega},\mu,\nu)} \Big ) f_0(y|x)dy\pi(dx)<+\infty.
\]

{\it Verification of M3 of Theorem  \ref{ThmWGMarkov2bis}.} Using the positivity assumption of the densities $h_m$, it follows immediately that, for every compact set $C$,
$ \inf_{(x,y) \in C} \inf_{(\boldsymbol{\omega},\mu,\nu) \in D_\eps} K(y|x,\boldsymbol{\omega},\mu,\nu)
>0$.

{\it Verification of M4 of Theorem  \ref{ThmWGMarkov2bis}.} 
The function $(\boldsymbol{\omega},\mu,\nu,x,y) \mapsto K(y|x,\boldsymbol{\omega},\mu,\nu)$
is continuous and hence uniformly continuous on the compact set $C\times D_\eps$. It follows
that the family $\{ (\boldsymbol{\omega},\mu,\nu) \mapsto K(y;\boldsymbol{\omega},\mu,\nu): (x,y) \in C\}$
is uniformly equicontinuous on $D_\eps$.
\end{proof}

\begin{proof}[Proof of Theorem \ref{teoconsMarkov3}]
Let $A:\{f: \int|\int g(y)[f(y|x)-f_0(y|x)]dy|  \lambda(dx)>\eps \}$ where $g$ is a bounded continuous function, say $|g(y)|\leq 1$. 
Let $C \subset \CY$ a compact set such that $ \lambda(C^c) <\eps/4$ and set $M=2 \lambda(C)$. If $f \in A$ then
$\int_C|\int g(y)[f(y|x)-f_0(y|x)]dy|  \lambda(dx)>\eps/2$ and $\sup_{x \in C} \int g(y)[f(y|x)-f_0(y|x)]dy >\eps/M$. 
Let $x_f=argmax\{\int g(y)[f(y|x)-f_0(y|x)]dy :x \in C  \}$, then if $x \in C$
\[
 \Big |\int g(y)[f(y|x)-f_0(y|x)]dy  \Big | > \frac{\eps}{M}-R(f,x)
\]
where 
$
 R(f,x):= \Big |\int g(y)[f(y|x)-f(y|x_f)]dy  \Big|+ \Big|\int g(y)[f_0(y|x)-f_0(y|x_f)]dy  \Big|$.
Now for every compact set $K \subset (0,1) \times \RE^+$ and every compact set $C' \subset \CY$
the functions $K \times C  \times C' \times \Delta_M \ni (\mu,\nu,x,y,\boldsymbol{\omega})
\mapsto b_{\mu,\nu}^*(H(y|x,\boldsymbol{\omega}))$  and
$C \times  C' \times \Delta_M \ni (x,y,\boldsymbol{\omega})
\mapsto h(y|x,\boldsymbol{\omega})$  are uniformly continuous. Hence for every $\eps'$ there is $\delta=\delta(\eps,K,C,C')$ such that
\[
\begin{split}
& \sup_{(\mu,\nu,y,\boldsymbol{\omega}) \in K \times  C' \times \Delta_M}
 |b_{\mu,\nu}^*(H(y|x,\boldsymbol{\omega}))-b_{\mu,\nu}^*(H(y|x',\boldsymbol{\omega}))|\leq \eps',  \\ 
 & \sup_{(y,\boldsymbol{\omega}) \in C' \times \Delta_M} |h(y|x,\boldsymbol{\omega})-h(y|x',\boldsymbol{\omega})|
\leq \eps' \\
\end{split}
\]
whenever $|x-x'|\leq \delta$. Moreover $S:=\sup_{(\mu,\nu,x,y,\boldsymbol{\omega}) \in K \times C \times  C' \times \Delta_M}
 |b_{\mu,\nu}^*(H(y|x,\boldsymbol{\omega}))| <+\infty$.
Now for $f$ in the support of the prior, recalling that $f(y|x)=\int b_{\mu,\nu}^*(H(y|x,\boldsymbol{\omega})) h(y|x,\boldsymbol{\omega}) G(d\mu d\nu d \boldsymbol{\omega})$, 
for some $G$ with $supp(G) \subset supp(G_0) \subset K \times \Delta_M$, it follows that
whenever $|x-x_f |\leq \delta$ one can write
\[
\begin{split}
  \Big | & \int_{C'} g(y)[f(y|x)-f(y|x_f)]dy  \Big|  \\  
& \leq \int_{K \times \Delta_M}
\int_{C'} \Big [ | b_{\mu,\nu}^*(H(y|x,\boldsymbol{\omega})) -b_{\mu,\nu}^*(H(y|x_f,\boldsymbol{\omega}))| h(y|x,\boldsymbol{\omega})  \\
& + b_{\mu,\nu}^*(H(y|x_f,\boldsymbol{\omega}))| h(y|x,\boldsymbol{\omega}) - h(y|x_f,\boldsymbol{\omega})|  \Big ]
dy  G(d\mu d\nu d \boldsymbol{\omega})\\
& \leq  \eps' \int_{K \times \Delta_M} 
\int_{C'}    h(y|x,\boldsymbol{\omega})[1 + \frac{  S  }{h(y|x,\boldsymbol{\omega}) } ] dy  G(d\mu d\nu d \boldsymbol{\omega}) \\
& \leq  \eps' \int_{K \times \Delta_M} \int_{C'} h(y|x,\boldsymbol{\omega})(1+ \frac{S}{c^*}) dy G(d\mu d\nu d \boldsymbol{\omega})\leq  \eps'(1+ \frac{S}{c^*})
\end{split}
\]
where $c^*=\min_{m=1,\dots,M} \inf_{ (x,y) \in C\times C'}f_m(y|x)>0$ by assumption.% \eqref{lowerboundh0tris}.

Consider now $C'=[-L,L]$, then by a simple change of variables 
\[
\int_{(C')^c} f(y|x)dy = \int_{K \times \Delta_M}\int_{(0,H(-L|x,\boldsymbol{\omega}))\cup (H(L|x,\boldsymbol{\omega}),1)}
 b_{\mu,\nu}^*(z) dz  G(d\mu d\nu d \boldsymbol{\omega})
\]
hence by \eqref{upperboundF}
\[
  \int_{(C')^c} g(y)[f(y|x)-f(y|x_f)]dy \leq 2 \int_{K \times \Delta_M}\int_{(0,\eta)\cup (1-\eta,1)}
 b_{\mu,\nu}^*(z) dz  G(d\mu d\nu d \boldsymbol{\omega}).
\]
Now it is easy to see that 
$\sup_{(\mu,\nu) \in K} b_{\mu,\nu}^*(z) \leq A b^*(z)$ for a suitable density $b^*$ and a suitable constant $A$. 
Hence, for every $\eps''$ one can find $C'=[-L,L]$ such that 
\[
 |\int_{(C')^c} g(y)[f(y|x)-f(y|x_f)]|dy \leq \eps''
\]
for every $f$ in the support of $\Pi$. Combining this statement with the first part of the proof 
it follows that there is $\delta$ such that for every $f \in A$ which is in the support of $\Pi$ one has 
$R(f,x) \leq \eps/2M$ whenever $|x-x_f| \leq \delta$. Hence for such $f$ and $x$ 
one has $|\int g(y)[f(y|x)-f_0(y|x)]dy  | > \eps/(2M)$.
Thus one can partition $C=\cup_{r=1}^{R} C_r$ and $A=\cup_{r=1}^{R} A_r$ such that the length of $C_r$ is
at most $\delta$ and 
if $f \in A_r$ then $\inf_{x \in C_r} |\int g(y)[f(y|x)-f_0(y|x)]dy  | > \eps/(2M)$. 
Recalling that by Theorem \ref{teoconsMarkov2} we already know that 
the Kullback-Leibler property holds, 
by Corollary 4.1 in \cite{TangGhosal07}
it follows that $\Pi^*_n(A_r) \to 0$ almost surely and hence $\Pi^*_n(A) \to 0$ a.s..
\end{proof}

{

\begin{lemma}\label{lemma:gaussiantailmixture}
If $0<\sigma_{-} < \min_{m=1,\dots,M}  \sigma_m \leq \max_m \sigma_m< \sigma_{+}$, then
for every $y$ in $\mathbb{R}$ and $\mu_1,\dots,\mu_M$
\begin{equation}\label{dis_gauss}
\frac{\sigma_-}{\sigma_+}e^{-\frac{C^-}{2} \mu_m^2} \varphi(y|0,\sigma_-^2)  \leq \varphi(y|{\mu_m,\sigma_{m}^2}) \leq \frac{\sigma_+}{\sigma_-}
e^{\frac{C^+}{2} \mu_m^2} \varphi(y|0,\sigma_+^2) 
\end{equation}
with $C^+:=\max_{m=1,\dots,M} (\sigma_m^2-\sigma_-^2)^{-1}$
and $C^+:=\max_{m=1,\dots,M} (\sigma_+^2-\sigma_m^2)^{-1}$.
\end{lemma}

\begin{proof} 
It is easy to check that if $\sigma_a^2 \leq \sigma_b^2$ 
\[
\begin{split}
\frac{\varphi(y|\mu_a,\sigma_a^2)}{\varphi(y|\mu_b,\sigma_b^2)}  & =
\frac{\sigma_b}{\sigma_a} e^{-\frac{1}{2 \sigma^2_*}(y-\mu^*)^2}e^{ -\frac{1}{2} \Big[ \frac{\mu_a^2}{\sigma_a^2}
-\frac{\mu_b^2}{\sigma_b^2}-\big(  \frac{\mu_a}{\sigma_a^2}- \frac{\mu_b}{\sigma_b^2}\big)^2 \sigma_*^2 \Big ] } \\
& \leq \frac{\sigma_b}{\sigma_a}e^{ -\frac{1}{2} \Big[ \frac{\mu_a^2}{\sigma_a^2}
-\frac{\mu_b^2}{\sigma_b^2}-\big(  \frac{\mu_a}{\sigma_a^2}- \frac{\mu_b}{\sigma_b^2}\big)^2 \sigma_*^2 \Big ] } \\
\end{split}
\]
with $1/\sigma_*^2:=1/\sigma_a^2-1/\sigma_b^2$ and $\mu_*:=\sigma^2_*(\mu_a/\sigma_a^2-\mu_b/\sigma_b^2)$
By using the previous inequalit  with $\mu_a=0,\sigma_a^2=\sigma_-^2$ and $\mu_b=\mu_m,\sigma_b^2=\sigma_m^2$  (
$\mu_b=0,\sigma_b^2=\sigma_+^2$ and $\mu_a=\mu_m,\sigma_a^2=\sigma_m^2$, respectively) after  some algebra 
one gets the first (second, respectively) inequality in the thesis. 
\end{proof}

{\bf Details of Example \ref{exMix}.}
With the notation of Example \ref{exMix}, 
if  $h(y_t|y_{t-1},\boldsymbol{\omega})=\sum_{m=1}^M\omega_{m} \varphi(y_t |{\mu_m +\phi_m y_{t-1}  ,\sigma_m^2}) $
it follows from Lemma \ref{lemma:gaussiantailmixture} that
\[
 C_1 e^{-C_2 x^2} \varphi(y|0,\sigma_-^2) \leq h(y|x,\boldsymbol{\omega}) \leq  C_3  e^{C_4 x^2} \varphi(y|0,\sigma_+^2),
\]
 for suitable constants $C_1,C_2,C_3,C_4$.
Hence it is easy to see that \ref{upperboundF} of  Theorem  \ref{teoconsMarkov3} holds true. 
In order verify \eqref{eq:condmarkov1} of Theorem  \ref{teoconsMarkov3} we set $r^-(y|x):= C_1 e^{-C_2 x^2} \varphi(y|0,\sigma_-^2)$,
\[
R_+(y|x) := C_3^2 e^{2C_4 x^2} \Phi(y|x,\sigma_+^2)(1- \Phi(y|x,\sigma_+^2)), 
\]
and
\[
R_-(y|x) := C_1^2 e^{2C_2 x^2} \Phi(y|x,\sigma_-^2)(1- \Phi(y|x,\sigma_-^2)).
\]
Using (ii) of Example \ref{ex:completeincomplete}, after some computations, one can  show that 
\[
|\log(R_+(y|x))|+|\log(R_-(y|x))|+|\log(r_-(y|x))| \leq K[1 + x^2+y^2]
\]
for a suitable constant $K$.  Analogously, using again Lemma \ref{lemma:gaussiantailmixture} and recalling the definition of $f_0$, one can prove that
\[
|\log(f_0(y|x))| \leq K'[1 + x^2+y^2] \quad \text{and} \quad f_0(y|x) \leq K'' \varphi_0(y|x,\sigma^2_*) 
\]
for a suitable $\sigma^2_*$ and $K', K''$.
% \int\int \Big [|\log(R_-(y|x))\Big|+\Big|\log \Big (\frac{R_+(y|x)}{r_-(y|x)} \Big )\Big| + \log(f_0(y|x))  \Big]f_0(y|x) dy \pi(dx)<+\infty,
Hence, to check  that \eqref{eq:condmarkov1} in Theorem  \ref{teoconsMarkov3} is satisfied it suffices to prove that 
\[
 \int\int   [1 + x^2+y^2] \varphi_0(y|x,\sigma^2_*)  dy \pi(dx) \leq K^{'''}\int [1+x^2]  \pi(dx)<+\infty.
 \]
 Given the specific form of the mAR process we are dealing with, the existence of a stationary solution 
 with finte second moment is well-known. 
 }
 
\section{Further simulation results}\label{appendix_simulation}
In order to complete the simulation study, the following 
combination-calibration schemes have been considered in addition to 
the infinite component beta mixture (BM$_{\infty}$).

\begin{itemize}
\item Beta-transformed linear pool (BM$_1$)
\begin{equation*}
f(y|\btheta)=f_{\alpha,\beta}\left(H(y|\omega)\right)h(y|\omega),
\end{equation*}
where $\btheta=(\alpha,\beta,\omega)$, $h(y|\omega)=\omega
\varphi(y|-1,1)+(1-\omega) \varphi(y|2,1)$ and $H(y|\omega)=\omega
\nu(y|-1,1)+(1-\omega) \nu(y|2,1)$.
\item Two-component finite beta mixture model (BM$_2$)
\begin{equation*}
f(y|\btheta) = w f_{\alpha_1,\beta_1}\left(H(y|\omega_1)\right)
h(y|\omega_1)+(1-w)f_{\alpha_2,\beta_2}\left(H(y|\omega_2)\right)h(y|\omega_2),
\end{equation*}
where $\btheta = (w, \alpha_1, \alpha_2, \beta_1, \beta_2, \omega_1,
\omega_2)$, and $h(y|\omega)$ and $H(y|\omega)$ have been defined as
in the BC model.
\end{itemize}
In the simulation experiments, the hyper\-parameter setting for the BC
and BMC model is $\xi_{j\mu}=2$, $\xi_{j\nu}=0.1$ and
$\xi_{j\omega}=1$, and $\xi_{jw}=1$, $j=1,2$. The priors are
informative, but with a large prior variance, thus one can expect
posterior inference should not be affected by the hyper\-parameter
settings. Our experiments show that the results, in terms of
calibration, do not change when considering less informative prior
settings, and secondly that the use of improper prior distributions in
mixtures model, even if possible, still remains an open issue. See
e.g.~\cite{Was00} for a discussion on the use of improper prior in
mixture modelling.

%For expository purposes we arbitrarily set, in Table
%\ref{tab:mot}, $\alpha_{1}=\alpha$, $\beta_1=\beta$ and $w=1$ for the
%BC models and $\omega_1=\omega$ for the models with common linear
%combination.
\subsection{Multimodality}
Figure \ref{fig:mot1} shows the empirical cdfs of different sequences of
probability integral transform (PIT).  In all the experiments, the PIT
of the non-calibrated model (red lines) is far from the standard
uniform (black lines).  In these datasets, the BC clearly lacks
calibration.  The BC cdf (green line) is closer to uniformity than the
NC model, but it has difficulties in deforming the combination density
some parts of the support.

\begin{figure}[t]
\begin{center}
\setlength{\tabcolsep}{1pt}
\begin{tabular}{ccc}
{\small $\bp=(1/5,1/5,3/5)$} &&
{\small $\bp=(1/7,1/7,5/7)$} \\
\includegraphics[width=5cm]{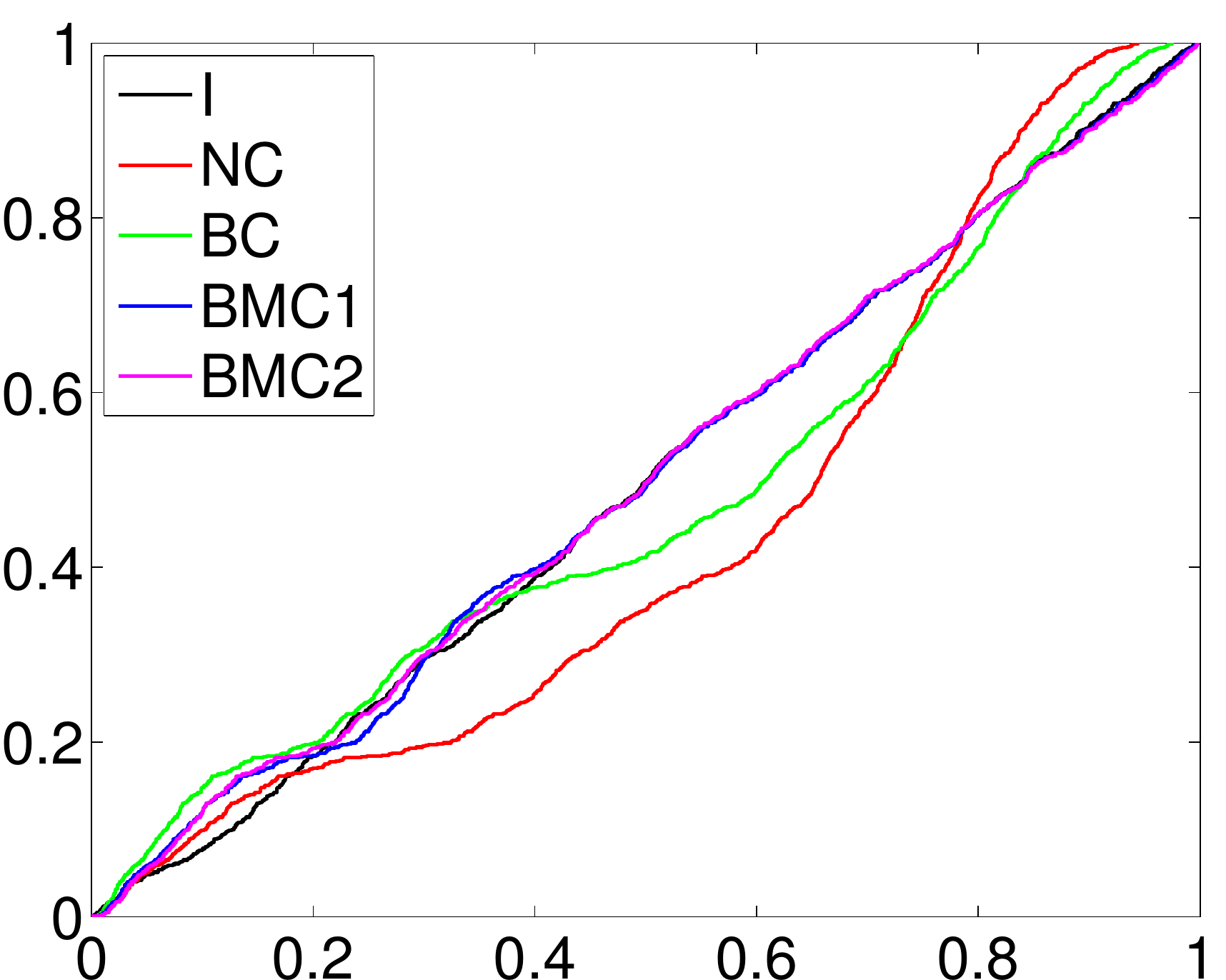} &&
\includegraphics[width=5cm]{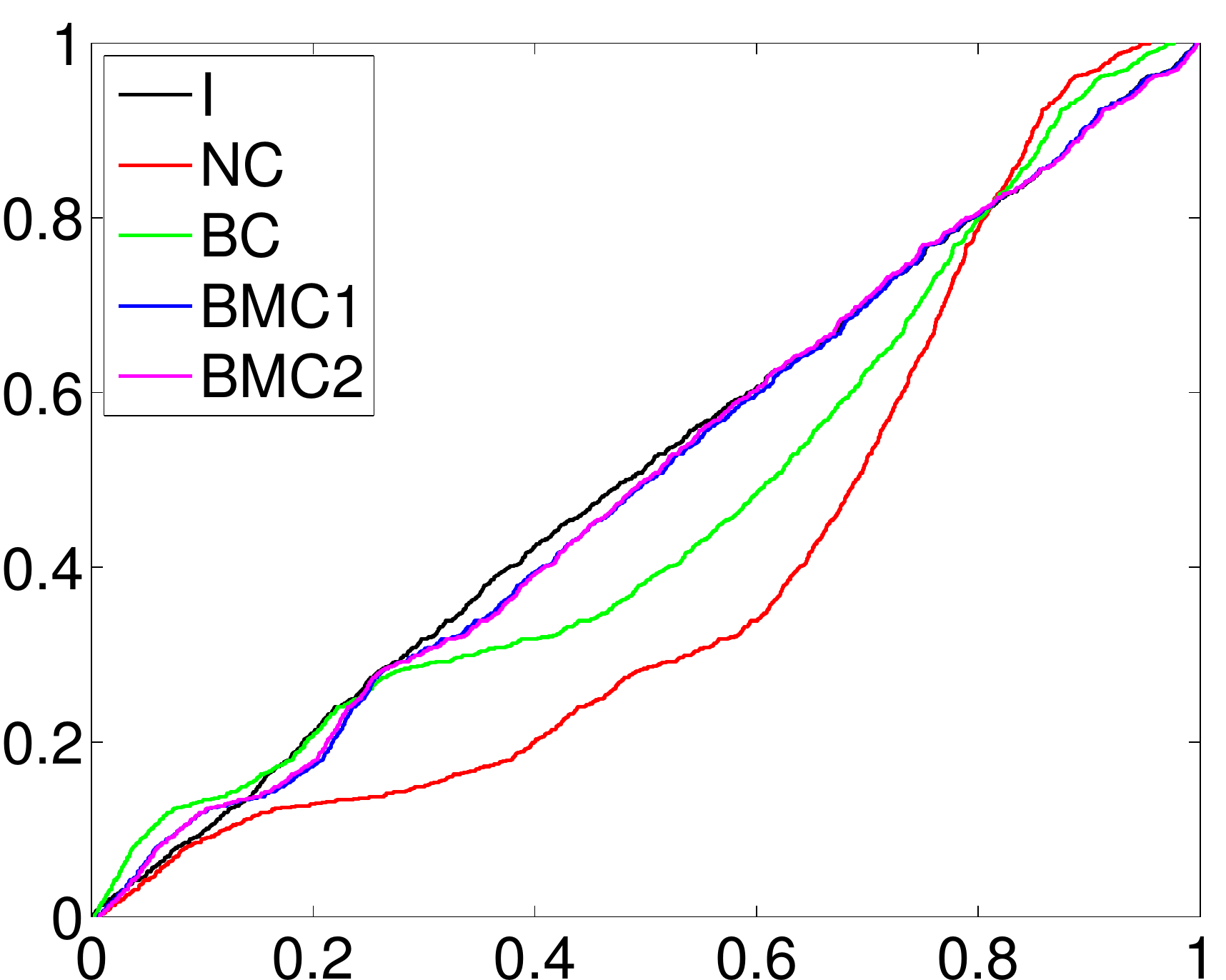} \\
{\small $\bp=(3/5,1/5,1/5)$} &&
{\small $\bp=(5/7,1/7,1/7)$ \rule{0mm}{6mm}} \\
\includegraphics[width=5cm]{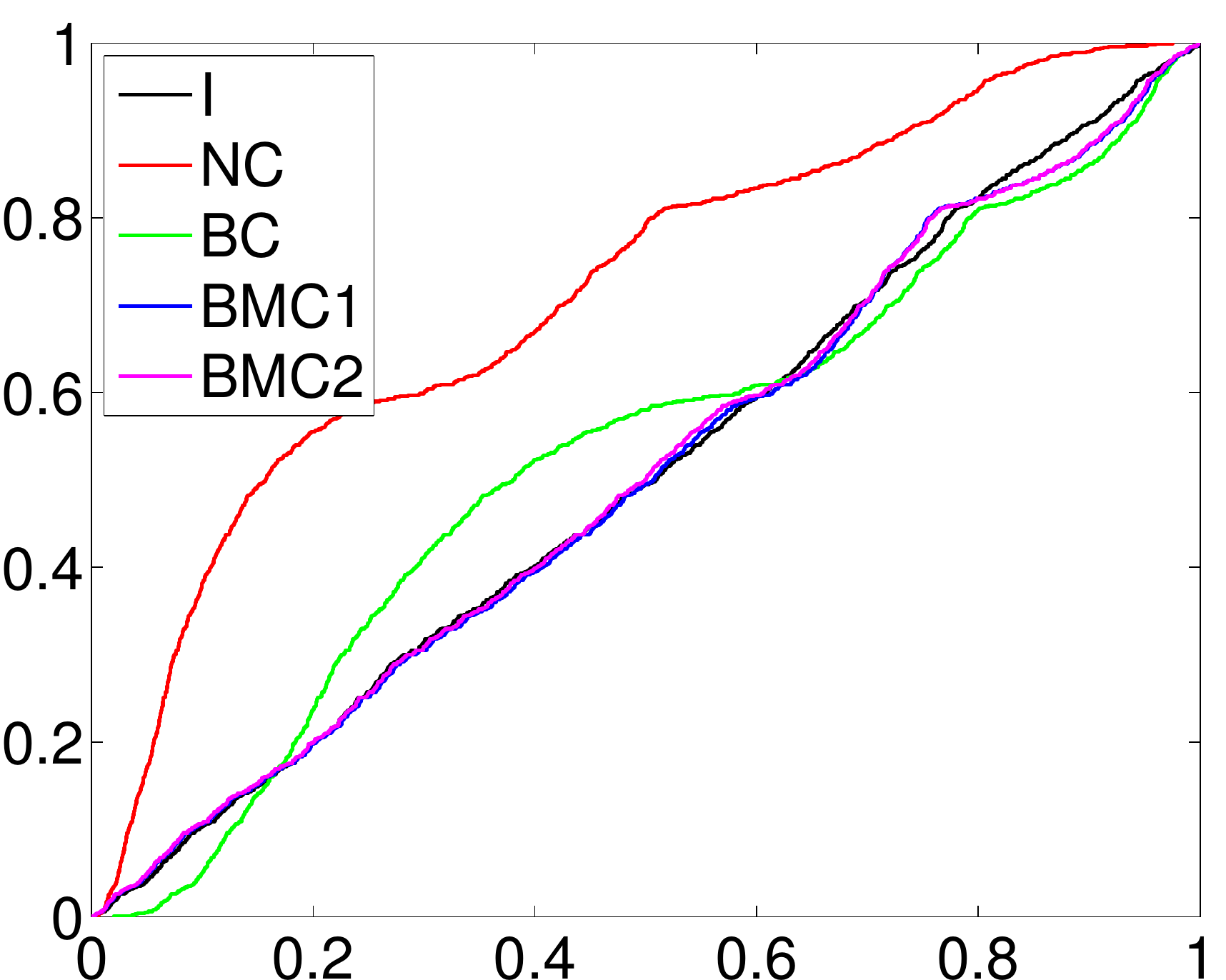} &&
\includegraphics[width=5cm]{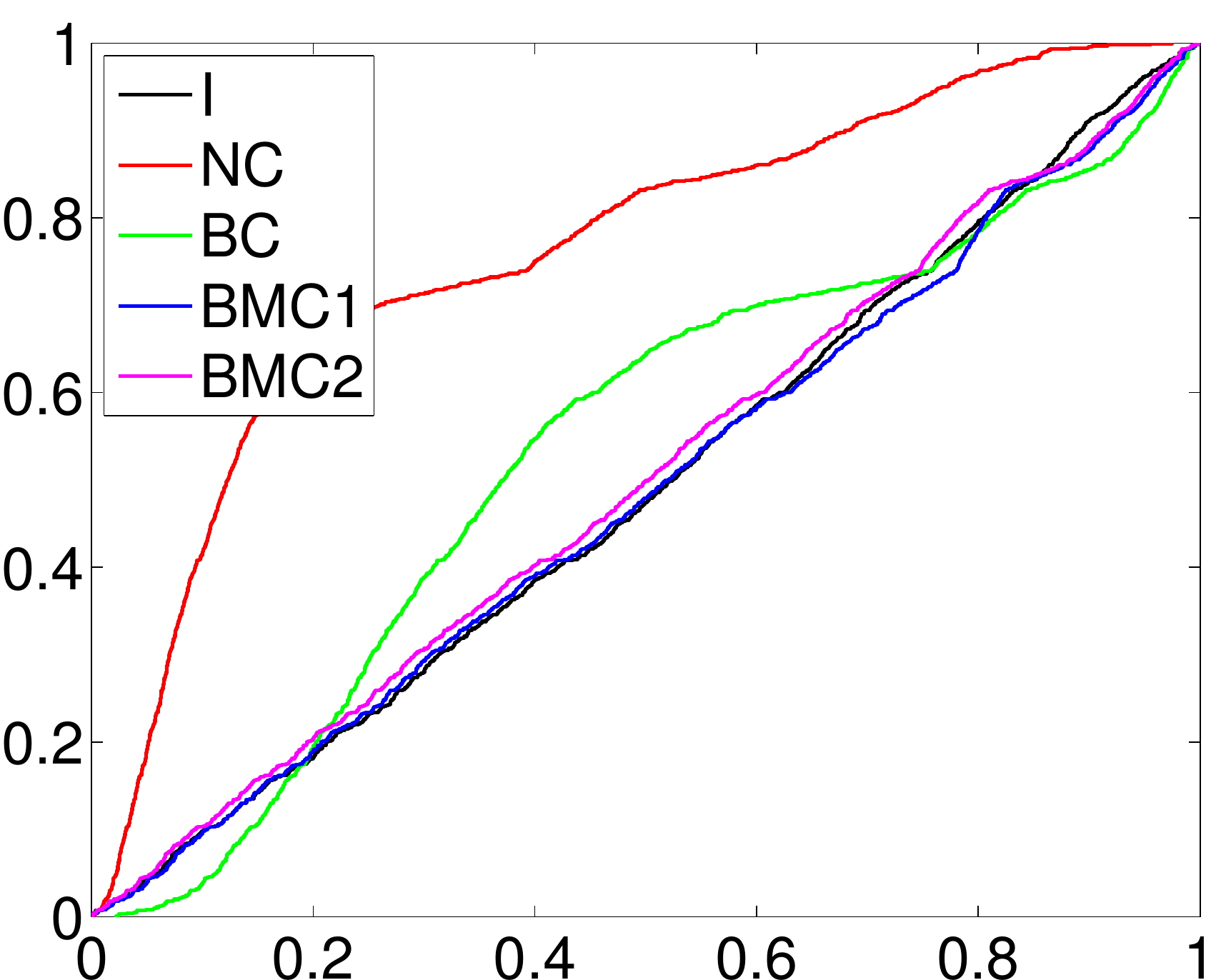}\\
\end{tabular}
\end{center}
\caption{PITs cdf for different calibration models and different datasets.}
\label{fig:mot1}
\end{figure}

More specifically, the two-component beta calibrations are able to
achieve a more flexible deformation of the cdf linear combination
providing a calibrated cdf (blue and magenta lines) which is close to
the uniform cdf.  Figure \ref{fig:mot3} shows the results of the
calibration and combination procedure decomposed along the different
components of the mixture.  As an example consider the first dataset,
generated with $\bp=(1/5,1/5,3/5)$.  The solid and dashed blue lines
in the top-left plot of Figure \ref{fig:mot3} show the contribution of
the first and second component respectively of the BMC1 mixture model
to the calibration of the density.  The first component mainly calibrates
the pdf on the positive part of the support and the second
component calibrates the pdf on the negative part of the support.
Both components assign the same weights ($\omega=0.449$) to the first model 
in the pool, i.e. $\cN(-1,1)$.  This weight is higher than in the BC model, which
has a less flexible calibration function and thus assigns a lower
weight $\omega=0.202$ to the first model in the pool.  The solid and
dashed magenta lines in the top-left plot of Figure \ref{fig:mot3}
show a behaviour similar to the BMC1 components.

\begin{figure}[t]
\begin{center}
\vspace{-20pt}
\setlength{\tabcolsep}{1pt}
\begin{tabular}{ccc}
{\small$\bp=(1/5,1/5,3/5)$}& $\,$&{\small$\bp=(1/7,1/7,5/7)$} \\
\includegraphics[width=5cm]{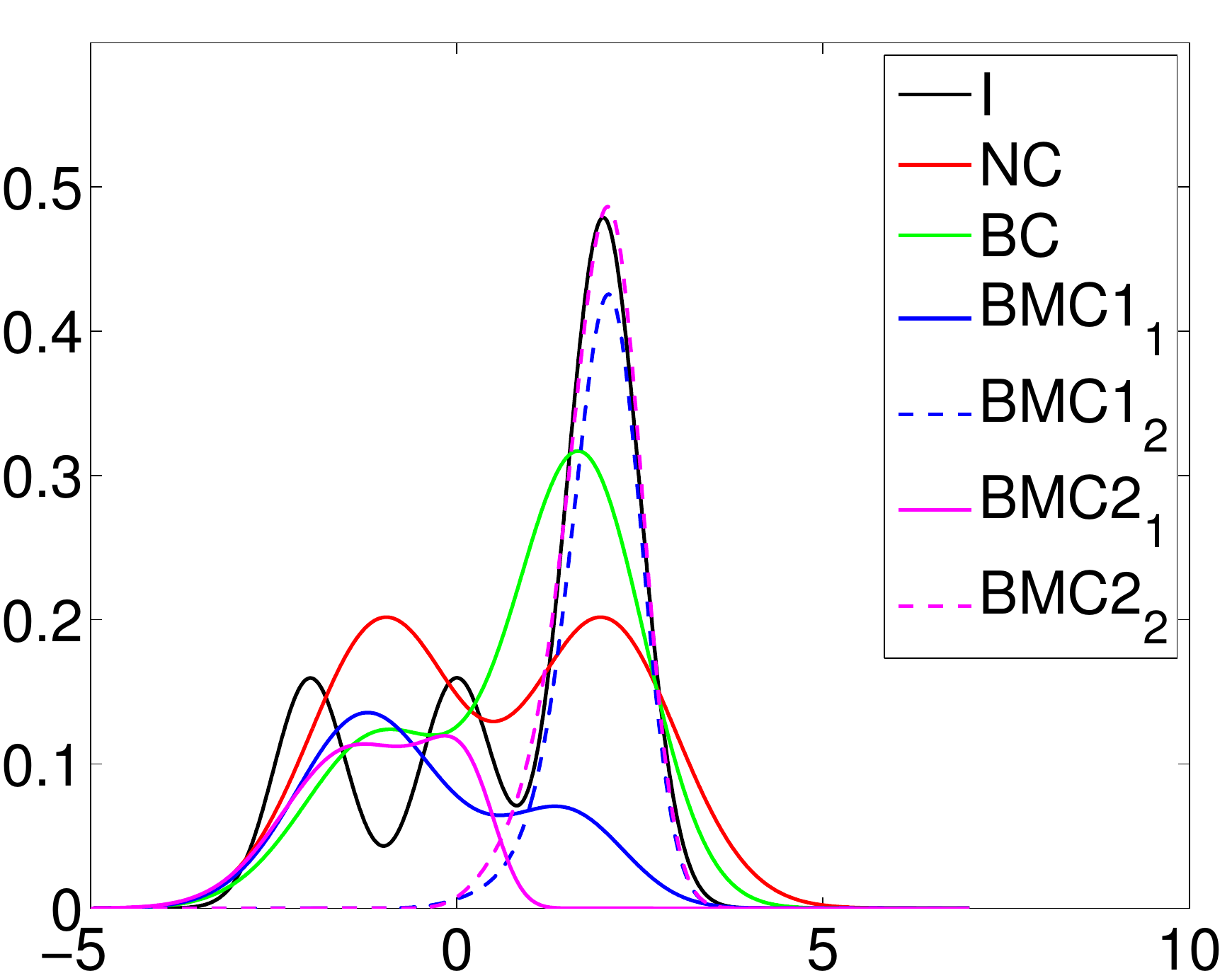}& $\,$ &
\includegraphics[width=5cm]{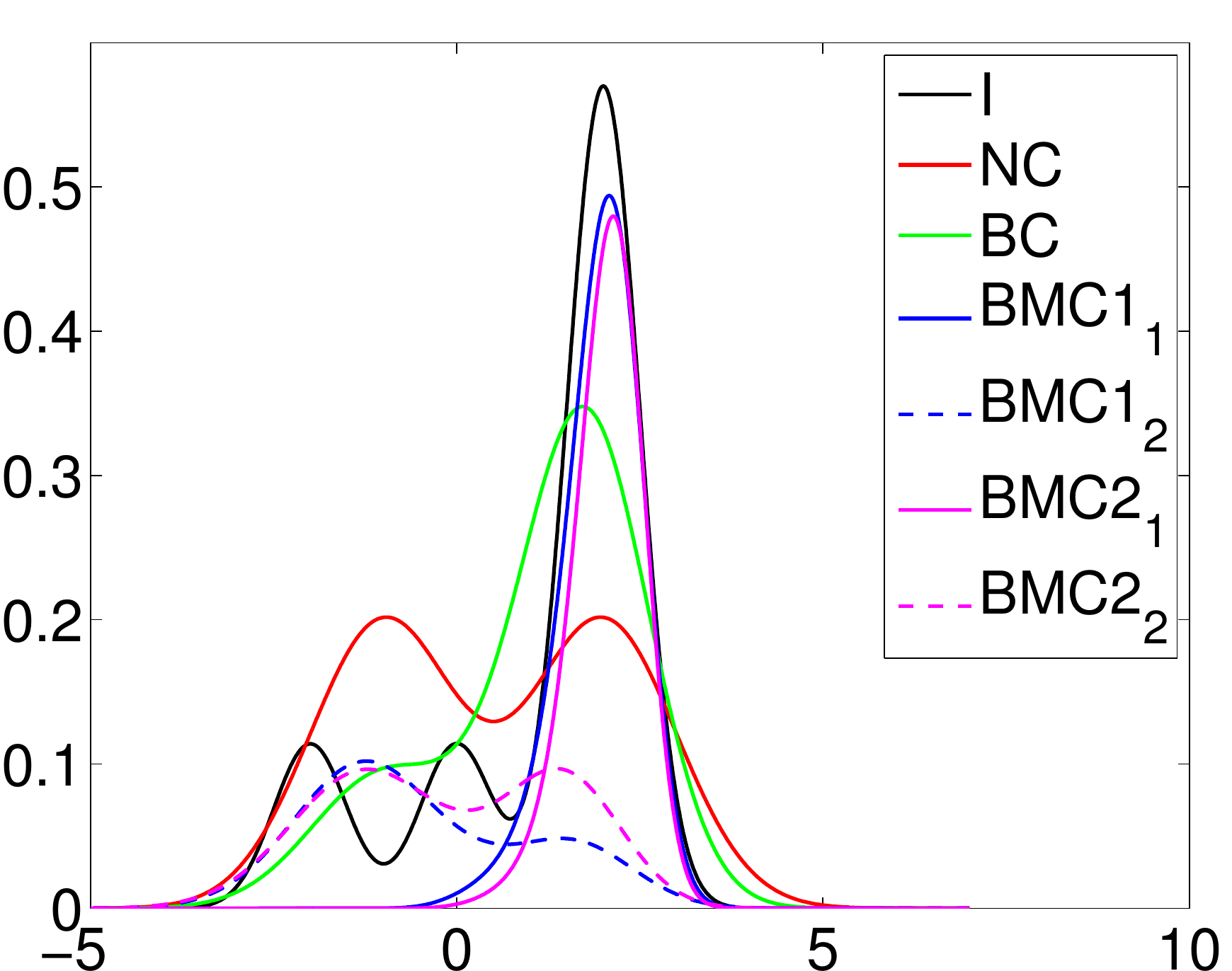} \\
{\small$\bp=(3/5,1/5,1/5)$}& $\,$&{\small$\bp=(1/7,1/7,5/7)$} \\
\includegraphics[width=5cm]{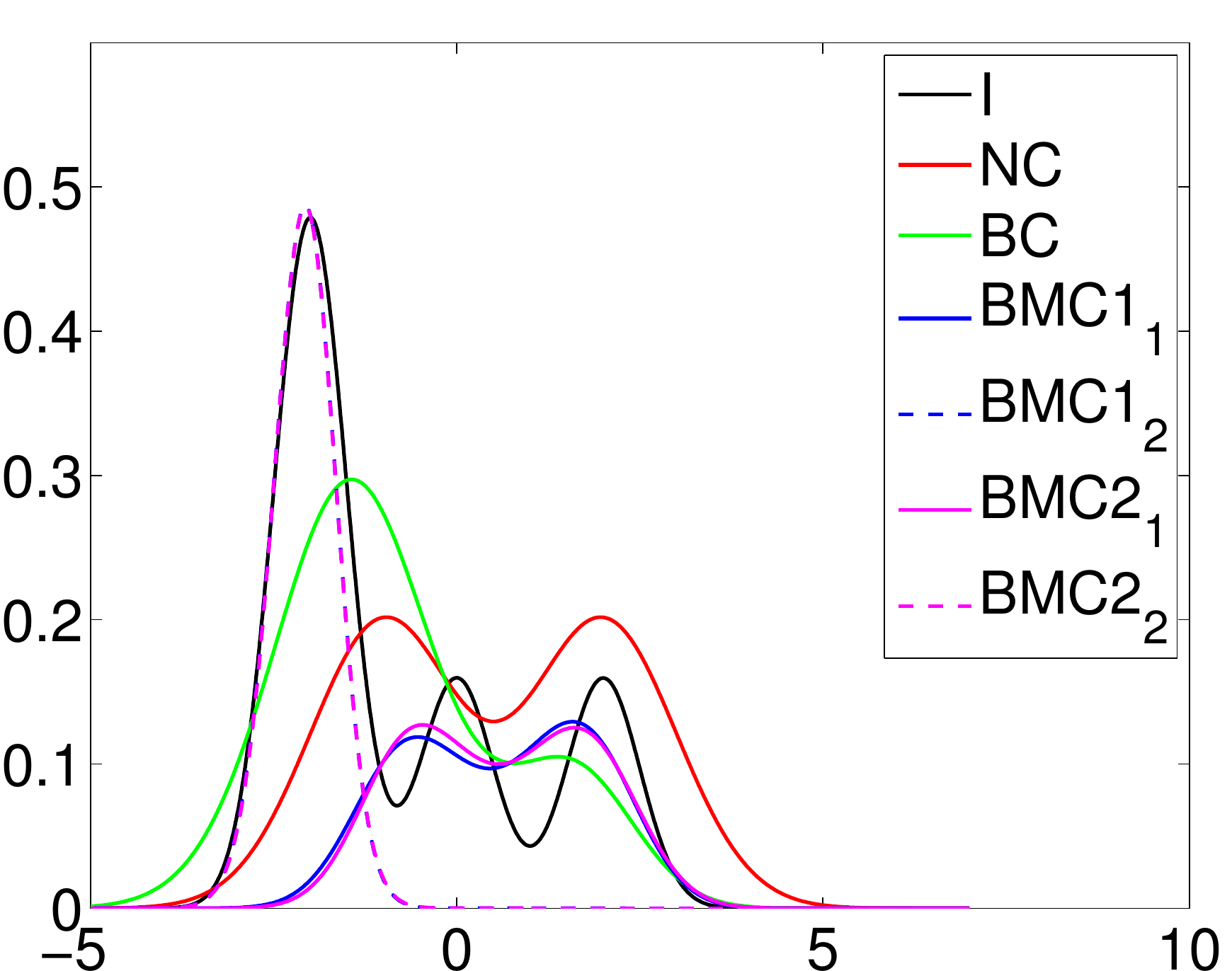} & $\,$ &
\includegraphics[width=5cm]{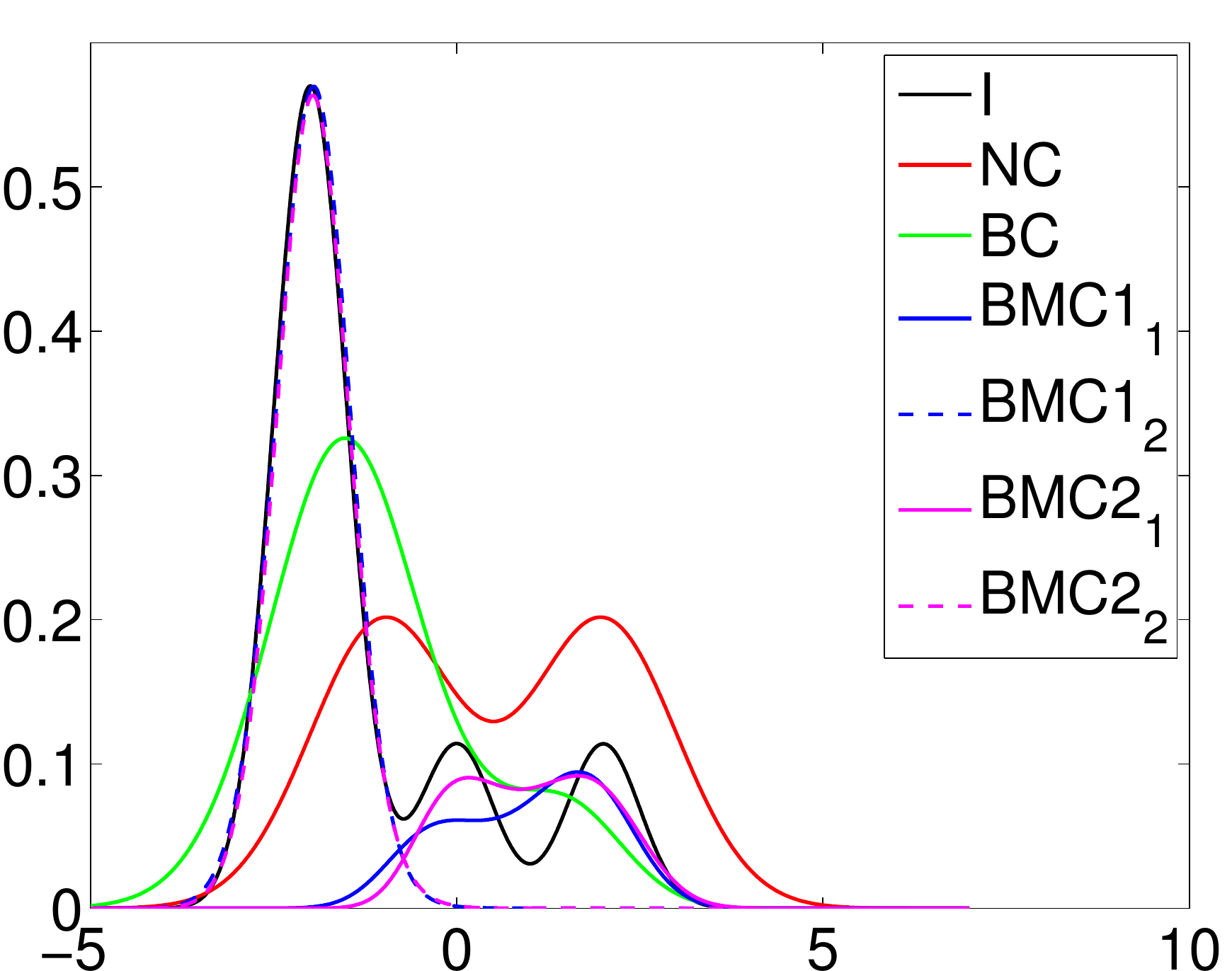}\\
\end{tabular}
\caption{Contribution of the calibration components for different
  models BC, BMC1 (first and second mixture component, $BMC1_1$ and
  $BMC1_2$) and BMC2 (first and second mixture component, $BMC2_1$ and
  $BMC2_2$), and different datasets.}  \label{fig:mot3}
\end{center}
\end{figure}

The first BMC2 component assigns weight $\omega_{1}=0.043$ to the first 
model in the pool. This means that the calibration on the negative part 
of the support set is done mainly using the predictive distribution of the second
model, $\cN(2,1)$.  The calibration of the positive part of the set is
obtained thanks to the second BMC2 component which assigns
weight $\omega_{2}=0.667$ to the first model in the pool.

\subsection{Heavy Tails}
Fig \ref{fig:mot1.heavy} focuses on the
right tail of the predictive pdf and shows results for the calibrated
and beta calibrated PITs cdf and their 99\% HPD. There is strong
evidence of the difficulties of the BC model in calibrating the tails.
The BC underestimates the tail probability and over-estimates the
central part of the distribution.  The BMC1 and BMC2 models are able
instead to provide well-calibrated PITs on the tails of the
distribution.

\begin{figure}[t]
\vspace{-20pt}
\begin{center}
\setlength{\tabcolsep}{1pt}
\begin{tabular}{ccc}
\includegraphics[width=5cm]{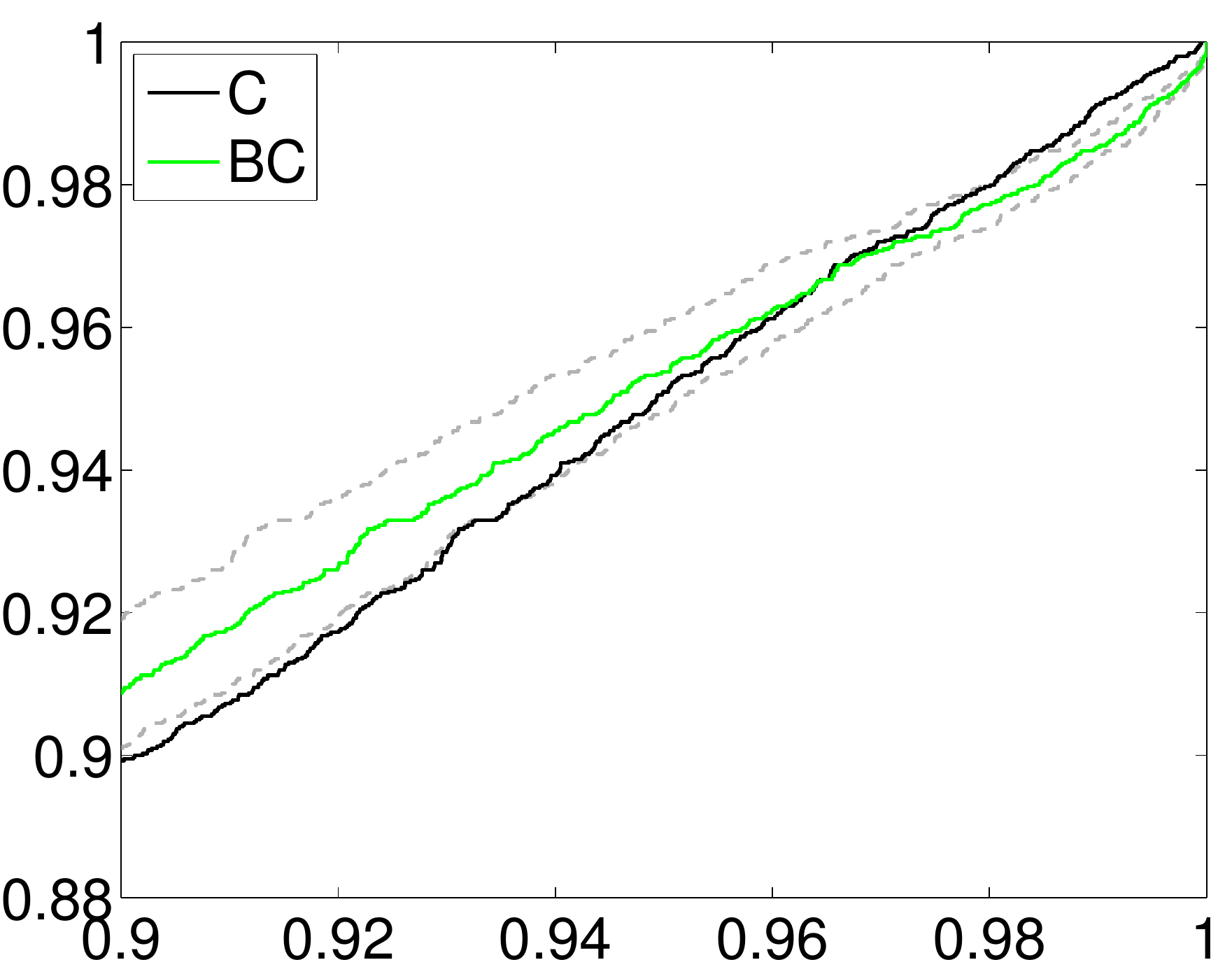} \\
\includegraphics[width=5cm]{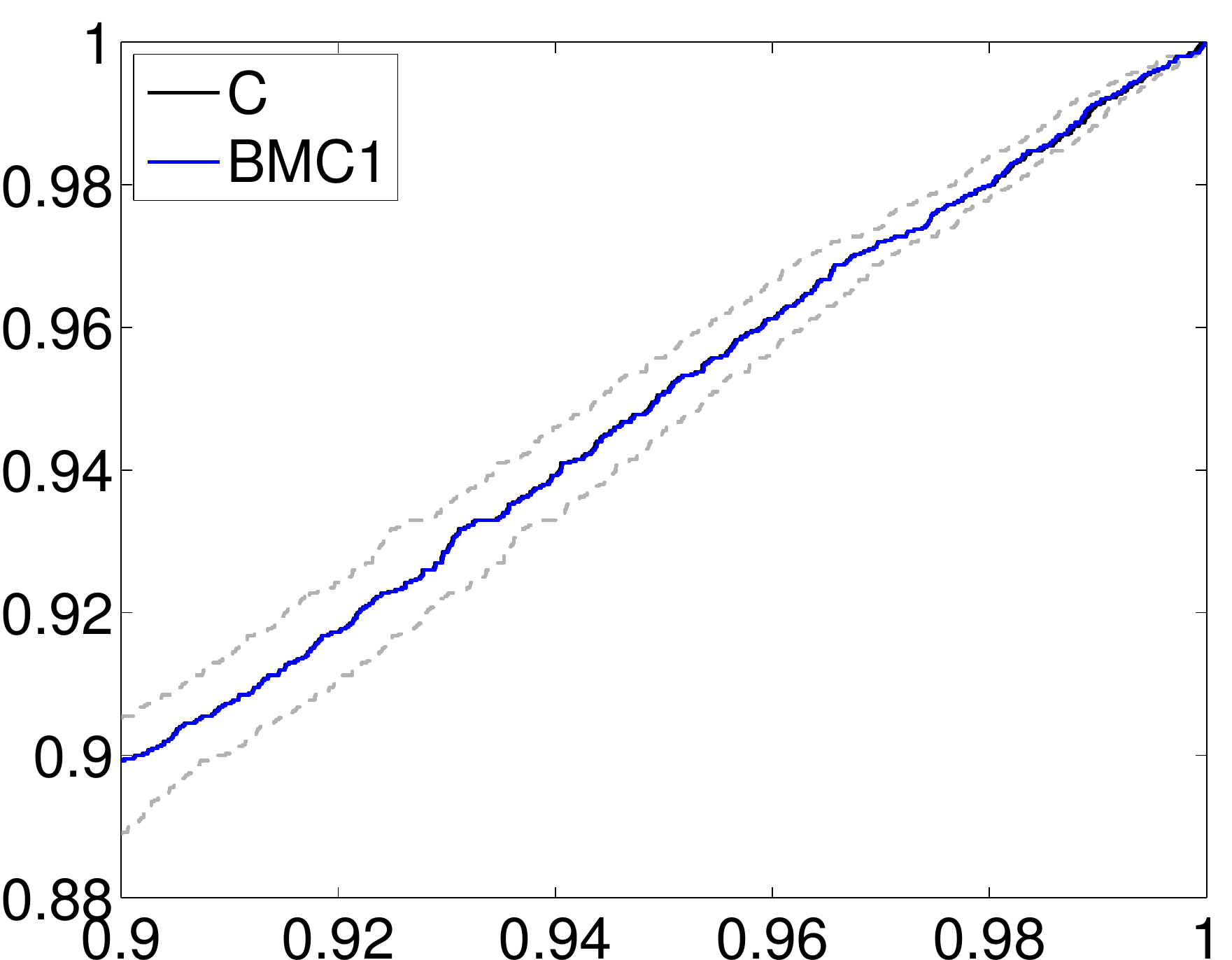}
\includegraphics[width=5cm]{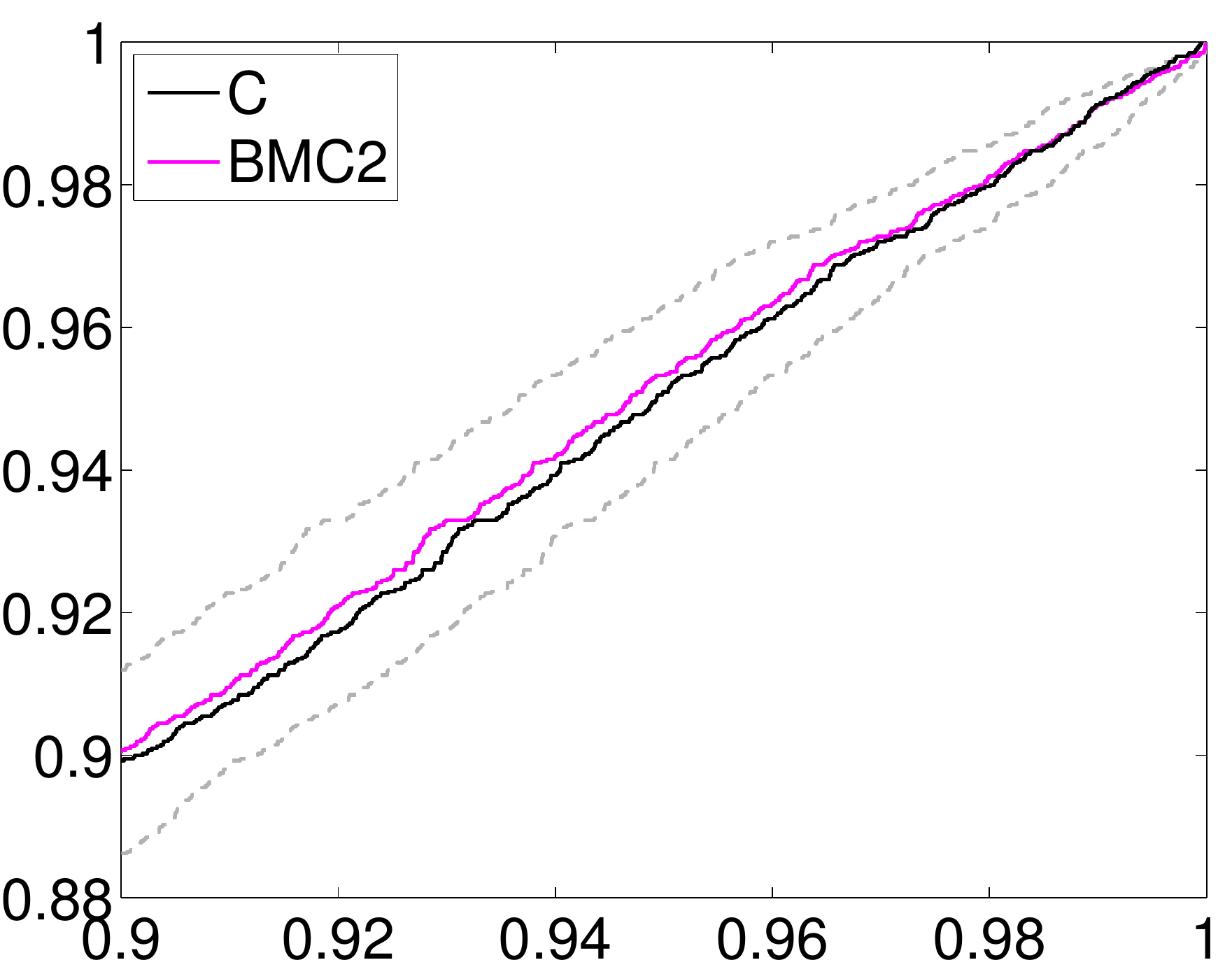}
\end{tabular}
\caption{Results of the Bayesian calibration model BC (top), BMC1
  (middle) and BMC2 (bottom) for the right tail of the predictive
  distribution. In each plot the PITs cdf of the calibrated (solid
  black line) and beta calibrated model (solid coloured line) and
  their 99\% HPD region (gray dashed
  lines)} \label{fig:mot1.heavy}
\end{center}
\end{figure}

\newpage

\section{Further real data results}\label{appendix_real}
\begin{figure}[!h]
\begin{center}
\includegraphics[width=6cm]{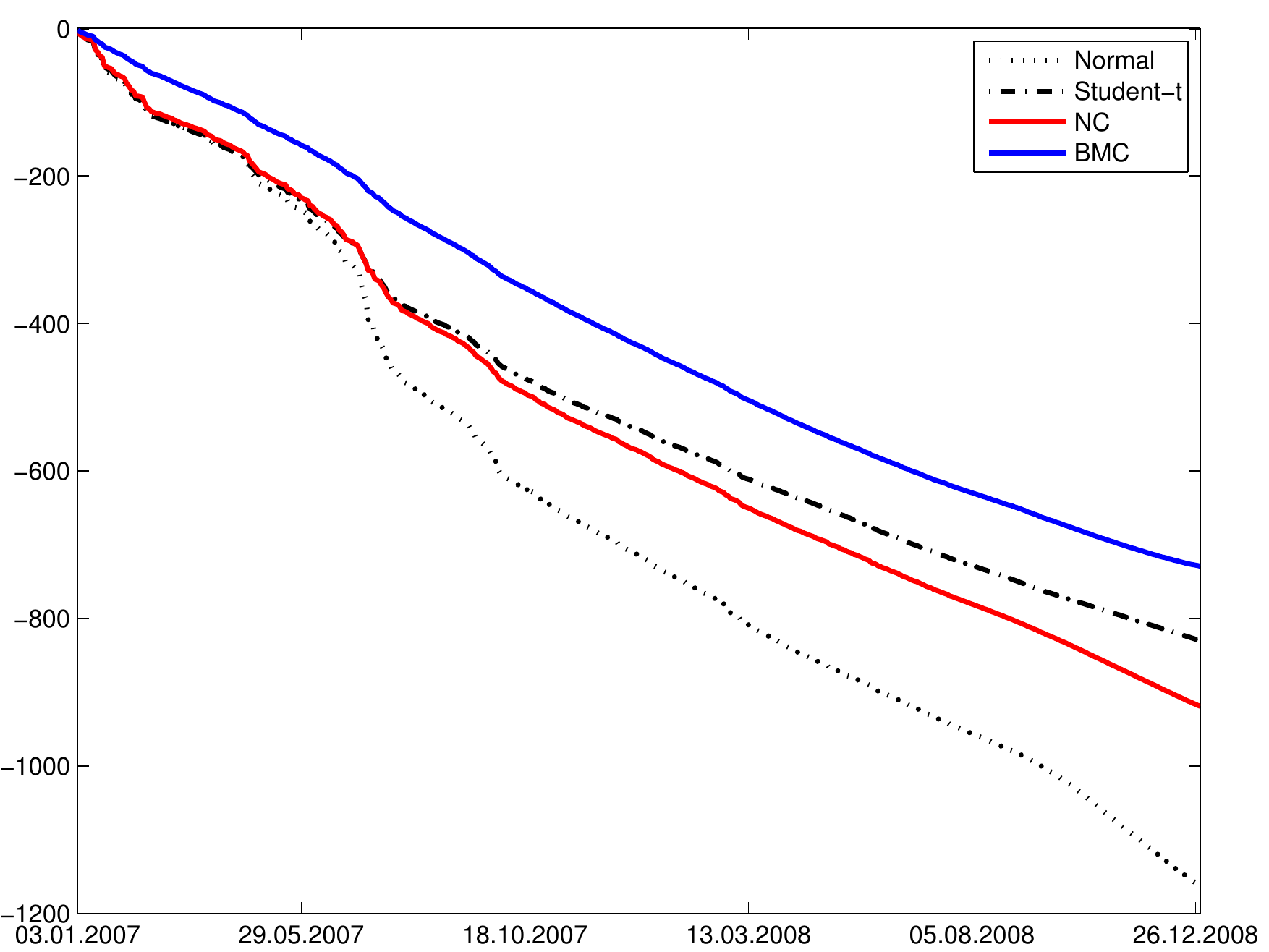}
\caption{Cumulative log scores for the Normal GARCH model (Normal),
  $t$-GARCH model (Student-$t$), linear pooling (NC) and beta
  mixture calibration (BMC) over the sample period from January 1,
  2007 to December 31, 2008.}  \label{fig:cum.LS}
\end{center}
\end{figure}

\begin{figure}[!h]
\begin{center}
\includegraphics[width=6cm]{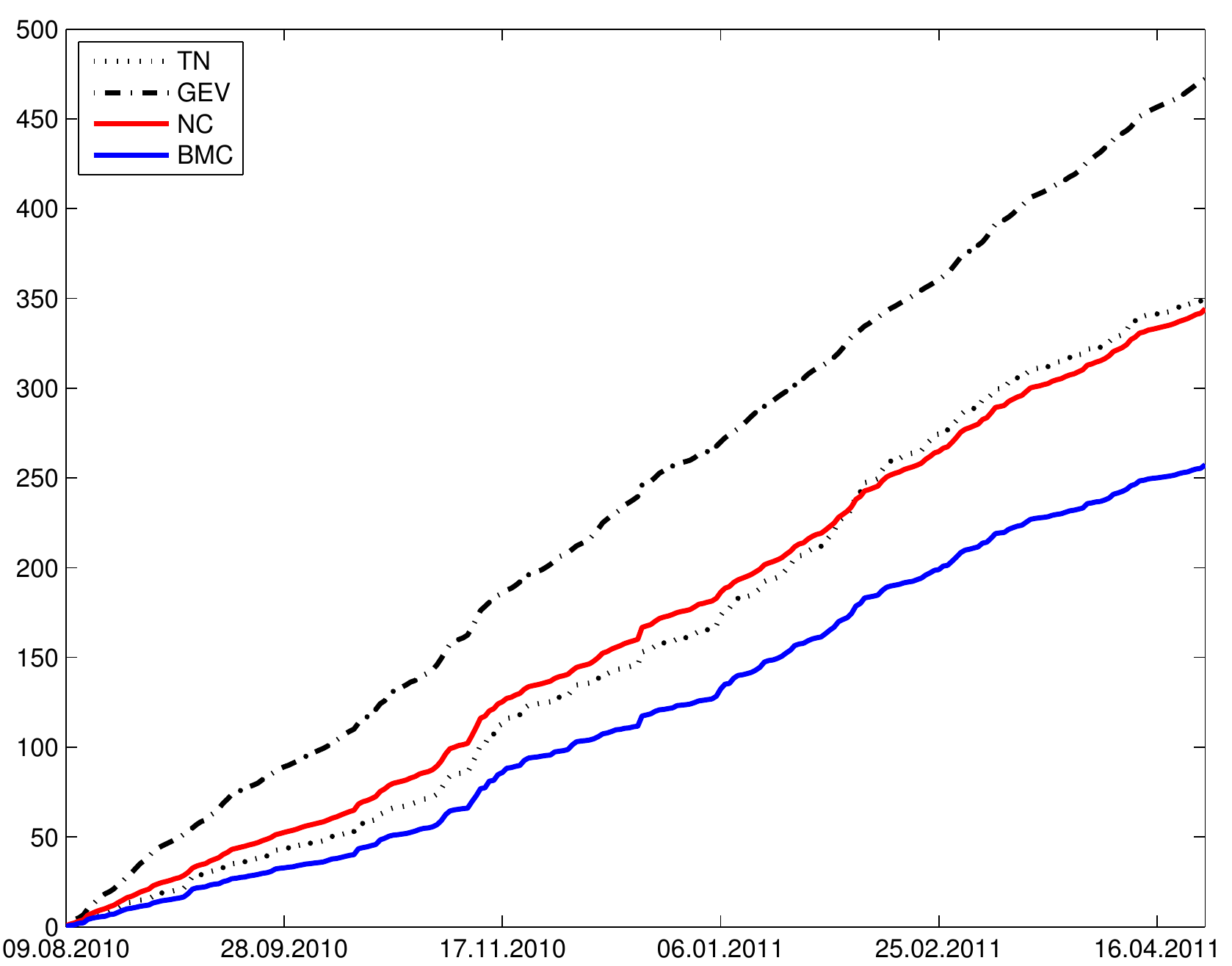}
\caption{Cumulative CRPS for the truncated normal (TN), the the
  generalized extreme value distribution (GEV), linear pooling (NC)
  and beta mixture calibration (BMC) over the sample period from
  August 9, 2010 to 30 April 2011.}  \label{fig:wind}
\end{center}
\end{figure}

\begin{figure}[p]
\begin{center}
\includegraphics[width=7.5cm]{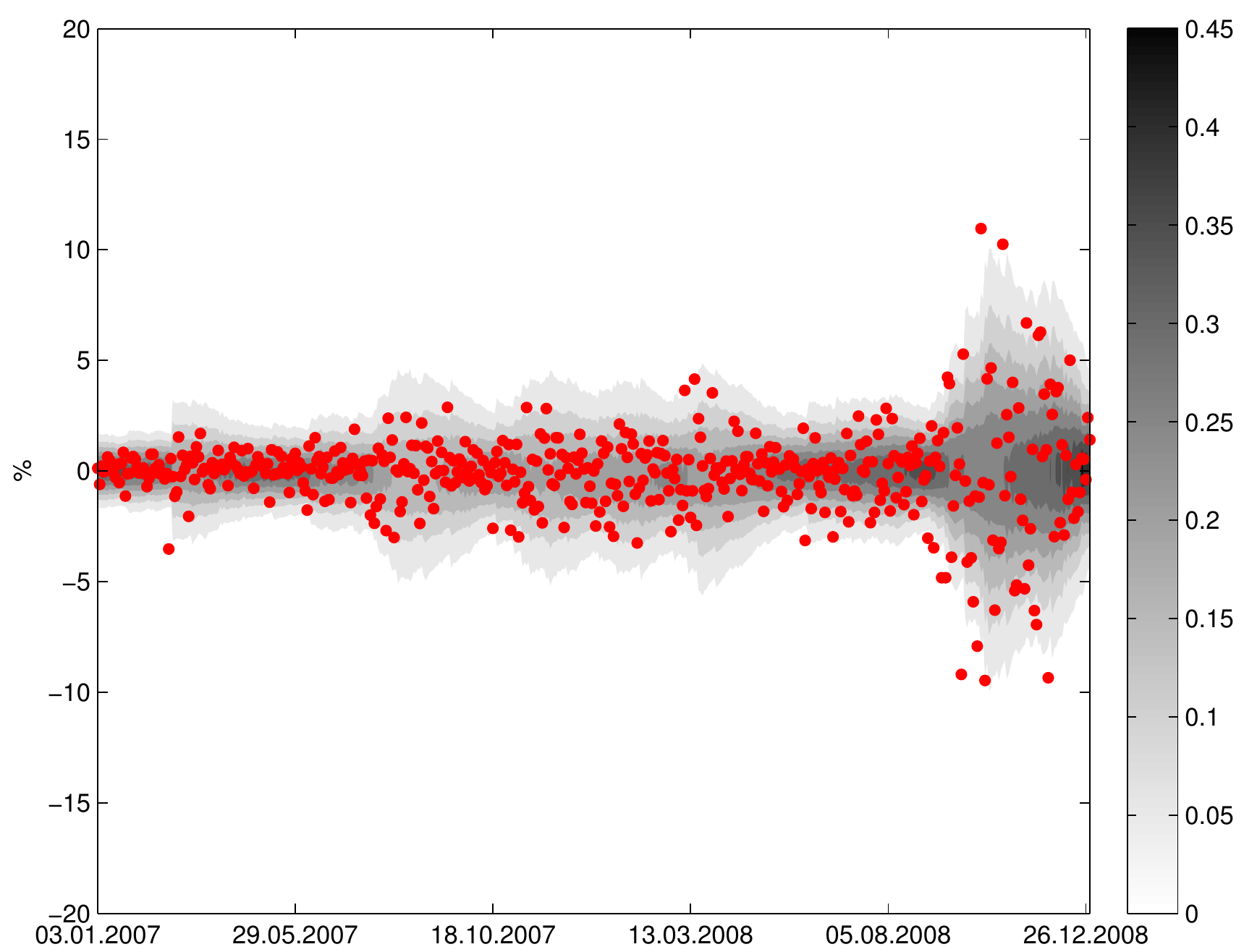}
\caption{Fanchart of the BMC model and observations (red points) over
  the sample period from January 1, 2007 to December 31,
  2008.}  \label{fig:fanchart}
\end{center}
\end{figure}

\begin{figure}[p]
\begin{center}
\includegraphics[width=7.5cm]{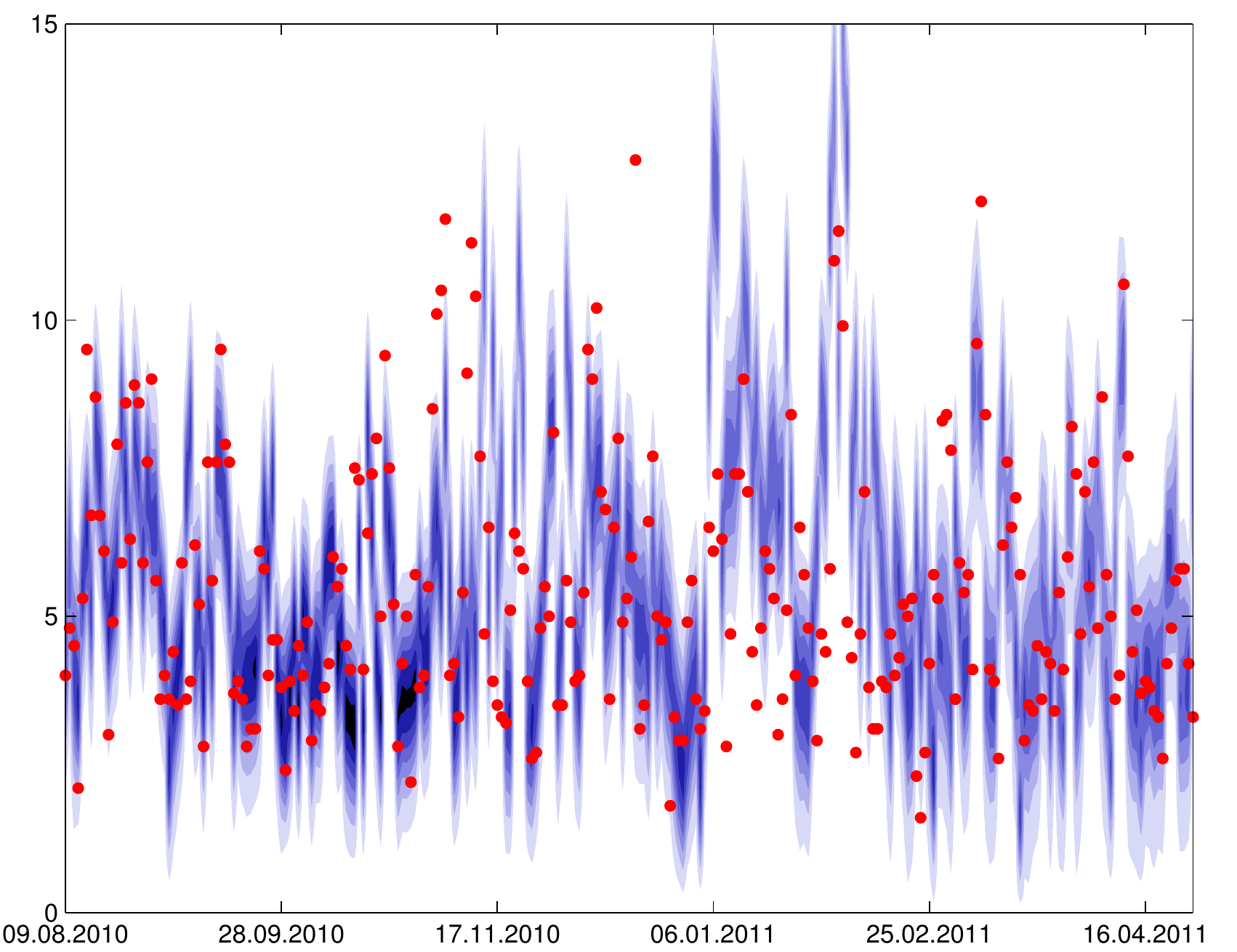}
\caption{Fanchart of the BMC model and observations (red points) over
  the sample period from August 9, 2010 to 30 April
  2011.}  \label{fig:fanchartwind}
\end{center}
\end{figure}

%\bibliographystyle{apalike}
%\bibliography{20142711}

\end{document}